\documentclass[english,12pt, reqno]{amsart}
\usepackage[T1]{fontenc}
\usepackage[latin9]{inputenc}
\usepackage{amsmath}
\usepackage{amsthm}
\usepackage{amssymb}

\usepackage{setspace}
\usepackage[thm]{macros}

\usepackage{tikz}
\usetikzlibrary{arrows.meta}

\makeatletter

\usepackage{hyperref}
\usepackage{changepage}
\usepackage{pdflscape}
\usepackage{afterpage}
\usepackage{dcolumn}
\usepackage{multirow}
\usepackage{booktabs}
\usepackage{threeparttable}
\newcolumntype{d}[1]{D{.}{.}{#1}}

\usepackage{comment}

\usepackage[all=normal, paragraphs=tight, wordspacing=tight, paragraphs=tight,
bibbreaks=tight,
floats=tight, bibnotes=tight, indent=tight]
{savetrees}

\newgeometry{margin=1.25in}

\newlength{\bibitemsep}
\newlength{\bibparskip}
\let\oldthebibliography\thebibliography
\renewcommand\thebibliography[1]{%
  \oldthebibliography{#1}%
  \setlength{\parskip}{\bibitemsep}%
  \setlength{\itemsep}{\bibparskip}%
}

\makeatother

\usepackage{babel}

\usepackage{titletoc}
\newcommand\DoToC{%
 \vskip1cm
 \startcontents[sections]
  \printcontents[sections]{l}{1}{\textbf{Contents}\vskip3pt\hrule\vskip5pt\setcounter{tocdepth}{2}}
  \vskip5pt\hrule\vskip5pt
}

\makeatletter

\renewcommand{\paragraph}{%
  \@startsection{paragraph}{4}%
  {\z@}{1.5ex \@plus 0.5ex \@minus .2ex}{-1em}%
  {\normalfont\normalsize\bfseries}%
}
\makeatother

\newcommand{\W}{\mathbb{W}}
\newcommand{\K}{\mathbb{K}}
\newcommand{\op}{\mathrm{op}}
\newcommand{\Lip}{\mathrm{Lip}} 
\newcommand{\Lin}{\mathrm{lin}} 
\newcommand{\cT}{\mathcal{T}}
\renewcommand{\S}{\mathbb{S}}
\newcommand{\td}{\tilde}
\newcommand{\lb}{\left(}
  \newcommand{\rb}{\right)}
\newcommand{\cE}{\mathcal{E}}

\begin{document}
\title{Nonparametric Identification of Demand \\ without Exogenous Product Characteristics \vspace{0.6cm}}
\author{Kirill Borusyak \and Jiafeng Chen \and Peter Hull \and Lihua Lei
\vspace{0.1cm}
\footnote{\textsuperscript{*}February 2026 (first draft: December 2025). Borusyak: University of California, Berkeley,
k.borusyak@berkeley.edu; Chen: Stanford University,
jiafeng@stanford.edu; Hull: Brown University, peter\_hull@brown.edu; Lei: Stanford
University, lihualei@stanford.edu. We thank Nikhil Agarwal, Isaiah Andrews,  Steve Berry, Jimbo Brand, Giovanni Compiani, Jos\'e Ignacio Cuesta, Matt Gentzkow, Jeff Gortmaker, Phil Haile, Nail Kashaev, Toru Kitagawa, Ashesh Rambachan, Adam Rosen, Jon Roth, Jesse Shapiro, Paulo Somaini, and Daniel Yi Xu for helpful comments. Google Gemini and OpenAI's ChatGPT contributed valuable insights. Refine.ink was used to check for consistency and
clarity.}}
\maketitle
\begin{abstract}
\singlespacing
We study identification of differentiated product demand from market-level data when product characteristics can be endogenous. Past work suggests nonparametric identification may be impossible: that is, in addition to standard price instruments,  exogenous characteristic-based instruments are essentially necessary to identify  sufficiently flexible demand models with standard index restrictions. We show, however, that price counterfactuals are nonparametrically identified using recentered instruments---which combine exogenous price instruments with possibly endogenous product characteristics---under a weaker index restriction and a new condition we term \emph{faithfulness}. We argue that faithfulness, like the usual completeness condition for nonparametric instrumental variable identification, is best viewed as a technical requirement on the strength of identifying variation rather than a substantive economic or statistical restriction. We show the two conditions are closely related, though generally distinct. We conclude with several practical implications for the parametric estimation of demand counterfactuals. 

\vspace{1em}


\noindent\thispagestyle{empty}
\end{abstract}

\noindent\newpage{}

\noindent
\global\long\def\expec#1{\mathbb{E}\left[#1\right]}%
\global\long\def\var#1{\mathrm{Var}\left[#1\right]}%
\global\long\def\cov#1{\mathrm{Cov}\left[#1\right]}%
\global\long\def\one{\mathbf{1}}%
\global\long\def\diag{\operatorname{diag}}%
\global\long\def\plim{\operatorname*{plim}}%

\pagestyle{plain}

\section{Introduction}

Conventional wisdom holds that exogenous supply-side instruments trace out demand curves and identify important  counterfactuals: how quantities would change under hypothetical price changes \citep{wright1928tariff}. Indeed, when considering demand for a single product, \cite{agrist2000interpretation} show such instruments are enough to nonparametrically identify certain average price elasticities across markets. Modern demand analyses, however, involve multiple differentiated goods and target counterfactuals that change prices in a \emph{particular} market \citep{berry2021foundations}. These analyses usually  incorporate other kinds of instruments: functions of the observed characteristics of competing products \citep[e.g.,][]  {berry1995auto,berry1999voluntary,gandhi2019measuring}. It is well-known that such characteristic-based instruments can be invalid if firm entry is strategic or  characteristics are otherwise endogenous  \citep{berry1995auto,petrin2022identification} and they can moreover reduce robustness to model misspecification \citep{andrews2025structural}.

Can other instruments be found to avoid these issues? Are exogenous shocks to prices enough to identify market-specific counterfactuals? In influential work, \citet{berryhaile} argue that the answer is generally no: i.e., that exogenous characteristic-based instruments are essentially necessary to flexibly estimate demand with market-level data. Specifically, they consider a nonparametric demand model of the form:
\begin{align}\label{eq:simpleBH}
    \sigma^{-1}(S,P)=X+\xi \equiv \delta,%
\end{align}
where $S$ and $P$ are $J$-vectors of the observed quantity shares and prices of $J$ goods in a market, $X$ is a $J$-vector of an observed product characteristic, and $\xi$ is an unobserved $J$-vector of demand shocks. The demand shocks and characteristics combine linearly in an index $\delta$. The unknown $\sigma^{-1}(\cdot)$ is the inverse of a demand function $\sigma(d,p)$, which returns the  shares that would arise if $\delta$ were set to $d$ and $P$ were set to $p$.\footnote{$\sigma(\cdot)$ may also depend on other observable characteristics $\tilde{X}$, which we suppress here for exposition.} 

\citet{berryhaile} argue that observing a $J$-vector of exogenous price instruments $Z$---or even having exogenous prices---generally does not identify $\sigma(\cdot)$ or even price counterfactuals, i.e. the effects of changing $P$ while holding $\delta$ fixed.\footnote{This is directly stated in a review, \citet[p.~6]{berry2021foundations}: ``having valid  instruments for all $J$ prices will not generally suffice for identification of  [\ldots] the \emph{ceteris paribus} effects of price changes.''} Intuitively, the  inverse demand function $\sigma^{-1}(\cdot)$  depends separately on $S$ and $P$ but $Z$ affects market shares only through prices; this suggests needing other instruments that shift $S$ holding $P$ fixed.\footnote{\citet{berryhaile} also consider a model in which $P$ enters the index and is excluded  from $\sigma^{-1}(\cdot)$,  concluding that price instruments suffice for  identification, though again assuming that characteristics are exogenous. \citet{borusyak2025estimating} show that recentered instruments identify price counterfactuals  in this case, allowing endogenous characteristics.}  The characteristic $X$  that enters the index uniquely fits this bill, provided it is also exogenous. All functions of $(X,Z)$ can  then serve as instruments and identify $\sigma(\cdot)$ so long as they induce enough  variation in $(S,P)$---a standard technical condition known as completeness \citep{newey2003instrumental}.

We reexamine this setting and reach a more optimistic conclusion: while exogenous price shocks are not generally enough for nonparametric identification of market-specific counterfactuals on their own, they \emph{can} suffice when combined with possibly \emph{endogenous} characteristics. This is because combinations of an as-good-as-randomly assigned $Z$ and an endogenous $X$ that are mean-zero given $X$---what \cite{borusyak2023nonrandom} call recentered instruments---are valid instruments with identifying power. In particular, they yield a simple test: for any candidate $\check\sigma(\cdot)$, if the implied index $\check \delta=\check \sigma^{-1}(S,P)$ correlates with a recentered instrument, then $\check \sigma(\cdot)$ cannot be the true demand function. 

We introduce a new technical condition---\emph{faithfulness}---under which all price counterfactuals are identified by candidate inverse demand functions that survive this recentered instrument test. Under faithfulness, any surviving $\check{\sigma}^{-1}(\cdot)$  equals an invertible transformation of the true $\sigma^{-1}(\cdot)$. The transformation is unknown, making demand not fully identified; we cannot, for example, learn counterfactuals that change the endogenous characteristics. But this is not a problem for price counterfactuals, as the unknown transformation is normalized away in those calculations.

Faithfulness requires the variation in $Z$ and $X$ to be rich enough to make all price effects on any functions of the form $H(\delta,P)$ detectable. More precisely, it says that functions of the form $H(\delta, P)$ cannot be mean-independent of the exogenous $Z$, given the possibly endogenous $X$, unless they are constant in $P$. This is a useful condition because any  candidate $\check\sigma^{-1}(\cdot)$ can be written as such a function: $\check\sigma^{-1} (S,P)=\check\sigma^{-1}(\sigma(\delta,P),P)\equiv\check{H}(\delta,P)$. Thus, if faithfulness holds, any $\check\sigma^{-1}(S,P)$ that survives the recentered instrument test (and is thus mean-independent of $Z$ given $X$) must be a transformation of the true $\sigma^{-1}(S,P)=\delta$: i.e., $\check H(\delta, P) = \check H(\delta)=\check H (\sigma^{-1}(S,P))$.

Intuitively, faithfulness is satisfied when two conditions hold: \emph{(i)} $Z$ is a strong instrument for $P$, given $X$, and \emph{(ii)} $X$ is a strong \emph{proxy} for $\delta$. Under \emph{(i)}, we have identification of all conditional-on-$X$ average price effects on any $\check H(\delta,P)$. This is not generally enough to make all price effects on $\check H(\delta,P)$ detectable, however, since the identified effects average over the unobserved $\delta\mid X$ distribution. Complex interactions between $\delta$ and $P$ could potentially average out to zero conditional on $X$, making a candidate $\check\sigma^{-1}(\cdot)$ survive the recentered instrument test despite not being a transformation of the true $\sigma^{-1}(\cdot)$. The proxy condition \emph{(ii)} prevents this by ensuring the variation in $X$ is sufficiently predictive of $\delta$ such that it is impossible for such interactions to average out across all the strata of $X$. Importantly, this argument does not require $X$ to causally vary the distribution of $\delta$: the strong proxy condition can hold even when all characteristics are endogenous.\footnote{Allowing characteristics to be endogenous further relaxes the index restriction in \cite{berryhaile}  by allowing $\delta$ to be an arbitrary function of $X$ and $\xi$ and for $\xi$ to have an arbitrary dimension.}

This intuition highlights the conceptually distinct roles of $Z$ (as an instrument) and $X$ (as a proxy) in our identification results. Unlike with  $Z$, researchers do not need to justify why the characteristics $X$ are as-good-as-randomly assigned or otherwise independent of unobserved demand shocks; indeed, faithfulness is even more plausible when firms strategically choose characteristics using partial information on $\delta$.  This distinction is also reflected in how $X$ and $Z$ would be used in forming ``technical instruments'' for estimation: while $X$ can be useful when combined with $Z$ and recentered, it has no identifying power on its own (i.e., in conventional ``BLP instruments'') as  recentered functions of $X$ only are identically zero.

While faithfulness is a novel condition, we argue it is a close cousin to completeness. Completeness says that 
functions of the form $\check\sigma^{-1}(S,P)$ cannot be mean-independent of $(X,Z)$ unless they are constant in $(S,P)$. Like completeness, faithfulness is therefore best viewed as a requirement on the strength of identifying variation rather than a substantive economic or statistical assumption. 

To explore the close relationship between faithfulness and completeness, we study when one follows from the other. In the limit where $Z$ is a perfect instrument for price (i.e., $P=Z$), we show faithfulness and completeness are equivalent. Outside this limit, faithfulness follows from completeness under several conditions on either pricing or the utility index. 
On the pricing side, it suffices that either \emph{(i)} $X$ and $Z$ combine in a nonparametric index $\lambda$ such that $P$ is independent of $(X, Z)$ given $(\lambda,\delta)$ or \emph{(ii)} the derivatives of $P$ with respect to $Z$ satisfy a particular separability condition. Both conditions are compatible with firms engaging in Bertrand--Nash pricing with certain forms of marginal costs.
On the $\delta$ index side, we show faithfulness follows from completeness when \emph{(i)} $\delta$ and $X$ have finite support with the same number of values or \emph{(ii)}
$\delta \mid X$ can be transformed to follow a location-scale model satisfying certain smoothness conditions---regardless of how prices are determined. 
In the other direction, completeness of $ (S,P) 
\mid (X, Z)$ follows from faithfulness when $\delta \mid X$ is complete.

Importantly, researchers do not need to take a stand on the specific sufficient conditions---our identification arguments are the same so long as faithfulness holds. 
Our menu of conditions shows faithfulness can hold without strong economic or statistical restrictions. Given completeness, faithfulness is not a big additional leap. 

Overall, our analysis clarifies the types of variation that are useful for identifying different demand counterfactuals. The model specifies potential outcomes in two ``treatments'': prices and characteristics \citep{chen2025reinterpreting}. If counterfactuals in both treatments are of interest, then it is unsurprising that exogenous variation in both $X$ and $P$ is needed---leading to the familiar ``$2J$'' instrument requirement in \citet{berryhaile}. If only the average causal effects of price were of interest, then $J$ instruments for price would intuitively suffice. Price counterfactuals in individual markets have an intermediate requirement: only average price effects are needed, but a sufficiently large and diverse set of them. This is satisfied by having $J$ strong instruments and $J$ strong proxies.

Our analysis also yields several practical insights for applied researchers estimating parametric demand models. Most importantly, it strengthens the recommendation---made previously in the parametric analyses of \citet{borusyak2025estimating} and \citet{andrews2025structural}---that researchers use recentered instruments in order to make their price counterfactual estimates more robust to characteristic endogeneity and model misspecification concerns. Our nonparametric results reassure practitioners that such robustness is not tied to any particular functional form or distributional assumptions. Our results also suggest a new role for potentially endogenous variation in $X$, as proxying for unobserved model heterogeneity, which practitioners can assess alongside the conceptually distinct requirement of price instrument exogeneity. 

This paper contributes to two main literatures. First, we add to an influential literature on the identification of differentiated product demand with market-level data started by \cite{berry1994} and \cite{berry1995auto,berry1999voluntary}; most closely related is the nonparametric identification analysis of \cite{berryhaile}. We depart from most of this literature by allowing product characteristics to be endogenous and by focusing on the identification of price counterfactuals. In this way, our market-level identification results complement Proposition 1 in \citet{borusyak2025estimating} which shows nonparametric identification of price counterfactuals without exogenous characteristics in a model with an index restriction on price and a standard completeness condition.  Our analysis also relates to papers like \cite{ackerbergcrawford} which study the estimation of parametric demand models with endogenous characteristics;  \cite{borusyak2025estimating} propose recentered instruments for this case.

From this literature, our focus on nonparametric identification of price counterfactuals aligns most closely with \citet{berryhaile24}, who call such counterfactuals ``conditional demand.'' Their analysis is also similar in allowing product characteristics $X$ to be endogenous and using price instruments that are exogenous conditional on $X$. The key difference is that \citet{berryhaile24} leverage  ``micro data'': i.e., variation in market shares across consumer characteristics. Our results show that price counterfactuals can be identified with market-level data only. This entails a very different identification strategy: while \cite{berryhaile24} essentially condition on $X$ and do not require variation on it, we use $X$ as a proxy for unobserved heterogeneity. 
 
Second, we contribute to a large literature studying nonparametric identification with
instrumental variables, including \cite{brownmatzkin}, \cite{npv03},
\cite{newey2003instrumental}, \cite{altonjimatzkin}, \cite{cik07}, \citet{chiappori2015nonparametric}, 
\cite{imbensnewey},
\citet{torgovitsky2015identification}, 
\citet{d2015identification}, and \citet{blundell2017individual}. 
The model we study imposes an index restriction in the spirit of \cite{berryhaile}, albeit a substantially weaker one. Thus, it is still restrictive relative to the fully-general potential outcomes model that \cite{agrist2000interpretation} consider for linear instrumental variables (IV) estimation. 
The index restriction is important for identifying market-specific counterfactuals, which are %
generally not given by the local average demand elasticities that linear IV identifies \citep{berry2021foundations,chen2025reinterpreting}. Beyond demand, our results can be restated to apply to triangular models in which unobserved heterogeneity enters the second stage through an index with a potentially endogenous covariate.\footnote{Specifically, our results apply to the model characterized by $Y = g(W_1, \delta(W_2, \xi))$ and $W_1 = h(Z, W_2, \omega)$ with $Z
\indep (\xi, \omega) \mid W_2 $ and $(Y, W_1, \delta, W_2) \in \R^J$. } 
Our faithfulness
condition appears new and may prove useful for identifying other nonparametric models
with this structure.\footnote{Our faithfulness condition derives its name from and relates
conceptually to a condition in the causal discovery literature
\citep{spirtes2000causation}, which considers whether a joint distribution of variables is
``faithful'' to an underlying causal graph;  \Cref{sub:faithfulness_in_causal_graphs} details this connection. \Cref{sub:triangular_model} shows that some
nonparametric identification results in \cite{imbensnewey}, \cite{torgovitsky2015identification}, and \cite{d2015identification} can be interpreted as verifying an analog of
faithfulness.}

The rest of this paper is organized as follows. \Cref{sec:theory} develops the main identification results. \Cref{sec:sufficient_conditions} relates faithfulness to completeness via a series of sufficient and necessary conditions.  \Cref{sec:practical} discusses key practical insights for parametric estimation. \Cref{sec:conclusion} concludes. For expositional ease we omit various technical details in the main text (e.g., regularity conditions on measurability, existence of moments, and support); \Cref{sec:proofs} gives precise theoretical statements and proofs.

\section{Theory}\label{sec:theory}
\subsection{Setup}
We consider a demand system in a market with $J$ products. The product quantity shares (or some other quantity measures) $S\in\mathbb{R}^{J}$ are given by:
\begin{align}
S=\sigma(\delta,\tilde{X},P),\quad \delta=\delta(\bar{X},\xi)\label{eq:demand}
\end{align} 
for prices
 $P\in\mathbb{R}^{J}$, observed characteristics $\bar{X}=(X, \tilde X)\in\mathbb{R}^
 {J}\times\mathcal{X}$, and unobserved demand shocks $\xi\in\Xi$ which capture
 latent consumer tastes or unobserved product characteristics. Focusing on identification, we suppress market subscripts.
 
 \cref{eq:demand} follows \citet
 {berryhaile} by imposing an index restriction in which the
 demand shocks and characteristics enter together through the unobserved $\delta\in \mathbb
 {R}^J$ with the ``special characteristic'' $X$ otherwise excluded from $\sigma$.\footnote{Note that the meaning of the $\delta$ index differs from that in the literature on parametric demand models following \cite{berry1995auto} where it usually denotes the mean utility vector over products. Most importantly, here $\delta$ does not include any effects of price on demand.}   We also follow \cite{berryhaile} by assuming demand is invertible in this index: 

\smallskip
\begin{as}
\label{as:invert}
  For every $(\tilde{x},p)$, the map $d\mapsto \sigma(d, \tilde{x},p)$ is invertible. 
\end{as}
\smallskip

\noindent Invertibility follows here from a ``connected substitutes'' condition \citep{berry2013connected}. 
 
 We are interested in identifying price counterfactuals: the quantities $\sigma(\delta, \tilde X, p^\prime)$ that would be observed if prices were set to some counterfactual $p^\prime$ while holding $\delta,\tilde{X}$ fixed at their realized values. These capture both local changes in prices (i.e. own- and cross-price elasticities) and global counterfactuals that can inform many positive and normative analyses---such as merger simulations, conduct testing, and studying the effects of product entry and exit.\footnote{Under standard assumptions on price-setting, firm optimization only involves demand at counterfactual prices---such that our counterfactuals are enough for computing marginal costs, markups, and welfare changes from mergers or entry/exit (modeled by prices moving from/to infinity). \citet{borusyak2025estimating} discuss how $\delta$ can be obtained for new products if characteristics are endogenous. %
 }  It is straightforward to extend our analysis to counterfactuals in other product attributes by interpreting $P$ as those non-price variables of interest.
 
 To identify price counterfactuals, we assume that a set of price instruments  $Z\in \R^{d_z}$, $d_z \ge J$, is observed in addition to $S$, $\bar{X}$ and $P$. We formalize the sense in which these instruments are exogenous below. The identification challenge stems from the unobservability of $\xi$ and the unknown nonparametric functions $\sigma$ and $\delta$.

\cite{berryhaile} consider identification of $\sigma$ with the additional assumptions of
 \emph{(i)} exclusion of $\tilde X$ from the utility index, $\delta(\bar X, \xi) = \delta(X ,\xi)$, \emph{(ii)} a linear index, $\delta (X, \xi) = X +\xi$ with $\xi\in\mathbb{R}^J$, and \emph{(iii)} mean-independence of the
 unobserved demand shocks from the characteristics and instruments: $\E
 [\xi \mid \bar{X},Z] = 0$. The last assumption captures the exogeneity of both $\bar{X}$ and $Z$. Together, \emph{(i)}--\emph{(iii)} imply:
 \begin{align}\label{eq:linear_exogenous}
  \E[\delta \mid \bar{X},Z] = X.
\end{align}

\cref{eq:linear_exogenous} identifies $\sigma$ under a completeness condition \citep{newey2003instrumental,ai2003efficient,andrews2011examples,d2011completeness,miao2018identifying}:
\smallskip
\begin{as}
  \label{as:completeness}
The distribution of $(S, P, \tilde X) \mid (X,Z, \tilde X)$ is complete: For all $h$ with
$\E[\norm{h(S,P, \tilde X)}] < \infty$, if $
  \E[h(S,P, \tilde X) \mid X,Z,\tilde X] = 0$ then $ h(S,P,\tilde X) = 0$. Equivalently, under \cref{as:invert}, the distribution of $(\delta, P, \tilde X) \mid (X, Z, \tilde X)$ is complete.
\end{as}
\smallskip
\noindent Identification follows from the fact that, under \cref{as:invert,eq:linear_exogenous}, 
\[ 
X - \E[\sigma^{-1}(S, P, \tilde X) \mid \bar{X},Z] = 0,
\numberthis \label{eq:bh_integral_eqn}
\] while \cref{as:completeness} ensures the solution to this integral equation in $\sigma^{-1}$ is
 unique. 
 
  \cite{berryhaile}'s identification result leads to a standard intuition on the need for exogenous variation in the special characteristic $X$. Assumptions \emph{(i)}--\emph{(iii)} yield a structural equation, $X = \sigma^{-1}(S, P, \tilde X) - \xi$, with $2J$ endogenous variables $(S,P)$ on the right-hand side after conditioning on $\tilde X$. This suggests a need for $2J$ exogenous instruments: $J$ ``for the prices'' $P$ and an additional $J$ ``for the shares'' $S$. When $X$ is exogenous, it can serve as the latter set of instruments since it is excluded from the right-hand side. The role of exogenous characteristics, then, is not to generate
 variation in prices but to generate variation in shares conditional on prices so that one can disentangle
 the two arguments of $\sigma^{-1}$. 
 The completeness condition ensures that transformations of $X$ and $Z$ provide sufficient variation for this. 

For price counterfactuals, we work under a partial relaxation of \cref{eq:linear_exogenous} that only imposes conditional exogeneity of the price instruments:

 \smallskip
 \begin{as}
   \label{as:exogenous}
    $Z \indep \delta \mid \bar X.$
 \end{as}

 \smallskip
\noindent This restriction strengthens the mean independence in \eqref{eq:linear_exogenous} to  full independence of $Z$ given $\bar{X}$ but crucially drops the exogeneity of $\bar{X}$.\footnote{\Cref{as:exogenous} is implied by the more standard $Z \indep \xi \mid \bar X$ but is weaker when $\xi$ has a higher dimension than $\delta$.
We do not view the strengthening from mean- to full independence as meaningfully restrictive, since there is little reason to distinguish $\xi$ from some transformation $g(\xi)$. Full independence amounts to mean-independence of any $g(\xi)$, treating all functions symmetrically.
}
\Cref{as:exogenous} can be motivated by viewing $Z$ as a
set of supply-side shocks drawn in some true or natural experiment after product
characteristics are determined \citep{borusyak2025estimating}.\footnote{For example, in
the U.S. automobile market, $Z$ may contain exchange rate shocks in cars'
countries of production that are excluded from demand and drawn randomly after cars are designed but before prices
are set. In this case $Z\indep (\xi,\bar X)$, implying \cref{as:exogenous}. See \citet{ackerbergcrawford} and \citet{borusyak2025estimating} for further discussion; 
Appendix A of \citet{berryhaile24} shows how \cref{as:exogenous} can arise for different types of price instruments in various economic models.}

Our relaxation of \Cref{eq:linear_exogenous} is motivated by three related misspecification
concerns. First, suppose the functional form $\delta (\bar X, \xi) = X + \xi$ is correctly
specified but characteristics $\bar X$ are chosen strategically by firms with partial
knowledge of demand conditions as captured by $\xi$. Then $\E [\xi\mid\bar{X}]=0$ is
unlikely to hold. Second, $\xi$ may include physical product characteristics chosen by firms along
with $\bar{X}$ but unobserved by the econometrician. Again, there is little reason they
would be uncorrelated with $\bar{X}$. 
Finally, suppose that characteristics $\bar{X}$ are fully independent of $\xi$ but that the functional form of $\delta$ is not linear or does not exclude $\tilde X$: i.e., $\delta(\bar X, \xi) \neq X+\xi$.
Then \eqref
{eq:linear_exogenous} need not hold, and model-implied price counterfactuals relying on
\eqref{eq:linear_exogenous}  may be incorrect  \citep{andrews2025structural}. This in particular accommodates the view of $\xi$ as a ``statistical disturbance'' rather than an economically meaningful object, in which case $\bar{X}\indep \xi$ by definition of $\xi$ but  exclusion and linearity are restrictive.\footnote{Any two random vectors $(\delta, \bar X)$ can be represented as $\delta = \tilde\delta(\bar X, \tilde\xi)$ for some $\tilde\xi \indep \bar X$ and some measurable function $\tilde\delta$---for instance by using Knothe--Rosenblatt (conditional quantile) transforms. In such a representation, $\tilde\xi$ need not have a straightforward economic interpretation.
} Relaxing the
linear index functional form substantively weakens restrictions on the heterogeneity of causal
effects of $\bar X$ on $S$ \citep{chen2025reinterpreting}; it also allows
$\xi$ to be multi-dimensional for each product, in contrast to most demand models in the literature.\footnote{For instance, Appendix B of \citet{berryhaile} contains a nonseparable model in which $\delta = \delta(X, \xi)$, but restricts the map coordinate-wise: $\delta_j = \delta_j(X_j,\xi_j)$ with the function increasing in $\xi_j$. This requires $\xi$ to be of dimension $J$, in addition to other restrictions that we remove.}

\subsection{Identification strategy}
Under \cref{as:exogenous}, 
\begin{align}
  \E[\delta \mid  \bar{X},Z] = \E[\delta \mid \bar
   {X}] \equiv k_0(\bar{X}) \numberthis \label{eq:moment_condition_with_z}
\end{align}
for some unknown $k_0$.
This motivates our identification strategy. The true inverse demand function, which returns $\delta$, is mean-independent of $Z$ given  $\bar{X}$:
\begin{equation}\label{eq:cond_moment} \E[\sigma^{-1}(S,P,\tilde{X}) \mid \bar{X},Z] = \E[\sigma^{-1}(S,P,\tilde{X}) \mid \bar{X}].\end{equation}
We can thus rule out some candidate inverse demand functions $\check\sigma^{-1}(S,P,\tilde X)$: those  that violate the
conditional moment restriction \eqref{eq:cond_moment}.\footnote{\cref{as:exogenous} implies conditional independence of higher moments of $\sigma^{-1}(S,P,\tilde X)$ too, which could in principle be used for identification. We focus on identification via first moments, which is more aligned with \cite{berryhaile} as well as common parametric estimation strategies.} 

To develop this strategy, we first ease some notation. Our identification results fully condition on the ``non-special'' characteristics $\tilde X$. To simplify notation, we suppress $\tilde X$ and replace $\bar{X}$ with $X$ throughout.

We next define a (potentially non-sharp) identified set for $\sigma^{-1}$:
\begin{align*}
  \Theta_I \equiv 
 \br{h(s,p), \text{ $\R^J$-valued and invertible in $s$}\colon \E
 [h(S,P) \mid X,
 Z] = k(X) \text{ for some $k$}}. 
\end{align*}
This set is nonempty because \eqref{eq:cond_moment} implies it contains at least the true $\sigma^{-1}$.

The identified set can be more constructively defined via recentered instruments,  which take the form $R( X,Z) - \E[R( X,Z) \mid  X]$ for some function $R$ \citep{borusyak2023nonrandom}. The set of functions in $\Theta_I$ is exactly the set of functions that are orthogonal to any recentered instrument:
\begin{lemma}
    Let $h$ be $\mathbb{R}^J$-valued, invertible, and square-integrable. Then $ h\in \Theta_I$ if and only if $\E[h(S, P) \br{R( X, Z) - \E[R( X,Z) \mid X]} ] = 0$ for all square-integrable $R$.
\end{lemma}
\noindent Intuitively, a candidate $\check\sigma^{-1}$ can only be conditionally mean-independent of $Z$, and thus in $\Theta_I$, if it is uncorrelated with all conditionally mean-zero functions of $Z$.

The identified set is not a singleton that contains $\sigma^{-1}$ only: under \cref{as:exogenous}, all transformations of the true $\sigma^{-1}$ are independent of
 $Z$ given $X$ and hence in $\Theta_I$.
However, this does not necessarily impede point-identification of price counterfactuals: in fact, it suffices for $\sigma^{-1}$ to be identified only up to some transformation.\footnote{Statements like \cref{lemma:transform} are also intermediate steps in \citet[][p. 950]{matzkin2008identification}, \citet[][p. 1190]{torgovitsky2015identification}, and \citet[][p. 1143-1144]{berryhaile24}. In particular, \citet{berryhaile24} normalize their index to avoid considering analogous indeterminacy as with $T$.}
 \smallskip
\begin{restatable}{lemma}
{lemmatransform}
\label{lemma:transform}
Suppose all $\check{\sigma}^{-1}
  \in \Theta_I$ have the form $\check{\sigma}^{-1}(S, P) = T(\sigma^
  {-1}(S,P))$ for some $T : \R^J \to \R^J$. Then  price counterfactuals are identified.

\end{restatable}
\smallskip
\noindent Intuitively, counterfactual quantities given hypothetical prices $p^\prime$ under a  candidate model $\check{\sigma}^{-1}=T(\sigma^{-1})$ are computed from the observed $(S,P)$ as:
\begin{align*}
 \check{\sigma}(\check{\sigma}^{-1}(S,P),p^\prime)=\sigma(T^{-1}(T(\sigma^{-1}(S,P))),p^\prime)=\sigma(\delta,p^\prime).
\end{align*}
The transformation $T$ cancels, so it does not affect the counterfactual $\sigma(\delta,p^\prime)$.\footnote{Note that $T$ is implicitly invertible because $\Theta_I$ only includes invertible $\check\sigma^{-1}$ candidates.} 
 
Our identification strategy is therefore successful when $\Theta_I$ contains nothing but transformations of $\sigma^{-1}$. We next consider when this condition holds.

\subsection{A surprisingly insightful case: exogenous prices}
We start  with the case of $Z=P$: i.e., where prices are themselves exogenous ($P\indep \delta\mid X$). This is  unrealistic in many settings, but it is an illuminating baseline case for developing our strategy (and notably one which \citet{berry2021foundations} call ``surprisingly difficult''). In particular, it clarifies why exogenous variation in $P$ can be enough and what the endogenous variation in $X$ nevertheless contributes. This case can also be viewed as the limit of our general setting with perfect price instruments. 

With exogenous prices, price counterfactuals are identified under the same completeness condition as in \citet{berryhaile}---despite our substantial relaxation of their model.
Indeed, all characteristics can be fully endogenous in this case:
\begin{prop}
\label{prop:exogenous_price}
Suppose $P = Z$ and \cref{as:completeness,as:exogenous,as:invert}
 hold.
Then price counterfactuals are identified.
\end{prop}
\noindent To see this result, consider any candidate $\check\sigma^{-1}\in\Theta_I$. We can write $\check\sigma^{-1}(S,P)=\check\sigma^{-1}(\sigma(\delta,P),P)\equiv H(\delta,P)$ for an unknown $H$ such that $\E[H(\delta,P) \mid P, X] =
k (X)$. Fixing some price value $p_0$,\footnote{When $P$ is continuously distributed, conditioning on a realization of it requires measure-theoretic care. See \cref{prop:exogenous_price_x_first} and \cref{prop:p-index_x} for a formal argument.} we then have:
\begin{align*} 
0 &= \E[H(\delta,P)\mid P,X]-\E[H(\delta,p_0)\mid P=p_0,X]\\
&= \E[H(\delta,P)\mid P,X] - \E[H(\delta,p_0)\mid P,X]\\
\implies 0 &= H(\delta, p) - H(\delta, p_0) 
\end{align*}
where the first line uses the fact that $P$ does not enter $k(X)$, 
the second line uses \cref{as:exogenous}, and the third line uses \cref{as:completeness}. Thus $H(\delta, p) = H(\delta, p_0) $ does not depend on $p$: $H(\delta,P)=H(\delta)$. In turn, this implies that the candidate inverse demand function is a transformation of the
 true inverse demand function: \[\check\sigma^{-1}(S,P)=H(\delta,P)=H(\delta)=H(\sigma^{-1}(S,P)).\] Since this holds for any $\check\sigma^{-1}\in\Theta_I$, price counterfactuals are identified
 by \Cref{lemma:transform}.

 \cref{prop:exogenous_price} contains two key intuitions. First, for identifying price counterfactuals, it is enough to detect causal effects of prices only. Indeed, since it suffices  to identify $\sigma^{-1}(S,P)$ only up to a transformation (\cref{lemma:transform}), we only need to find some function $h(S,P)$ such that $h(\sigma(\delta,P),P)\equiv H(\delta,P)$ is a transformation of $\delta$ only with no dependence on $P$.  From a causal inference perspective, this means finding some outcome $H=h(S,P)$ which is fully unaffected by $P$. 
 
 Importantly, here we are only interested in whether $P$ has \emph{any} causal effects on $H$ and not in disentangling the effects that come through $S$ versus $P$. We thus sidestep the need for separate exogenous variation in $S$ given $P$, which is what compels characteristic exogeneity in \citet{berryhaile}. Hence, to trace out \emph{average} price effects, we require only $J$-dimensional exogenous price variation---not $2J$ instruments.

 The second key intuition is on the role of the additional $J$-dimensional variation in $X$ as a \emph{proxy} for $\delta$.\footnote{A related argument by \citet{ahn2025prognostic}, in a context with an exogenous binary treatment $P$,  uses $X$ in a similar way and calls it a ``prognostic variable.''} Tracing out average effects of $P$ is not generally enough: while \eqref{eq:demand} implies that the effect of $P$ on $S$ (and therefore on $H$) is fully governed by the index $\delta$, this index is unobserved. Zero average price effects, which average across the unobserved $\delta$, do not by itself imply
zero price effects for a given $\delta$.\footnote{Borrowing an example from \citet{benkard2006nonparametric}, if $X$ is independent of $\delta$ and thus variation in $X$ is not useful, then one can choose $\delta \sim \Norm(0, I_J)$ and $H(\delta, p)$ to be a $p$-dependent rotation of $\delta.$ This construction implies that $H(\delta, P) \indep (X, P)$, and thus, absent other conditions that rule out such $H$, the set $\Theta_I$  would include inverse demand candidates that are not just transformations of $\delta$---leading to incorrect price counterfactuals.}  This gap can be filled by leveraging variation in $X$, so long as it strongly correlates with $\delta$---i.e., ``proxies'' for it. For each observed value of $X$, exogenous price variation reveals the conditional average price effects on $H$.  
Completeness of $(\delta, P) \mid (X, Z)$---which here, with $Z=P$ and $\delta\indep P\mid X$, is equivalent to completeness of $\delta \mid X$---ensures that this set of identified effects is rich enough to span the set of price effects for each unobserved value of $\delta$. Hence, under completeness, the finding of no conditional-on-$X$ average price effects ensures no causal dependence of $H$ on $P$.\footnote{Section 3 in \citet{chen2025reinterpreting} shows that many counterfactual predictions from structural models analogously extrapolate from a lack of average causal effects to a total lack of causal effects.} 
 
This proxy role of $X$ is substantively different from the role that $X$ plays in \citet{berryhaile}, as an instrument ``for the shares.'' Here, $X$ need not be exogenous with respect to the economically meaningful unobserved demand shock $\xi$, nor must it causally shift $\delta$. For instance, it is perfectly fine if $X$ correlates with $\delta$ because of selection only. As a result, while practitioners should justify why $Z$ is as-good-as-randomly assigned and why it is not strategically chosen (in order to satisfy \cref{as:exogenous}), no such arguments are needed for $X$. 

Of course, one can always make $X$ ``technically exogenous'' by rewriting \eqref{eq:cond_moment} as:
\begin{equation*}
0 = \sigma^{-1}(S, P) - k_0(X) + \tilde\xi, \qquad \E[\tilde\xi \mid X, P] = 0.
\end{equation*}
But since $k_0$ is unknown, $X$ is not excluded from this model: variation in it informs both $k_0$ directly and $\sigma^{-1}$ indirectly through $S$. This  highlights that we cannot merely redefine residuals, make $X$ exogenous without loss, and appeal to the argument of \citet{berryhaile}; fundamentally, our results rely on a distinct use of $X$. 

The proxy role of $X$ is also substantively different from the role of other characteristics $\tilde{X}$, which---being fully conditioned on---might be viewed here as ``controls.'' While there are no restrictions on the amount or nature of variation in $\tilde{X}$ or its effects on demand, $X$ must be strongly correlated with $\delta$ and excluded from demand in the sense of the index restriction in \eqref{eq:demand}. Indeed, the proxy role of $X$ is what gives the index restriction bite: with a constant or completely randomly generated $X$, any demand model could be written in the form of \eqref{eq:demand}. 

 We next generalize these results and intuitions to settings with endogenous prices.

\subsection{Faithfulness and identification with endogenous prices}

When $P\neq Z$, we continue to need sufficiently rich  variation in $X$ to proxy for the unobserved $\delta$. We additionally need rich enough variation in the exogenous instruments $Z$ to trace out the causal effects of the now-endogenous prices. Perhaps surprisingly, the usual completeness condition is not the appropriate formalization of these two requirements. Instead, we consider a new condition to replace \cref{as:completeness}:

\begin{as}[\textbf{Faithfulness}]
  \label{as:faithfulness} The distribution of $(\delta, P) \mid(X,Z)$ is \emph
   {faithful}: for all $\R^J$-valued $H(\delta,P)$, if $
  \E[H(\delta,P) \mid X,Z] = k(X)$ does not depend on $Z$ then $H$ does not depend on $P$, i.e., $H(\delta,P) = H(\delta). 
$
\end{as}

Under faithfulness, price counterfactuals are identified even with endogenous prices:\begin{restatable}{prop}{lemmamain}
\label{lemma:main}
Under \Cref{as:invert,as:exogenous,as:faithfulness}, price counterfactuals are identified.
\end{restatable}
\noindent  The proof of \cref{lemma:main}  follows the same steps as for \cref{prop:exogenous_price}: write $\check\sigma^
{-1} (S,P)=\check\sigma^{-1}(\sigma(\delta,P),P) \equiv H(\delta,P)$ and note
that
$\E[H(\delta,P)\mid X,Z]=\E[\check\sigma^{-1}(S,P)\mid
    X,Z]=k(X)$ for some $k$ if $\check\sigma^{-1}\in\Theta_I$.
Under faithfulness, this means that $\check\sigma^{-1}(S,P)=H
 (\delta,P)=H(\delta)=H(\sigma^{-1}(S,P))$ is a transformation of the
 true inverse demand function. All price counterfactuals are thus identified
 by \Cref{lemma:transform}.  

The two key intuitions from the exogenous price case are retained in \cref{lemma:main}. First, because we are only interested in whether $P$ has any effects on $H(\delta,P)=h(S,P)$, and not in disentangling its direct effects from indirect effects through $S$, we avoid the need for separate exogenous variation in $S$ given $P$.  Second, though variation in $X$ need not be exogenous, it is still critical for proxying for the unobserved $\delta$. As in the exogenous price case, complex interactions between $\delta$ and $P$ can potentially average out and threaten identification by obscuring true price effects on $H(\delta,P)$. Here faithfulness, rather than completeness, rules out such scenarios via rich variation in $X$ that strongly predicts $\delta$, together with strong price instruments $Z$.\footnote{\Cref{sec:no_x} shows how, under additional restrictions, analogs of faithfulness can hold without an $X$ that satisfies the index restriction and proxies for $\delta$. These include the case where $P$ combines linearly with $\delta$ (as in Proposition 1 of \cite{borusyak2025estimating}), and where $J=1$ with $\sigma$ monotone in $\delta$ (as in \cite{imbensnewey}, \cite{torgovitsky2015identification}, and \cite{d2015identification}).}

How different is faithfulness from the usual completeness condition? They are similar in kind: both are high-level conditions that yield identification  by directly asserting that certain operations are injective. Namely, they link the variation in some $H(\delta,P)$ that is detectable via conditional expectations in $(X,Z)$ to the structural dependence of $H$ on $(\delta,P)$. Completeness requires that if $H(\delta, P)$ varies then $\E[H(\delta,P) \mid X, Z]$ also varies. In this sense, conditional expectations ``faithfully reflect''   changes in $(\delta,P)$. Faithfulness instead says that if prices truly affect a function of demand primitives, these effects must show up in the instrument-induced price variation conditional on 
  $X$. Conversely, functions with no instrument-induced conditional-on-$X$ variation cannot depend on  $P$.\footnote{Note that faithfulness is different than conditional-on-$X$ completeness---i.e., completeness of $(\delta,P,X)\mid(X,Z)$. This condition cannot hold here as $Z$ generates no variation in $\delta$.}

In the nonparametric instrumental variables literature, 
completeness is often treated as a technical condition which encodes a sense of instrument strength while not imposing substantive restrictions \citep[see, e.g.,][]{newey2003instrumental,ai2003efficient,darolles2011nonparametric}.\footnote{This is especially true when nonparametric identification is viewed as a theoretical argument for how particular parametric structure does not ``drive'' conclusions. In any given parametric model, one could posit a strong-instrument-type condition where all non-constant functions $w \in \mathcal{W}$ of endogenous variables $w(S, P)$ are assumed to correlate with some function of $(X,Z)$, over some parametrized class $\mathcal{W}$. Completeness is the limit of such conditions as we enlarge $\mathcal{W}$ to include all (integrable) functions. } Lower-level conditions are given by, e.g., \citet{d2011completeness} and \citet{andrews2011examples}. We argue that faithfulness should be similarly treated as a technical condition which encodes a sense of instrument strength for $Z$ as well as proxy strength for $X$. 

To bolster this argument, we next detail connections between the two conditions---showing, in particular, that faithfulness follows from completeness under a wide range of different lower-level conditions. Together these results suggest faithfulness and completeness, while non-nested in general, are close cousins. Thus, to the extent nonparametric identification under completeness reassures practitioners that demand can be flexibly estimated with exogenous characteristics and price instruments, our results under faithfulness should likewise reassure that price counterfactuals can be flexibly estimated with recentered instruments.

\section{Relationship between completeness and faithfulness}
\label{sec:sufficient_conditions}

We first provide a useful calibration: when $Z$ is a perfect instrument (i.e., the exogenous price case), faithfulness and completeness are exactly equivalent.
  \begin{prop}
  \label{prop:equiv}
      Suppose $P = Z$ and \cref{as:exogenous} holds. Then \cref{as:completeness} is equivalent to \cref{as:faithfulness}.  Hence \cref{prop:exogenous_price} is a special case of \cref{lemma:main}.
  \end{prop}
\noindent 
In this sense, faithfulness is not an exotic new assumption: it is the natural technical condition for leveraging the conditional moment restriction \eqref{eq:cond_moment}. 

Outside of the $P=Z$ case, faithfulness and completeness are distinct conditions. Two counterexamples illustrate this: \Cref{sub:fwithoutC} exhibits a class of demand models where faithfulness holds but $(\delta, P) \mid (X, Z)$ needs not be complete, while \Cref{sub:failure_of_faithfulness} presents a distribution $(\delta, P) \mid (X,Z)$ that is complete but not faithful. 

We next show that the two conditions are nevertheless closely related, in the sense that each condition implies the other under additional restrictions. 

\subsection{When does completeness imply faithfulness?}
\label{subsec:C_to_F}

We first present four non-nested conditions under which completeness implies faithfulness; these conditions are sufficient but not necessary. The sufficient conditions upper-bound the extent to which faithfulness is ``stronger'' than completeness. 

\subsubsection{Restrictions on price-setting}\label{sec:model-p}
We start from two sufficient conditions that restrict how price depends on the
observables and unobservables of the model. Both conditions extend the case of exogenous prices. For some function $f$, write \[
  P = f(X, Z, \delta, 
  \tilde\omega), \qquad \tilde \omega \indep (X, Z, \delta) \numberthis 
  \label{eq:price_eqn_main}.
\]
Here $\tilde \omega$ is an unobservable of arbitrary dimension that captures residual
variation in $P$ that is independent of $(X, Z, \delta)$. So far, \Cref{eq:price_eqn_main} is without loss of generality.\footnote{Note that $f$ is not a structural function because of the parametrization of $\tilde\omega$, but it is consistent with any structural formulation. Indeed, for any structural shock $\omega$ possibly correlated with $(X,Z,\delta)$ and $P=\tilde f(X,Z,\delta,\omega)$, one can represent $\omega = f_\omega(X,Z,\delta, \tilde\omega)$ and $P=\tilde f(X,Z,\delta,f_\omega(X,Z,\delta,\tilde\omega))$.} In what follows we place different restrictions on $f$. At the end of this subsection, we show that these restrictions can be restated as similar conditions on marginal costs under Bertrand--Nash pricing.

The first sufficient condition is that $X$ and $Z$ enter price only through the utility index $\delta$ and an index $\lambda(X, Z)=\lambda$ that is invertible in $Z$:
\smallskip
\begin{as}
\label{as:p-index} $P \indep (X, Z) \mid (\lambda(X, Z), \delta)$, for some $\lambda(x,
z)$ that is invertible in $z$. Equivalently, in \eqref{eq:price_eqn_main}, $ f(X, Z,
\delta,\tilde \omega) = f(\lambda
(X,Z), \delta, \tilde \omega)$.
\end{as} 
\smallskip
\noindent This assumption is satisfied, in particular, when $X$
 enters price in \eqref{eq:price_eqn_main} only through $\delta$---paralleling how it enters $\sigma(\cdot)$. Under this index restriction for price, \cref{as:p-index} is satisfied with 
 $\lambda(X, Z) = Z$. By further setting $Z=P$, this special case also nests exogenous prices and generates \cref{prop:exogenous_price} as a corollary. In general we have:

\smallskip
\begin{restatable}{prop}{proppindex}
\label{prop:p-index}
  \cref{as:p-index,as:exogenous,as:completeness} imply \cref{as:faithfulness}.
\end{restatable}
\smallskip

To see how this result follows, note that for any realization $\lambda_0$ of $\lambda$:
\begin{align*}
k(X) &= \E[H(\delta, P) \mid X, Z] 
\\
& = \E[\E[H(\delta, P) \mid \delta, \lambda] \mid
X, Z] \tag{Iterated expectations, \cref{as:p-index}}
\\
&= \E[\E[H(\delta, P) \mid \delta, \lambda = \lambda_0] \mid X, Z = \lambda^{-1}
(X, \lambda_0)] \tag{$Z$ does not enter $k(X)$}\\
&= \E[\E[H(\delta, P) \mid \delta, \lambda = \lambda_0] \mid X] \tag{\cref{as:exogenous}} \\ 
&= \E[\E[H(\delta, P) \mid \delta, \lambda = \lambda_0] \mid X, Z]. \tag{\cref{as:exogenous}}
\end{align*} Hence for $H(\delta) = \E[H(\delta, P) \mid \delta, \lambda
 = \lambda_0]$ we have $\E[H(\delta) \mid X, Z] = k(X)$. Finally, by completeness (\cref{as:completeness}), $H
 (\delta, P) = H(\delta)$. 

The second sufficient condition instead imposes a separability condition on the derivatives of price with respect to $Z$:
\smallskip
\begin{as}
\label{as:model-p}

  In \eqref{eq:price_eqn_main}, $f$ is continuously differentiable in $Z$ with Jacobian
  $D_z f$, which satisfies the following separability condition: for measurable functions
  $A, B$,
    \[
      \underbrace{D_z f(X, Z, \delta,  \tilde\omega)}_{J \times d_z} = \underbrace{A(f(X, Z, \delta,
      \tilde \omega),
            \delta)}_{J \times J} \cdot \underbrace{B
      (X, Z)}_{J \times d_z} \numberthis \label{eq:separable_maintext}
    \]
    where $A(f(X, \delta,
      Z, \tilde \omega),
            \delta)$ and $B(X,Z)$ are full row rank with $\dim(Z)\equiv d_z \ge J$.
\end{as}
\smallskip
\noindent  \cref{as:model-p} holds, in particular, when  \[ P = f_0\bigg( f_1(X,Z) + f_2(X, \delta , \tilde \omega),
  \delta
  \bigg)  \numberthis \label{eq:fn_form}
\]
(with regularity conditions on $f_0$ and $f_1$ given in \cref{lemma:regular}). The
restriction in \eqref{eq:fn_form} is that $Z$ enters $P$ through an index $f_1(X,Z) + f_2
(X, \delta,\tilde \omega)$, which is a form of separability between observed and
unobserved cost shifters. Clearly, it is satisfied when $P=Z$.\footnote{Moreover, under \eqref{eq:fn_form} and suitable regularity conditions on $f_0, f_1, f_2$, faithfulness is \emph{equivalent} to completeness. To see this, suppose faithfulness holds but completeness fails. By \cref{prop:fimpliesc_maintext}, below, $\delta \mid X$ is not complete: there exists $h(\delta)\ne 0$ such that $\E[h(\delta)\mid X]=0$. Consider $\E[h(\delta) f_0^{-1}(P, \delta) \mid X, Z] = \E[h(\delta) \mid X]f_1(X, Z) + \E[h(\delta) f_2(X, \delta, \tilde\omega) \mid X, Z]$. The first term is 0 and the second does not vary with $Z$. Yet, $h(\delta)f_0^{-1}(P,\delta)$ depends on $P$, contradicting faithfulness.
}

\smallskip
\begin{prop}
\label{prop:model-p}

\cref{as:model-p,as:completeness,as:exogenous} imply \cref{as:faithfulness}. 
\end{prop}
\smallskip

The logic for this result is as follows: under \Cref{eq:separable_maintext}, given
any differentiable $H(\delta, P)$ with $H_p(\delta,P)\equiv\diff{H}{P}$,\footnote{Because we assume $H$ is differentiable
here, we have to modify faithfulness to restrict only differentiable functions. These
complications are resolved in \cref{prop:model-p_x}, which redefines faithfulness and
ensures that differentiation is exchangeable with expectation.} we have
\begin{align*}
0_{J\times d_z} &= \diff{}{z} \E[H(\delta, P) \mid X, Z=z]=\diff{}{z}\E[H(\delta,f(X,Z,\delta,\omega))\mid X,Z=z] \\
&= \E\bk{H_p(\delta, P) A(P, \delta) \mid X,
Z=z}
B(X, z) \tag{\cref{as:model-p}}\\
\implies 0_{J\times J} &= \E\bk{H_p(\delta, P) A(P, \delta) \mid X,
Z} \tag{$B$ is full-rank} \\
\implies 0_{J\times J} &= H_p(\delta, P) A(P, \delta) \tag{\cref{as:completeness}} \\
\implies \hspace{0.4cm} 0_J &= H_p(\delta, P). \tag{$A$ is nonsingular}
\end{align*}
By the fundamental theorem of calculus, $H_p(\delta, P) = 0$ implies $H(\delta, P)=H(\delta)$.  

\Cref{as:p-index,as:model-p} are non-nested. 
\Cref{as:p-index} allows richer interactions between
how $Z$ and $\tilde \omega$: e.g., $P=f(Z,\tilde \omega)$ always satisfies
\cref{as:p-index} but not necessarily \cref{as:model-p}. \Cref{as:model-p} allows richer interactions between $X$ and $\tilde \omega$: e.g., $P=Z+f(X,\tilde \omega)$ always
satisfies \cref{as:model-p} but not necessarily \cref{as:p-index}.

\cref{as:p-index,as:model-p} can be economically grounded with more primitive conditions on marginal costs. Under Bertrand--Nash pricing with constant marginal costs that are represented without loss of generality as $c(X,Z,\delta,\tilde\omega)\equiv C\in\mathbb R^J$, one can always write the equilibrium prices as
\begin{equation}
P=g(C,\delta)\numberthis \label{eq:bertrand-nash}
\end{equation}
for some function $g$ (see \cref{sub:BertrandNash}). Thus, if one assumes that $X$ and $Z$ enter marginal costs via the index $\lambda$, i.e.,
$C=c(\lambda(X,Z),\delta,\tilde\omega),$ then
\cref{as:p-index} follows. Similarly, \cref{as:model-p} holds if $C=f_0(f_1(Z, X) +
f_2(\delta, X,\tilde \omega),\delta)$ with $g$ and $f_0,f_1$ satisfying certain regularity conditions.

\subsubsection{Restrictions on the $\delta$ index}

Our second set of sufficient conditions leverage statistical restrictions on the
conditional distribution of $\delta \mid X$. With these conditions, it is possible to show that \emph{every} function of $X$ belonging to some known class $\mathcal{K}$ can be written as $\E[H(\delta) \mid X]$ for some function $H$. If this is true, we can then conclude from completeness that $H(\delta, P) = H(\delta)$ for some $H(\delta)$, so long as $\E[H(\delta, P) \mid X, Z] \in \mathcal{K}$. To be sure, these conditions are not fully general: 
our main goal is to demonstrate the existence of restrictions that are purely statistical; these restrictions should not necessarily be viewed as recommended modeling choices. 

One simple case is if $\delta$
and $X$ are discrete and known to take the same finite number of values. 
Here, completeness directly implies that any function $k (X)$ can be represented as a
projection of some function $H(\delta)$ to $(X,Z)$-space.  
\begin{restatable}{prop}{lemmafinite}
\label{lemma:finite}
Fix $M \in \N$. Assume the support of $\delta$ and $X$ are both finite sets of size $M$.
Then \cref{as:completeness,as:exogenous} imply \cref{as:faithfulness}.
\end{restatable}

The same strategy can be used in the case of continuously-distributed $X$ and $\delta$. 
Suppose that  $\delta\mid X$ can be transformed into a location-scale
model of the form:
\[
    a(\delta) = b(X) + \Sigma(X) \epsilon, \quad \epsilon \sim q(\cdot), \quad \epsilon
    \indep X, \numberthis \label{eq:location-scale-maintext}
\]
for some invertible $a(\cdot)$ and continuously distributed $J$-dimensional $\epsilon$ with density $q(\cdot)$; $\epsilon$ can be seen as reparametrizing the component of $\delta$ that is independent of $X$. Note that \eqref{eq:location-scale-maintext} is a purely statistical assumption as $\epsilon$ is generally distinct from $\xi$.

We consider the following assumptions on $a(\cdot)$, $b(\cdot)$, $\Sigma(\cdot)$, and $q(\cdot)$:

\begin{as}
\label{as:smoothness}
Fix some integers $s > J+1$ and $M = J+2s+1$. 
  Let \[\K^s = \br{u : \R^J \to \R^J \colon u(x) = Ax + r(x), A \in \R^{J\times J},r \in
  \W^{s,\infty}(\R^J)}\] where $\W^{s,p}(\R^J)$ is a Bessel potential space of smoothness parameter $s$ and integrability parameter $p$, defined in \eqref{eq:bessel}.
  Let $\mathcal K = \K^M$. Assume
\begin{enumerate}
  \item $a,b$ in \eqref{eq:location-scale-maintext} are invertible.

  \item (Density smoothness for $\epsilon$)  The Fourier transform of $q$ and its derivatives have bounded tails obeying
  \cref{as:location-scale-assns}(2).

\item (Limited heteroskedasticity) $\Sigma(x)$ is
uniformly close to some fixed $\Sigma_0$, in the sense that certain distances in
\cref{as:location-scale-assns}(3) are bounded above by a sufficiently small $\psi$.

\item (Smooth $b$ with Lipschitz inverse) It is known that $b\in \mathcal K$ and that \\ $\sup_
{x\in \R^J}\| (D b(x))^{-1}\|_{\op} <\infty
$
where $D b(x)$ is the Jacobian of $b$. 
\end{enumerate}
\end{as}

The key requirement in \cref{as:smoothness} is that $b$ is known to
fall in a ``well-behaved'' function class $\mathcal K$, which---loosely speaking---consists of functions that deviate from a linear map by a suitably smooth function with $s$th-order derivatives in $L^p$. Because of this, we can limit
$\Theta_I$ to just those functions $h$ with conditional expectations in $\mathcal K$: \[
     \Theta_I(\mathcal K) = \br{h : \E[h(S, P) \mid X, Z] \in \mathcal K} \ni a(\sigma^
     {-1}(\cdot, \cdot)).
 \]

The other regularity conditions in \cref{as:smoothness} establish that $\mathcal K$ is
 sufficiently small and the operator $u\mapsto\E[u(\delta) \mid X]$ is sufficiently
 well-behaved so that $\mathcal K$ can be entirely
 populated by conditional expectations of functions of $\delta$.\footnote{\Cref{as:smoothness}(2) is satisfied by
 distributions with Gamma-like tails, though it rules out Gaussian distributions. See
 \cref{lem:B2} and \cref{rmk:radial_gamma}. We strongly suspect that these restrictions
 are not essential and can be further relaxed.} That is, for any $k \in \mathcal
 K$, there exists some $H (\delta)$ such that $
     \E[H(\delta) \mid X] = k(X).
 $ Completeness would then imply a version of faithfulness with respect to $\mathcal K$. That
 is, for
 any candidate $H(\delta, P)$ where $\E[H(\delta, P) \mid X, Z] \in \mathcal K$, it follows that for
 some $H(\delta)$, we have  \[
     \E[H(\delta,P) \mid X, Z] = \E[H(\delta) \mid X, Z],
 \]
 and therefore $H(\delta, P) = H(\delta)$ by completeness. We summarize this argument in the following result and verify that the
 conclusion of \cref{lemma:main} continues to hold. 

\begin{prop}
\label{prop:model_delta}
  Under \cref{as:completeness,as:exogenous,as:invert,as:smoothness}, for any integrable $H
  (\delta,
  P)$, if $\E[H(\delta, P) \mid X, Z] \in \mathcal K$, then $H
  (\delta,
  P) = H(\delta)$ for some $H(\delta)$. As a result, price counterfactuals are identified.
  \end{prop}

 Unlike \cref{lemma:finite}, which requires an exact equivalence in the number of support points of $\delta$ and $X$,  \cref{as:smoothness} is not knife-edge in nature---suggesting that  it can be further relaxed to allow for more flexible models for $\delta \mid X$. Overall, the combination of results in \Cref{subsec:C_to_F} confirms that faithfulness can follow from completeness without strong
 economic or statistical assumptions. We expect that many other sufficient conditions for faithfulness
 exist as well.

 \subsection{When does faithfulness imply completeness?}

The above arguments verify faithfulness from completeness by, intuitively, using $X$ as a
proxy for $\delta$. Since faithfulness  only requires ruling out the possibility of certain
$\delta$ and $P$ interactions averaging out, it is possible---indeed shown by the example in  \cref{sub:fwithoutC}---that faithfulness
does not fully use the completeness of $\delta \mid X$. In other words, completeness of
$\delta \mid X$ is sufficient for faithfulness but may not be necessary. It turns out  this is the only reason that faithfulness does not always imply completeness:

\begin{restatable}{prop}{propfimpliesc}
\label{prop:fimpliesc_maintext}
Suppose \cref{as:exogenous} holds and $\delta \mid X$ is complete. Then \cref{as:faithfulness} implies 
\cref{as:completeness}. 
  
\end{restatable}

\noindent The proof is simple: Given $\E[H(\delta, P) \mid X, Z] = 0$, 
faithfulness implies that $H(\delta, P) = H(\delta)$ since $0$ is constant in $Z$. Exogeneity of $Z$ further implies that $\E
[H(\delta) \mid X, Z] = \E[H(\delta) \mid X] = 0$. Finally, completeness of $\delta \mid
X$ implies that $H(\delta) = 0$.

\section{Practical Implications}\label{sec:practical}

Our nonparametric identification results yield several insights for applied researchers estimating \emph{parametric} demand models. 
Namely, they suggest key conditions that practitioners can scrutinize in order to use our results to argue (perhaps informally) that their parametric assumptions serve the usual role of filling gaps left by insufficient identifying variation. These conditions, discussed here, concern the counterfactuals of interest, the available variation in the price instruments and characteristics, and the estimation method. We also discuss how the conditions differ from what researchers  appealing to \citet{berryhaile} and \citet{berryhaile24} would need to argue.

First, in order to appeal to our results, a researcher should be interested in estimating counterfactuals that change a product attribute for which some plausibly exogenous variation is available. Most commonly, this attribute is the product's price. Price counterfactuals arise, for instance, in merger simulations, when studying product exit or entry (which can be understood as moving prices to or from infinity), and when price elasticities are fundamentally of interest. In each of these cases, one can plausibly observe a set of supply-side shocks $Z$ which exogenously vary prices $P$. 

Second, a researcher should argue there exists at least one ``special'' characteristic $X$  which satisfies the index restriction in \eqref{eq:demand}. For mixed logit models, this would mean $X$ enters demand without a random coefficient; \citet{berryhaile} argue this requirement is  satisfied in most applications. Note that no assumptions are generally needed on how the other characteristics $\tilde{X}$ enter demand.

The researcher should further argue that the special characteristic is a strong proxy for the index $\delta$. In practice, this means that $X$ is strongly predictive of demand holding the other product attributes fixed. Again, no such requirement is generally imposed on the other observed characteristics $\tilde X$.

Importantly, and in contrast with the \cite{berryhaile} results, researchers motivated by our identification results need not justify that either $X$ or $\tilde X$ are exogenous. Several distinct forms of endogeneity are allowed: firms can strategically design products (i.e., choose characteristics both observed and unobserved by the researcher) or choose which markets to enter with a given product with some knowledge of local demand conditions. Moreover, the functional form restrictions that the researcher imposes on how characteristics enter demand need not be correct; for instance, the model is still able to predict price counterfactuals correctly  if quadratic terms in $X$ are incorrectly excluded, violating their exogeneity \citep{andrews2025structural}.

Third, a researcher should scrutinize the exogeneity of the price instruments $Z$. These should be fully independent of the $\delta$ index conditional on $(X,\tilde X)$ which in practice will generally involve certain timing assumptions. The cleanest scenario is when $Z$ is a set of unconditionally as-good-as-random shocks that have some variation across products, are realized after the product design is chosen and entry decisions are made but before $P$ is set, and do not affect demand except through $P$. Short-term unanticipated fluctuations in the exchange rate of countries producing different goods \citep{borusyak2025estimating} or deviations of realized input prices from futures markets values \citep{ackerbergcrawford} are possible examples of such price instruments. However, there are also scenarios where only the conditional exogeneity of $Z$ is plausible; see Appendix A of \citet{berryhaile24} for a detailed discussion.

Fourth, to avoid any bias from the potentially endogenous $(X,\tilde{X})$, the moment conditions used in estimation must be based on recentered instruments. Functions of $Z$ alone are not typically enough to identify parametric demand models beyond pure logit, so other instruments have to be chosen. This paper shows that recentered instruments can be sufficient for this goal nonparametrically, while \citet{borusyak2025estimating} propose specific constructions of powerful recentered instruments in parametric models. Importantly, other instruments---such as ``BLP instruments'' that are functions of characteristics only---are not generally valid without characteristic exogeneity.\footnote{In fact, functions of own and rivals' characteristics can be controlled for and potentially improve estimation efficiency by absorbing some residual variation \citep{borusyak2025estimating}.}

We conclude here by noting that with our identification strategy works with only market-level data, which are widely available. Other types of data---e.g., on how market shares vary by consumer characteristics within the market or on second choices of consumers---are known to be helpful for demand estimation. However, these types of data are not always available and our results show that they are not necessary for identification even nonparametrically. Moreover, even when microdata are available, the identification strategies developed for them impose additional homogeneity assumptions \citep[][Section 4]{chen2025reinterpreting} which our identification strategy does not require.

\section{Conclusion}\label{sec:conclusion}

We have shown that price counterfactuals are identified without exogenous product
characteristics in a nonparametric demand model with a weak index restriction, given exogenous
instruments that induce sufficient variation in prices and a strong proxy for the index. The richness of this identifying variation is captured by a
new faithfulness condition, which essentially
requires that the instruments and the proxies make all causal price effects detectable. We show through a
variety of non-nested sufficient conditions that faithfulness can follow from a standard completeness condition  without strong
statistical or economic restrictions. These results reassure practitioners that price counterfactuals can be reliably identified with recentered instruments, without needing to justify that observed characteristics are as-if-randomly assigned or chosen non-strategically.  We suspect the new faithfulness condition and identification results may also prove useful in other nonparametric models. 

We have not provided a theoretical analysis of nonparametric estimation based on these results, which would naturally require 
additional regularity conditions \citep{compiani2022market,chen2015sieve}. As in \cite{berryhaile}, we leave developing these
conditions to future research. We nevertheless hope that our identification analysis will help
guide empirical researchers towards more robust and credible estimation strategies, by clarifying the
kinds of variation that can reveal counterfactuals in flexible demand models.
Most importantly, in contrast to widely-held intuition, our results
suggest researchers can generally avoid conventional characteristic-based
instruments---and the potential biases associated with them---by looking for plausibly exogenous supply
shocks and leveraging them via recentered instruments.

\bigskip
\bibliographystyle{aer}
\bibliography{refs}

\appendix 

\numberwithin{lemma}{section}
\numberwithin{theorem}{section}
\numberwithin{cor}{section}
\numberwithin{prop}{section}
\numberwithin{as}{section}
\numberwithin{rmk}{section}
\numberwithin{figure}{section}

\numberwithin{equation}{section}

\begin{appendix}

\section{Detailed Statements and Proofs}\label{sec:proofs}

\subsection{Setup}
We first restate the assumptions and results in \Cref{sec:theory,sec:sufficient_conditions} more formally.
Let $F^*$ denote the distribution of $(S, P, \delta, X, Z)$. We assume for each $F^*$
there is some $\sigma$ such that $S = \sigma(\delta, P)$ almost-surely. Let $F$ be the
distribution of $(S, P, X, Z)$. The assumptions will restrict the class of $\mathcal F^*
\ni F^* $, which we keep implicit. Identification of price counterfactuals is formally
defined in the following sense: Counterfactuals are identified if there exists some
function of the observed data that predicts them, and these predictions are correct on a
set of $F^*$-probability 1.

\begin{defn} 
We say that price counterfactuals are identified at $F$ if there exists a
 mapping $C(s,p,p')$ such that, for any $F^*$ generating $F$ with a corresponding
 $\sigma
 = \sigma_{F^*}$, there exists a set $E$ with $\P_{F^*}(E)=1$ where
 \[C (s,p,p') = \sigma(\sigma^{-1}(s,p), p') \] for all values $(s, p, \delta, x, z) \in
 E$ and $ (s', p', \delta, x, z) \in E$. We say that price counterfactuals are identified
  if they are identified at all $F$ generated by some $F^*$.
\end{defn}

We first restate \cref{as:invert}:
\begin{as} 
\label{as:invertible_x}

$\sigma$ is invertible in $\delta$: For some measurable $\sigma^{-1}$, $ \delta = \sigma^
{-1}(S, P) $ almost surely. Assume that $\E[\norm{\sigma^{-1}(S, P)}] < \infty$. 

\end{as}

We next restate \cref{as:completeness,as:exogenous}. Here, completeness is defined relative to
all integrable functions. Weakening completeness by changing integrable to
square-integrable does not affect subsequent results, so long as we modify them
accordingly. 

\begin{as}
    \label{as:complete_x}
    For all $F$, the conditional distribution $(S,P) \mid (X,Z)$ is complete with
    respect to integrable functions: For any integrable $h$, \[
        \E_{F}[h(S, P) \mid X, Z] = 0 \implies h(S,P) = 0 \text{ almost surely. }
    \]
\end{as}

\begin{as}
\label{as:exogenous_x}
    For each $F^*$,  $Z \indep \delta \mid X$ under $F^*$. 
\end{as}

Fix some class of functions $\mathcal H$ mapping $(S, P)$ to $\R^J$. Suppose $0 \in \mathcal H$ and all $h \in
\mathcal H$ are integrable at all $F$. Moreover, suppose all $h(s, p) \in \mathcal {H}$ are invertible in $s$: There exists $h^{-1}$ such that,
almost surely, $
S = h^{-1}(h(S,P), P)
$. Fix
 some class of functions $\mathcal K$ mapping $X$ to $\R^J$. Define an identified set for
 $\sigma^{-1}$ with respect to $\mathcal H$ and $\mathcal K$: \[
    \Theta_I = \Theta_I(\mathcal H, \mathcal K, F) = \br{
        h \in \mathcal H : \E_F[h(S,P) \mid X, Z] = \E_F[h(S,P) \mid X] \in \mathcal K
    }.
\] We will examine throughout identification relative to these function classes $\mathcal
 H$ and $\mathcal K$, representing a priori restrictions (e.g., integrability,
 smoothness, etc.) of a researcher.

\subsection{Identification with recentered instruments}
\begin{lemma}
\label{lemma:recentering}

Let $\mathcal H$ and $\mathcal K$ be all square-integrable functions of $(S,P)$ and of $X$, respectively. Let $\mathcal R$ collect all square-integrable, conditionally mean zero functions of $(X, Z)$: For $R \in \mathcal {R}$, $\E[R(X, Z) \mid X] = 0$ and $\E[R(X,Z)^2] < \infty$. Then \[
\Theta_I = \br{h \in \mathcal H: \E[h(S,P) R(X,Z)] = 0 \,\,\forall R \in \mathcal{R}}.
\]

\end{lemma}

\begin{proof}
First, consider $h \in \Theta_I$---let $k(X) = \E[h(S,P) \mid X,Z]$---and $R \in \mathcal{R}$. Since $R$ and $h$ are both square integrable, $hR$ is integrable. Therefore $\E[hR] = \E[\E[h R \mid X, Z]] = \E[k(X) R(X,Z)] = \E[k(X) \E[R(X,Z) \mid X]] = 0$. Conversely, let $h \in \mathcal {H}$ and $\E[hR] = 0$ for all $R \in \mathcal{R}$. With slight abuse of notation, we overload $h$ to refer to an element of the vector-valued function $h$. Let $R(X,Z) = \E[h(S,P) \mid X,Z] - \E[h(S,P) \mid X] \in \mathcal{R}$. Then \[
\E[hR] = \E[R^2] = 0 \implies R = 0. 
\] 
Hence $\E[h(S,P) \mid X,Z] = \E[h(S,P) \mid X] \in \mathcal{K}$, and thus $h \in \Theta_I$. \qedhere

\end{proof}

\subsection{Suffices to identify $\sigma^{-1}$ up to transformation}
The following formalizes \cref{lemma:transform}: That ensuring $\Theta_I$ contains solely
functions of the form $T(\sigma^{-1}(\cdot))$ suffices for identifying price
counterfactuals. 
\begin{lemma} 
\label{lemma:transform_x}
Assume that for every $F^*$, \begin{enumerate}

\item Every $h \in \Theta_I$ is of the form $h(S, P) = T(\sigma^{-1}(S, P))$ almost surely
 for some measurable $T : \R^J \to \R^J$. 

 \item $\Theta_I$ is nonempty. 
 \end{enumerate}
Then counterfactuals are identified by $C(s, p, p') = h^{-1}(h(s,p), p')$ for any
$h \in \Theta_I$. 
\end{lemma}

\begin{proof}
Fix $F^*$ and its implied $F$. Fix any $h \in \Theta_I$. The assumptions imply that there exists a $F^*$-probability-one set
$E$ such that \[
    s = h^{-1}(h(s, p), p) \quad h(s,p) = T(\sigma^{-1}(s,p))  \quad \delta = \sigma^{-1}
    (s,p) \quad s = \sigma(\delta, p)
\]
for all $(s, p, \delta, x, z) \in E$. Thus $s = h^{-1}\pr{T(\delta), p}.
$ 
Likewise, for some $(s', p', \delta, x, z) \in E$,  \[
    \sigma(\delta, p') = s' = h^{-1}(T(\delta), p') = h^{-1}(h(s,p), p'). \qedhere
\]
\end{proof}

\subsection{Exogenous prices}

\cref{prop:exogenous_price} is a consequence of \cref{prop:p-index_x}: 

\begin{restatable}{prop}{propexopricefirst}
\label{prop:exogenous_price_x_first}
Let \cref{as:p-index_x} hold with $\lambda(X, Z) = Z$ and $Z = P$. Then price
counterfactuals are identified under \cref{as:exogenous_x,as:complete_x}.
\end{restatable}

\begin{proof}
  Let
\begin{enumerate}
    \item $\mathcal F$ be all integrable functions of $(\delta, P)$
    \item $\mathcal H$ be all integrable functions of $(S,P)$ invertible in $S$
    \item $\mathcal K$ be all integrable functions of $X$ 
\end{enumerate}    
Then \cref{prop:p-index_x} implies \emph{(i)} faithfulness follows from completeness under $\mathcal
F, \mathcal H, \mathcal K$ and \emph{(ii)}
\cref{prop:faith_x}'s assumptions hold. Identification follows by
\cref{prop:faith_x}.
\end{proof}

\subsection{Suffices to impose faithfulness}
Let $\mathcal F$ denote a class of functions mapping from $(\delta, P)$ to $\R^J$. We first restate \cref{as:faithfulness} taking into account $\mathcal F$ and
$\mathcal K$:  

\begin{as} 
\label{as:faithful_x} Under \cref{as:exogenous_x},  the conditional distribution $
 (\delta, P) \mid (X, Z )$ satisfies faithfulness with respect to $\mathcal F$ and
 $\mathcal K$: For any $H \in \mathcal F$ such that $\E_{F^*}[H(\delta, P) \mid X,
 Z] \in \mathcal K$, we have that $H(\delta, P) = H(\delta)$ almost surely for some
 measurable $H(\delta)$.
\end{as}

We now formalize \cref{lemma:main}---that faithfulness is sufficient
for identification---relative to these definitions and assumptions: 

\begin{prop} 
\label{prop:faith_x}
Fix $(\mathcal H, \mathcal K, \mathcal F)$. Suppose \cref{as:invertible_x,as:exogenous_x,as:faithful_x} hold for every $F^*$.
 Assume further that for every $F^*$, 
\begin{enumerate}
    \item  There exists an invertible  $T_0(\cdot)$ such that $\E_{F^*}[T_0(\delta) \mid
    X, Z] \in \mathcal K$, $T_0(\sigma^{-1}(\cdot)) \in \mathcal H$. 
    \item For every $h \in \mathcal H$, $h(\sigma(\delta, p), p) \equiv H(\delta, p) \in
    \mathcal F$.
\end{enumerate}
Then price counterfactuals are identified.
\end{prop}

\begin{proof} By \cref{lemma:transform_x}, we need to check assumptions (1)--(2) in \cref
 {lemma:transform_x}. First, condition (1) assumed implies that $T_0(\sigma^{-1}(s, p)) \in \mathcal H$ and
 its conditional expectation lies in $\mathcal K$. Thus $T_0(\sigma^{-1}
 (s, p)) \in \Theta_I$, verifying (2) in \cref{lemma:transform_x}. Now fix $h \in \Theta_I$; by condition (2), $h
 (\sigma(\delta,p), p) \in \mathcal F$. By \cref{as:faithful_x}, we have that \[ h(S, P)
 = H(\delta) = H(\sigma^{-1}(S, P)) \] almost surely. This verifies assumption
 (1) in \cref{lemma:transform_x}.
\end{proof}

\subsection{Faithfulness verification from $\lambda$ index}\label{subsec:model-p-lambda}

Next, we state and verify the case (\cref{prop:p-index}) where 
\[ P \indep (X,Z) \mid \lambda
(X, Z), \delta.
\]
We show \cref{prop:equiv} as partly a corollary of \cref{prop:p-index} with
$P = \lambda(X, Z) = Z$.

\begin{as}
\label{as:p-index_x} Let $\lambda: \R^J \times \mathcal Z \to \mathcal L \subset \R^{d_z}$ be
an
index that is invertible in $Z$.
\begin{enumerate}
  \item For such a $\lambda$, $
  P \indep (X,Z)\mid \lambda(X, Z), \delta.$
  \item Let $\lambda = \lambda(X, Z)$. The joint distribution of $(X, \lambda)$ has a density with respect to some product
   measure $\mu_X \otimes \mu_\lambda$ over $\mathcal X \times \mathcal L$, which
   is
   strictly positive for $\mu_X \otimes \mu_\lambda$-almost every $
   (x,\lambda) \in \mathcal X  \times S_\lambda$ for some $S_\lambda \subset
  \mathcal L$ and $(\mu_X \otimes \mu_\lambda)(\R^J \times
  S_\lambda) > 0$.

\end{enumerate}
\end{as} 

\begin{restatable}{prop}{propexoprice}
\label{prop:exogenous_price_x}
Let \cref{as:exogenous_x} hold and let \cref{as:p-index_x} hold with $\lambda(X, Z) = Z$ and $Z = P$. Suppose also that $P$
is not degenerate given $\delta$: for any $S_\delta$ where $\P_{F^*}(S_\delta) > 0$, there are two disjoint sets
$S_1, S_2$ such that $\P_{F^*}(\delta \in S_\delta \cap P \in S_1) > 0, \P_{F^*}(\delta
\in S_\delta \cap P \in S_2) > 0$. Let
\begin{enumerate}
    \item $\mathcal F$ be all integrable functions of $(\delta, P)$
    \item $\mathcal H$ be all integrable functions of $(S,P)$ invertible in $S$
    \item $\mathcal K$ be all integrable functions of $X$. 
\end{enumerate}    
Then \cref
 {as:faithful_x} holds if and only if \cref{as:complete_x} holds. 
\end{restatable}

\begin{proof}
  The ``if'' direction is a corollary of \cref{prop:p-index_x}. For the only if direction,
  consider the contrapositive. Suppose \cref{as:complete_x} does not hold.

  Then, there exists a nonzero $f(\delta, P) \in \mathcal F$ such that $\E[f(\delta, P)
  \mid X, Z] = 0$. If $f(\delta, P) \neq f(\delta)$ for some $f(\delta)$, then 
  \cref{as:faithful_x} fails and we are done. Otherwise, we have $\E[f(\delta) \mid X, Z]
  = \E[f(\delta) \mid X] = 0$ but $f(\delta) \neq 0$. Let $S_\delta = \br{\delta : f
  (\delta) \neq
  0}$ be such that $\P_{F^*}(S_\delta) > 0$. By the non-degeneracy condition, we can
  choose sets $S_1, S_2$ and $B(P) = \one(P \in S_1)$. Then $
    H(P, \delta) = B(P)f(\delta)
  $
  is integrable and not constant in $P$. But $\E[H(P, \delta) \mid X, P] = B(P) \cdot 0 =
  0$. Thus \cref{as:faithful_x} does not hold.
\end{proof}

\begin{restatable}{prop}{proppindex}
\label{prop:p-index_x}

Let \begin{enumerate}
    \item $\mathcal F$ be all integrable functions of $(\delta, P)$
    \item $\mathcal H$ be all integrable functions of $(S,P)$ invertible in $S$
    \item $\mathcal K$ be all integrable functions of $X$

\end{enumerate}
Then, 
  \begin{enumerate}
    \item \cref
 {as:faithful_x} follows from \cref{as:complete_x} and \cref{as:exogenous_x}.
 \item The assumptions of \cref{prop:faith_x} are satisfied. 

\end{enumerate}
\end{restatable}

\begin{proof}

\begin{enumerate}[wide]

\item 
Take $H \in \mathcal F$ with \[
        \E \bk{H(\delta, P) \mid X, Z} = k(X) \in \mathcal K, \quad \text{a.s.}
    \]
    By law of iterated expectations, \begin{align}\label{eq:delta_X_lambda}
    k(X) = \E[ \E[ H(\delta, P) \mid \delta, X, \lambda ]\mid X,Z],\quad \text{a.s.}
    \end{align}
    since $(X, Z) = (X, \lambda^{-1}(X,\lambda))$ is measurable with respect to $(X, \lambda)$ and hence to $(\delta, X, \lambda)$. By \cref{as:p-index_x} (1), 
    \begin{equation}\label{eq:r_H_def}
    \E[ H(\delta, P) \mid \delta, X, \lambda ] = \E[ H(\delta, P) \mid \delta, \lambda ] := r_H(\delta, \lambda),\quad \text{a.s.}
    \end{equation}
    Since $(X,\lambda) = (X, \lambda(X,Z))$ is measurable with respect to $(X, Z)$,
    \[k(X) = \E[k(X)\mid X,\lambda] = \E[\E[\bar{H}(\delta,\lambda)\mid X,Z]\mid X,\lambda] = \E[r_H(\delta,\lambda)\mid X,\lambda], \quad \text{a.s.}\]
    By \cref{as:p-index_x} (2), for $(\mu_X\otimes \mu_\lambda)$-almost every $(x, \lambda')$,
    \[k(x) = \E[r_{H}(\delta,\lambda')\mid X=x,\lambda=\lambda'].\]
    By \cref{as:exogenous_x}, 
    $\lambda \indep \delta \mid X.$
    Thus, for $(\mu_X\otimes \mu_\lambda)$-almost every $(x, \lambda')$, 
    \begin{equation}\label{eq:r_H_lambda}
    k(x) = \E[r_H(\delta,\lambda')\mid X=x].
    \end{equation}
    Fix any $\lambda_0\in S_\lambda$ such that the above equation holds for $\mu_X$-almost every $x$. Then, for $\mu_\lambda$-almost every $\lambda'$, 
    \begin{equation*}\label{eq:r_H_diff}
    0 = \E[r_H(\delta,\lambda') - r_H(\delta,\lambda_0)\mid X=x], \quad \text{for }\mu_X\text{-almost every } x.
    \end{equation*}
    Since $X$ has positive density on $\R^J$, \begin{equation*}
    \E[r_H(\delta,\lambda') - r_H(\delta,\lambda_0)\mid X] = 0 \,\, \text{ a.s. in }X\text{, for }\mu_\lambda\text{-almost every }\lambda'.
    \end{equation*}
    By \cref{as:complete_x}, 
    \[r_H(\delta,\lambda') = r_H(\delta,\lambda_0):= \tilde{H}(\delta) \,\, \text{ a.s. in }\delta\text{, for }\mu_\lambda\text{-almost every }\lambda'.\]
    Using \cref{as:p-index_x} again, $\lambda$ has a positive density on $S_\lambda$ and thus, 
    \[r_H(\delta,\lambda) = \tilde{H}(\delta)\quad \text{ a.s. in }(\delta,\lambda).\]
    Together with \eqref{eq:r_H_def}, we have
     \[\E[H(\delta, P) - \tilde{H}(\delta)\mid \delta, X, \lambda] = 0, \quad \text{a.s.}\]
    Taking conditional expectation given $(X,Z)$ on both sides and using the same argument for \eqref{eq:delta_X_lambda}, we obtain
    \[\E[H(\delta, P) - \tilde{H}(\delta)\mid X, Z] = 0, \quad \text{a.s.}.\]
    Completeness (\cref{as:complete_x}) then implies $H(\delta, P) = \tilde{H}(\delta)$, hence faithfulness.
\item The proof is analogous to \cref{prop:discrete_x}(2).  \qedhere
\end{enumerate}
\end{proof}

\subsection{Faithfulness verification from price separability}

We next restate \cref{as:model-p,prop:model-p}, which rely on restrictions about
derivatives of $P$ in $z$. \eqref{eq:price_eqn} is slightly different from---but is implied by---its counterpart in \cref{as:model-p}.

\begin{as}
\label{as:model-p_x} We have the following:
  \begin{enumerate}
 \item (Support) For almost every $\delta$, there is a connected open set $U_\delta
            \subset \R^J$ such that $P \mid \delta$ has a density $f(p \mid \delta) > 0$
            on $U_\delta$ and $\P(P \in U_\delta \mid \delta) = 1$. The distribution $
            (S,P,Z)\mid X$ is absolutely continuous with respect to the Lebesgue measure
            on $\R^{J \times J \times d_z}$. There exists an open set $U \subset \R^{J
            \times J \times d_z}$ where the density of $(S,P,Z) \mid X$ is positive on $U$
            and $P((S, P,Z) \in U \mid X) = 1$ a.s.

    \item (Price equation) The random variable $P$ can be represented as \[ P = f
    (X,\delta, Z, \tilde \omega)
    \text{ for } (\tilde \omega, \delta) \indep Z \mid X
      \numberthis
      \label{eq:price_eqn}
    \]
    for some function $f$ continuously differentiable in $Z$ with Jacobian $D_z f$. The joint
    distribution $ (\tilde \omega, \delta)
    \mid X$  has density $f(\tilde \omega, \delta \mid x)$ with respect to some dominating
    measure
    $\lambda_x$.

    \item (Price separability) The Jacobian of $f$ in $Z$ satisfies the following: For
    measurable functions
    $A, B$, \[
      \underbrace{D_z f(X, \delta, Z, \tilde \omega)}_{J \times d_z} = \underbrace{A(f(X, \delta,
      Z, \tilde \omega),
            \delta)}_{J \times J} \cdot \underbrace{B
      (X, Z)}_{J \times d_z} \numberthis \label{eq:separable}
    \]
    where $A(f(X, \delta,
      Z, \tilde \omega),
            \delta)$ and $B(X,Z)$ are both full-rank a.s. with $d_z \ge J$

            \item (Domination) There exists $G_1(x, \delta, \tilde \omega) \ge 0$ and $G_2(x, \delta,
            \tilde \omega) \ge 0$ such that a.s.,\footnote{For concreteness, we choose the
            Frobenius norm as the matrix norm.}
  \begin{align*}
  \norm{\sigma_p(\delta, P) f_z(X, Z, \delta, \tilde \omega)} \le G_1(X, \delta, \tilde \omega) \\ 
  \norm{f_z(X, Z, \delta, \tilde \omega)} \le G_2(X, \delta, \tilde \omega)
  \end{align*}
  and $\E[G_1^4 + G_2^4 \mid X] = \E[G_1^4 + G_2^4 \mid X, Z] < \infty$.

\end{enumerate}
\end{as}

\begin{prop}
\label{prop:model-p_x}

Let 
\begin{enumerate}
    \item $\mathcal
F$ denote the class of functions $H(\delta, P)$ where (i) $\E [H^2 \mid X, Z] < \infty$,
(ii) for $F^*$-almost every $\delta$, $H(\delta, \cdot)$ is continuously differentiable
with derivative $H_p(\delta, \cdot)$ on $U_\delta$, and (iii) \[\diff{}{z}\E[H(\delta,
f(X,\delta, z,
\tilde \omega))
\mid X, Z=z] = \E[H_p(\delta, f(X,\delta, z,
\tilde \omega)) D_z f(X,\delta, z,
\tilde \omega) \mid X]. \tag{Interchange of differentiation and expectation}\]

\item $\mathcal H \subset L^2_F(S,P)$ be
            the set of $\R^J$-valued functions $h (s,p)$ such that (i) $h$ is continuously
            differentiable in $(s,p)$ on $U$, (ii) the derivatives are locally
            Lipschitz: There exists $\epsilon > 0$ such that for each $j=1,\ldots,J$ and
            all values $s, s', p, p'$ where $\norm{(s,p) - (s',p')} < \epsilon$, \[
              \max\pr{\norm[\bigg]{\diff{h_j}{s}(s,p) - \diff{h_j}{s}(s',p')}, \norm[\bigg]{
              \diff{h_j}{p}(s,p) - \diff{h_j}{p}(s',p')}} \le L(s',p') \norm{(s,p) - (s',p')}
            \]
            where $\E[L^2(S,P) \mid X, Z] < \infty$, and (iii)
The derivatives $h_s, h_p$ are integrable: \[
    \E[\norm{h_s(S,P)}^2 + \norm{h_p(S,P)}^2 \mid X, Z] < \infty.
  \] (iv) $h$ is invertible in $S$.

\item $\mathcal K$ be all square-integrable functions of $X$.

\end{enumerate}

Assume $\sigma^{-1} \in  \mathcal H$. Suppose
\cref{as:model-p_x,as:complete_x,as:exogenous_x,as:invertible_x} hold.  Then,
\begin{enumerate}
    \item \cref{as:faithful_x} holds
    \item The assumptions in \cref{prop:faith_x} hold
\end{enumerate}
\end{prop}

\begin{proof}
\begin{enumerate}[wide]
    \item Let $H \in \mathcal F$. Write \begin{align*} k(x) = \E[H(\delta, P) \mid X=x, Z=z] &= \int
  H(\delta, f(x, \delta, z, \tilde \omega)) \, f(\tilde \omega, \delta \mid x) d
  \lambda_x.
  \end{align*}
  Since $\mathcal H\in \mathcal F$, differentiating both sides in $z$ yields \[ 0 =
    \E[H_p(\delta, P) A(P, \delta)  \mid X, Z] \cdot B(X, Z).
  \]
  Since $d_z \ge J$ and $B(X,Z)$ is full rank, we have that $
    \E[H_p(\delta, P) A(P, \delta)  \mid X, Z] = 0 
 $ almost surely.
  By completeness, since $A$ is full rank,\[
    H_p(\delta, P) A(P, \delta) = 0 \implies H_p(\delta, P) = 0 \text{ a.s. }
  \]
  Since $H_p$ is continuous in $p$ and $P$ is supported with
  positive density on $U_\delta$, $H_p (\delta, p) = 0$ on $U_\delta$. As a result,
  $H(\delta, \cdot)$ is constant on $U_\delta$ for almost every $\delta$. Hence there
  exist some  $H(\delta)$ where $ H(\delta) = H(\delta, P)
  $
  almost surely, verifying \cref{as:faithful_x}.

  \item \cref{lemma:leibniz} verifies that assumption (2) of \cref{prop:faith_x} holds.  Finally, assumption (1) of \cref{prop:faith_x} holds by choosing $T_0(d) = d$, since \[
    \sigma^{-1} \in \mathcal H, \quad \E[\sigma^{-1} \mid X, Z] = k(X) \in L^2_F(X) =
    \mathcal K,
  \]
  under the assumption that $\sigma^{-1}$ is square-integrable. \qedhere
\end{enumerate}
\end{proof}

\begin{lemma}
  \label{lemma:leibniz}
  Suppose \cref{as:model-p_x} holds. Let $\mathcal H, \mathcal F$ be defined as in 
  \cref{prop:model-p_x} and suppose $\sigma^{-1} \in \mathcal H$. 
  Define $
    H(\delta, p) = h(\sigma(\delta, p), p)
  $. Then $H \in \mathcal F$. 
\end{lemma}

\begin{proof}
  It suffices to show that differentiation in $z$ and expectation are interchangeable for
  $H$, since the other properties of $\mathcal F$ are assumed in $\mathcal H$. Since it
  suffices to show the claim entrywise over entries of $h$, we abuse notation and denote
  $h(s,p)$ as an entry and treat it as a scalar function.

  We let $f_z$ denote $D_z f$. Define $h^*(\delta, x, z, \tilde \omega) = h(\sigma(\delta,
  f(x,\delta ,z, \tilde \omega)), f(x,\delta, z,
\tilde \omega))$. Fix some $z_0$. Let \[
  h^*_z(\delta, X, Z, \tilde \omega) = \pr{h_s(S, P)' \sigma_p(\delta, P) + h_p(S, P)'}  f_z
  (X, \delta, Z,  \tilde \omega)
\] be the (transposed) gradient at $(\delta, X, Z, \tilde \omega)$. We show that if $h \in
\mathcal H$  then \[
  \diff{}{z}\E[h^*(\delta, X, Z, \tilde \omega) \mid X, Z] = \E[h^*_z(\delta, X, Z, \tilde \omega)
  \mid X, Z]. 
\]

It suffices to show that \[
  \lim_{t\to 0} \E\bk{\sup_{\norm{v} = 1}  \abs[\bigg]{  \frac{h^*(\delta, X, z_0+tv,
  \tilde \omega) - h^* (\delta, X,
  z_0, \tilde \omega)}
  {t} - h^*_z
  (\delta, X, z, \tilde \omega) v} \mid X} = 0. \numberthis \label{eq:L1error}
\]
\eqref{eq:L1error}  implies \[
  \lim_{t\to0} \sup_{\norm{v} = 1} \abs[\bigg]{\E\bk{\frac{h^*(\delta, X, z_0+tv, \tilde \omega)
  - h^* (\delta, X,
  z_0, \tilde \omega)}
  {t} \mid X} - \E[h^*_z(\delta, X, z, \tilde  \omega) \mid X]  v} = 0,
\]
which implies that $z \mapsto \E[h^*(\delta, X, z, \tilde \omega) \mid X] $ is (Frechet)
differentiable at $z_0$ and its gradient is equal to $\E[h^*_z(\delta, X, z,
\tilde \omega) \mid X]$.

Towards \eqref{eq:L1error}, at a fixed $(\delta, x, \tilde \omega)$, observe that since $h^*$ is
differentiable at $z_0$, we have pointwise convergence\[
  \lim_{t \to 0}\sup_{\norm{v} = 1}  \abs[\bigg]{  \frac{h^*(\delta, x, z_0+tv,
 \tilde  \omega) - h^* (\delta, x,
  z_0, \tilde \omega)}
  {t} - h^*_z
  (\delta, x, z_0, \tilde \omega) v} = 0.
\]
Thus it suffices to show the following and apply the dominated convergence theorem: For
all sufficiently small $t$, \[
   \sup_{\norm{v} = 1}  \abs[\bigg]{  \frac{h^*(\delta, X, z_0+tv,
  \tilde \omega) - h^* (\delta, X,
  z_0, \tilde \omega)}
  {t} - h^*_z
  (\delta, X, z_0, \tilde \omega) v} \le G_0(\delta, X, \tilde \omega)
  \numberthis \label{eq:dct_condition}
\]
where $\E[G_0(\delta, X, \tilde \omega)
  \mid X] < \infty.$ 

  Observe that \begin{align*}
  \norm{h^*_z(\delta, x, z_0, \tilde \omega)} &\le \norm{h_s(s, p)} \norm{\sigma_p(\delta, p)
    f_z
  (x, z_0, \delta, \tilde \omega)} +
    \norm{h_p (s, p)} \norm{f_z
  (x, z_0, \delta, \tilde \omega)} \\
  &\le \norm{h_s(s, p)} G_1(x, \delta,  \tilde \omega) + \norm{h_p (s, p)}G_2(x, \delta,  \tilde \omega)
  \end{align*}
  for $s = \sigma(\delta, x, z_0, \tilde  \omega), p = f(x, \delta, z_0,  \tilde \omega)$. Since the
  derivatives of $h$ and $G_1, G_2$ are all square-integrable, by Cauchy-Schwarz $
  \norm{h_s(s, p)} G_1(x, \delta, \tilde  \omega) + \norm{h_p (s, p)}G_2(x, \delta,  \tilde \omega)$ is an
  integrable dominating function.

Therefore it suffices to show that the difference quotient can be dominated \[
   \sup_{\norm{v} = 1}  \abs[\bigg]{  \frac{h^*(\delta, X, z_0+tv,
   \tilde \omega) - h^* (\delta, X,
  z_0,  \tilde \omega)}
  {t} } \le G_0(\delta, X,  \tilde \omega)
  \numberthis \label{eq:dct_condition_1}
\]

Towards \eqref{eq:dct_condition_1},  define $p_{tv} = f
(\delta,
x, z_0 + tv,  \tilde \omega)$ and $s_{tv} = \sigma(\delta, p_{tv}).$ By the mean-value theorem in
$h$, we have
\[h^*(\delta, x, z_0+tv,
   \tilde \omega) - h^* (\delta, X,
  z_0,  \tilde \omega) = h_s(\tilde s_{t,v}, \tilde p_{t,v}) (s_{tv} - s_0) + h_p(\tilde s_
  {t,v}, \tilde p_{t,v}) (p_{tv} - p_0) \numberthis \label{eq:expansion}
\]
where $(\tilde s_{t,v}, \tilde p_{t,v})$ is some point on the line segment connecting $
(s_0, p_0)$ with $(s_{tv}, p_{tv}).$

We can write $h_s(\tilde s_{tv}, \tilde p_{tv}) = h_s(s_0, p_0) + R_s(s_0, p_0)$ where
\[\norm{R_{s}(s_0, p_0)} \le L(s_0, p_0) t (G_1 (x,\delta,  \tilde \omega) + G_2(x, \delta,  \tilde \omega))\]
for all $t < \epsilon$. Similarly for $h_p$.
Thus \begin{align*}
\eqref{eq:expansion} &= h_s(s_0, p_0) (s_{tv} - s_0) + R_s(s_0, p_0) (s_{tv} - s_0) + h_p
(s_0, p_0) (p_{tv} - p_0) + R_p(s_0, p_0)  (p_{tv} - p_0)
\end{align*}
By the mean-value theorem applied to $z \mapsto \sigma(\delta, f(\delta, x, z,  \tilde \omega))$,
we have that \[
  s_{tv} - s_0 = \br{\sigma_p(\delta, \check p_{tv}) f_z(x, \check z, \delta,  \tilde \omega)}
  tv
\]
where $\check z$ is on the line segment between $z_0, z_0 + tv$ and $\check p_{tv} = f(x,
\check z, \delta,  \tilde \omega)$. We thus have \[
  \norm{s_{tv} - s_0} \le t G_1(x, \delta,  \tilde \omega). 
\]
Similarly, we have  \[
  \norm{p_{tv} - p_0} \le t G_2 (x, \delta,  \tilde \omega). 
\]

Thus, for all $t < \epsilon$, 
\begin{align*}
\sup_{\norm{v} =1} \norm[\bigg]{\frac{\eqref{eq:expansion}}{t}} &\le \norm{h_s(s_0,
p_0)} G_1 (x, \delta,
 \tilde \omega) + \norm{h_p(s_0, p_0)} G_2(x, \delta,  \tilde \omega) \\&\quad+ t L(s_0, p_0) \br{G_1(x,
\delta,  \tilde \omega) + G_2(x, \delta, \tilde  \omega)}^2
\end{align*}
The right-hand side is a function of $( \tilde \omega, \delta, x)$ that is integrable by
Cauchy--Schwarz. This concludes the proof.
\end{proof}

\begin{restatable}{lemma}{lemmaregular}
\label{lemma:regular}
  Let $(\tilde \omega, \delta) \indep Z \mid X $ and suppose $P$ is generated by 
  \eqref{eq:fn_form}. Assume the dimension of $z$ is at least $J$. Assume that
  \begin{enumerate}
  \item The codomains of $f_0, f_1, f_2$ are all $\R^J$.
 \item $f_0(\cdot, \delta)$ is invertible where $f_0^{-1}(\cdot, \delta)$ is
    continuously differentiable for every $\delta$ with invertible Jacobian $Q_0(p,
    \delta)
    \in \R^{J\times J}$. 
    \item $f_1$ is continuously differentiable in $z$ with full-rank Jacobian $Q_1(z, x)
    \in \R^ {J    \times d_z}$
  \end{enumerate}
  Then items 2 and 3 of \cref{as:model-p_x} are satisfied.
\end{restatable}

\begin{proof}
  Write $
    f_0^{-1}(p, \delta) =  f_1(z, x) + f_2(\delta, x, \tilde \omega).
 $
  Differentiate in $z$ implicitly to obtain \[
    Q_0(p, \delta) D_z f = Q_1(z, x).
  \]
  Since $Q_0$ is assumed to be invertible, we have \[
     D_z f = \underbrace{Q_0^{-1}(p, \delta)}_A \underbrace{Q_1(z,x)}_B.
  \]
  By assumption $A$ and $B$ are full rank.  
\end{proof}

\subsection{Faithfulness verification from discrete support} We verify completeness implies faithfulness if $X$ and $\delta$ are discrete with same-sized supports.

\begin{prop} 
\label{prop:discrete_x}
For every $F^*$, suppose the support of $\delta$ and $X$ are both finite sets
 of size $M \in \N$. Let 
\begin{enumerate}
    \item $\mathcal F$ be all integrable functions of $(\delta, P)$,
    \item  $\mathcal H$ be all integrable functions of $(S, P)$ invertible in $S$
    \item $\mathcal K$ be all $\R^J$-valued functions of $X$.
\end{enumerate}
 Then, with these choices, 
 \begin{enumerate}
    \item \cref
 {as:faithful_x} follows from \cref{as:complete_x} and \cref{as:exogenous_x}.
 \item The assumptions of \cref{prop:faith_x} are satisfied. 
 \end{enumerate}
\end{prop}

\begin{proof}
  Note any real-valued $h(\delta)$ and $g(x)$ can be represented as $\R^M$ vectors. Consider \[
  \E[h(\delta)\mid X,Z] = \E[h(\delta) \mid X] \equiv g(X). 
  \]
If we represent both as vectors in $\R^M$, then we have $
  Q'_{\delta \mid X} h = g,$ 
for a matrix $Q_{\delta \mid X}$  whose $(i,j)$\th{} entry  is $\P_{F^*}
 (\delta=\delta_i \mid X=x_j)$, where we number values in the support of $\delta$ as
 $\delta_i$ and likewise for values of $X$. The matrix $Q_ {\delta \mid X}$ is square by
 assumption. Completeness implies that $Q'_{\delta \mid X}$ is full-rank, and thus
 invertible.

Now, consider any function where $\E[H(\delta, P) \mid X, Z] = k(X) \in \mathcal K$. For
each coordinate $j$, there is a function $H_{0j}(\delta)$ which can be represented as $H_
{0j} = (Q'_{\delta \mid X})^{-1} k_j$. Construct $H_0 = (H_{01},\ldots, H_{0J})'$. By
construction $\E[H_0(\delta)\mid X, Z] = k(X)$. Completeness implies that $H(\delta, P) =
H_0(\delta)$. This shows (1).

For (2), we can easily check that $\sigma^{-1} \in \mathcal H$ and $\E[\sigma^{-1} \mid X,
Z] = \E[\delta \mid X] \in \mathcal K$. Thus (1) in \cref{prop:faith_x} is satisfied with
$T_0(\delta) = \delta$. If $h(S, P) \in \mathcal H$ is integrable, then so is $H
(\delta, P) = h(\sigma (\delta, P), P)$ since $S = \sigma(\delta, P)$ a.s. Thus
$H \in \mathcal F$. This verifies (2) in 
\cref{prop:faith_x}.
\end{proof}

\subsection{Faithfulness verification from location-scale model for $\delta$}\label{subsec:model-delta}

\subsubsection{Mathematical preliminaries}

Let $\N$ denote the set of nonnegative integers and
$\mathbb{S}_+^{J}$ denote the set of $J\times J$ positive semidefinite matrices. For any
multi-index $\alpha = (\alpha_1, \ldots, \alpha_J)\in \N^J$ and function $f(x): \R^{J}
\mapsto \R^s$, where $x\in \R^J$ and $s$ may be different from $J$, we write $|\alpha| =
\sum_{j=1}^{J}\alpha_j$, $\alpha! = \prod_{j=1}^{J}\alpha_j !$, and
\[x^\alpha = x_1^{\alpha_1}x_2^{\alpha_2}\cdots x_J^{\alpha_J}, \quad D_x^\alpha f(x) = D_{x_1}^{\alpha_1}D_{x_2}^{\alpha_2}\cdots D_{x_J}^{\alpha_J} f(x).\]
Furthermore, for any function $g: \R^{J}\mapsto \R$, we define its Fourier transform as 
\[\hat{g}(\omega) = \int_{\R^{J}}e^{-i\langle \omega, u\rangle}g(u)du.\]
For a function class $\mathcal{L}$, $f\in \mathcal{L}^{\otimes J}$ iff $f = (f_1, \ldots, f_J)'$ where $f_j\in \mathcal{L}$.

Given a function $v: \R^J\mapsto \R$ for which $v$ is positive everywhere, let $\W^{s,p}(\R^J)$ denote the Bessel potential space of indices $(s, p)$: 
\[\W^{s,p}(\R^J; \R^J) = \left\{u: \R^J \to \R^J: (1 + \|\omega\|_2^2)^{s/2}\hat{u}_j
(\omega)\in L^p(\R^J), \,\, j \in [J]\right\}. \numberthis \label{eq:bessel}\] 
For any $s = m+\sigma$ where $m= \lfloor s\rfloor$ is the integer part and $\sigma\in [0,
1) $ is the fractional part, let $\|\cdot\|_{\W^{s,p}}$ denote the
associated norm with
\[\|u\|_{\W^{s,p}} = \sum_{|\alpha|\le m}\|D^\alpha u\|_{L^p} + \sum_{\alpha: |\alpha|=m}[D^\alpha u]_{W^{\sigma,p}},\] 
where 
\[[h]_{\W^{\sigma,p}}:= \left(\int_{\R^J}\int_{\R^J}\frac{\|h(x) - h(y)\|^p}{\|x-y\|^{\sigma p +J }}dx dy\right)^{1/p},\quad [h]_{\W^{\sigma,\infty}}:=\sup_{x\neq y}\frac{\|h(x)-h(y)\|}{\|x-y\|^\sigma}.\]

Define the following space:
\begin{align*}
\K^s(\R^J; \R^J) &\equiv \Lin(\R^J; \R^J) \oplus \W^{s,\infty} (\R^J;\R^J) \numberthis \label{eq:Ks}\\
&\equiv \{u: \R^J\to \R^J: u(x) = Ax + r(x), A\in \R^{J\times J}, r\in \W^{s,\infty}
(\R^J, \R^J)\}
\end{align*}
where $\Lin(\R^J; \R^J)$ is the space of all bounded linear functions that map
$\R^J$ to $\R^J$. In the following, for notational convenience, we suppress the spaces
$\R^J$ or $\R^J;\R^J$ in $\W^{s,\infty}$ and $\K^s$. For any $s_0\le s_1$, since $\W^{s_1,
\infty}\subset \W^{s_0, \infty}$, we have $\K^{s_1}
\subset \K^{s_0}$.

Last, for a linear operator $\cT$ that maps a Banach space $\mathcal{U}$ to a Banach space $\mathcal{V}$, let 
\[\|\cT\|_{\mathcal{U}\mapsto \mathcal{V}} = \sup_{u\in \mathcal{U}}\frac{\|\cT u\|_{\mathcal{V}}}{\|u\|_{\mathcal{U}}}.\]

\subsubsection{Assumptions}

\begin{as}
\label{as:location-scale-assns}
There exists $s ,M \in \N$ with $M = J+2s+1$ and $s > J+1$ such that the following holds:
  \begin{enumerate}%
\item There exist invertible $a, b: \R^J \to \R^J$ with 
\[a(\delta) = b(X) + \Sigma(X)\epsilon \numberthis \label{eq:location-scale}\]
where $\epsilon\indep X$ and $\Sigma: \R^J\mapsto \S_+^J$. 

\item Let $q$ denote the density function of $-\epsilon$ with respect to the Lebesgue measure,
which we assume to exist. There exist $0 < C_- < C_+ < \infty$ such
that, for any
$\upsilon\in \R^J$,
\[C_-(1 + \|\upsilon\|_2^2)^{-s/2} \le |\hat{q}(\upsilon)|\le C_+(1 + \|\upsilon\|_2^2)^{-s/2},\]
and, for any $\alpha \in \N^J$ with $|\alpha|\le M$,
\[\|D_\upsilon^{\alpha} \hat{q}(\upsilon)\|\le C_+(1 + \|\upsilon\|_2^2)^{-s/2-|\alpha|/2}\]
\item There exist $\Sigma_0 \in \S_+^J$, and $\psi \in (0, \lambda_{\min}(\Sigma_0) / 2)$
such that,%
\[\sup_{x\in \R^J}\|\Sigma(x) - \Sigma_0\|_{\op} + \int_{\R^J}\|\Sigma(x) - \Sigma_0\|_{\op}dx\le \psi,\]
and, for any $\gamma\in \N^J$ with $|\gamma|\le M$,
\[\sup_{x\in \R^J}\|D_x^{\gamma}\Sigma(x)\|_{\op} + \int_{\R^J} \|D_x^{\gamma}\Sigma(x)\|_{\op}dx\le \psi.\]
\item $b\in \K^{M}(\R^J)$ and $\sup_{x\in \R^J}\|
(D b(x))^{-1}\|_{\op} <\infty$.
\end{enumerate}
\end{as}

Intuitively, 
\begin{itemize}
  \item \cref{as:location-scale-assns}(1) assumes that $\delta \mid X$ follows a
location-scale model.
\item \cref{as:location-scale-assns}(2) assumes that this familly has
shape corresponding to some density $q(\cdot)$, and $|\hat q(v)| \rateeq \norm{v}^{-s}$
for large $v$ in signal space, with derivatives up to $M$ of the Fourier transform
satisfying $
\norm{D_v^\alpha \hat q (v)} \lesssim \norm{v}^{-s-|\alpha|}$.
\item \cref{as:location-scale-assns}(3) imposes that $\Sigma(x)$ is close in operator norm
to some constant $\Sigma_0$ and similarly for the derivatives of $\Sigma(x)$.
\item \cref{as:location-scale-assns}(4) imposes that $b \in \K^M$ and its inverse $b^
{-1}$ is
Lipschitz. 
\end{itemize}

\subsubsection{Faithfulness verification}

\begin{prop}
 Assume that \cref{as:location-scale-assns} holds with sufficiently small $\psi$
 \eqref{eq:psi_bound}. Let $M,s$ be
 as in \cref{as:location-scale-assns}. Let
  \begin{enumerate}
    \item $\mathcal F$ be all integrable functions of $(\delta, P)$
    \item $\mathcal H$ be all integrable functions of $(S, P)$ invertible in $S$
    \item $\mathcal K = \K^{2s}$
  \end{enumerate}
  Then, with these choices, 
  \begin{enumerate}
    \item \cref{as:faithful_x} follows from \cref{as:complete_x,as:exogenous_x}
    \item The assumptions of \cref{prop:faith_x} are satisfied.
  \end{enumerate}
\end{prop}

\begin{proof}
  (1) follows immediately from \cref{thm:main} and \cref{as:complete_x},
  applied entrywise. For
  \cref{prop:faith_x}(1), observe that $b (X) = \E [a(\delta) \mid X, Z] \in \K^ {M}
  \subset \K^{2s} = \mathcal K$ 
  and $\sigma^{-1} \in
  \mathcal H$ with $T_0(\delta) = a(\delta)$. Likewise, \cref{prop:faith_x}(2) holds immediately.
\end{proof}

We now outline the remaining argument:
\begin{enumerate}
  \item \cref{thm:main} is the key result. Its proof argues that it is without loss to
  work with $\tilde \delta = \Sigma_0^{-1} \delta$, $\tilde X = \Sigma_0^{-1} b(X)$ such
  that $\tilde \delta = \tilde X + \tilde \Sigma (\tilde X) \epsilon$ and $\tilde \Sigma 
  (\tilde X)$ is centered around $I$. Moreover, it suffices to show that \[
    \E[u(\delta) \mid X] = k(X) \quad k \in \W^{2s,\infty}
  \]
  is solvable by some $u$. With these normalizations, \cref{thm:main_2} shows that there exists some $u(\delta)$
  such that $\E[u(\delta) \mid X ] = k(X)$, for scalar-valued functions $(u,k)$. The bulk
  of the argument justifies \cref{thm:main_2}.
  
  \item Suppose $\Sigma(x) = I$. Then $\E[u(\delta) \mid X] = (u \star q)(x)$ is a
  convolution. Thus we can take $u$ such that its Fourier transform is a ratio $\hat u = 
  \frac{\hat k} {\hat q}$. \Cref{lem:T0_inv} verifies that $u$ is a proper function. 
 
  \item Next, with $\Sigma(x) \neq I$, one could consider the deviation of the
  conditional expectation operator from the $\Sigma=I$ case: \[
    (\cT u)(x) \equiv (\cT_0 u)(x) + (\cE u)(x) = \br{\cT_0 \pr{\mathrm{Id} + \cT_0^{-1} \cE} u
    }(x)
  \]
  for $\cT$ the conditional expectation operator $\delta \mid X$, $\cT_0$ the
  conditional expectation operator under $\Sigma(x) = I$, and $\cE$ their difference. Now, we have the geometric
  expansion \[
    (\mathrm{Id} + \cT_0^{-1} \cE)^{-1} = \sum_{\ell=0}^\infty \pr{-\cT_0^{-1} \cE}^{\ell}
  \]
  upon verifying that the right-hand side converges. \Cref{lem:Delta_integral} derives the
  key condition for convergence.  Given a $k$, one can then construct $u$ as \[
    u = \pr{\sum_{\ell=0}^\infty \pr{-\cT_0^{-1} \cE}^{\ell}} \cT_0^{-1} k.
  \]

\item The key condition for convergence turns out to be a bound on the discrepancy in
Fourier space of the
integration kernels: $\hat q (\Sigma (x) \omega) - \hat q (\omega)$. The proof of 
\cref{thm:main_2} verifies that so long as $\Sigma(x)$ is sufficiently close to $I$ under
\cref{as:location-scale-assns}(3), this bound is attainable. 
\end{enumerate}

\begin{theorem}\label{thm:main}
Assume that \cref{as:location-scale-assns} holds with 
\[\psi < \psi^* \equiv \psi^*\left(J,s,C_{\pm},\lambda_{\min}(\Sigma_0), \|Db\|_{\W^
{M-1},\infty}, \|Db^{-1}\|_{L^\infty}\right) \numberthis \label{eq:psi_bound}\] 
for some constant $\psi^*$ that only depends on the parameters inside the parentheses.
Then for any $k\in \K^{2s}(\R^J, \R^J)$, there exists a proper function $u: \R^J \mapsto \R^J$ such that $\E[u
(\delta)\mid X] = k(X)$ almost surely. 
\end{theorem} 

\begin{proof}
We first note that it suffices to prove the claim for scalar-valued $k\in \K^{2s}(\R^J;
  \R)$ since we can concatenate the solutions. We argue that it is sufficient to prove
  \cref{thm:main_2}. The reason is that we can change variables and write $\tilde X$ for
  $b(X)$, $\tilde \epsilon$ for $\Sigma_0
  \epsilon$, and $\tilde \Sigma(\tilde X) = (\Sigma \circ b^{-1}(\tilde X)) \Sigma_0^
  {-1}$. Likewise, it is without loss to redefine $\delta$ as $a(\delta)$. 

  We need to show that, for sufficiently small choices of $\psi$, $\tilde \Sigma(\cdot) - I$ is
  suitably bounded by some $\tilde \psi$ so
  that \cref{as:location-scale-assns}(3) is satisfied for \cref{thm:main_2}. 
  \Cref{lem:tilde_Sigma} bounds $\tilde \Sigma(\cdot) - I$ in terms of $\Sigma(\cdot) -
  \Sigma_0$, where the bound depends only on the additional quantities $\|Db\|_{\W^
  {M-1,\infty}}, \|Db^ {-1}\|_ {L^\infty}, \lambda_{\min}^{-1}(\Sigma_0)$. Thus, there exists 
\[\psi^* := \psi^*\left(J,s,C_{\pm},\lambda_{\min}(\Sigma_0), \|Db\|_{\W^{M-1},\infty}, \|Db^{-1}\|_{L^\infty}\right)\] 
such that, if $\psi\le \psi^*$, we have that \begin{align*}
\|\tilde{\Sigma}(\cdot) - I\|_{\W^{M,\infty}} + \|\tilde{\Sigma}(\cdot) - I\|_{\W^{M,1}}\le \td{\psi}^*(J,s, C_{\pm}).
\end{align*}
required by \cref{thm:main_2}.

Now, having checked the conditions, we can apply \cref{thm:main_2} to $\tilde X$. Given $k
\in \K^ {2s}$, since
$b$ is invertible, we can write $
  k(X) = (k \circ b^{-1} \circ b)(X) = (k\circ b^{-1}) (\tilde X).
$ 
By \cref{lem:k_b_inv}, $k\circ b^{-1} \in \K^{2s}$. By definition, \[
  (k\circ b^{-1})(\tilde x) = A \tilde x + \tilde k(\tilde x) \quad \tilde k(\tilde x) \in
  \W^{2s, \infty}. 
\]
By \cref{thm:main_2}, we can then find $\tilde u(\delta)$ for which \[
  \E[\tilde u(\delta) \mid \tilde X] = \tilde k(\tilde X).
\]
Let $u(\delta) = A\delta + \tilde u(\delta)$, 
we then have \[
  \E[u(\delta) \mid \tilde X] = \E[u(\delta) \mid X] = (k\circ b^{-1})(\tilde X) = k(X).
\]
\end{proof}

\begin{theorem}\label{thm:main_2}
Suppose  \cref{as:location-scale-assns} holds with $a(\delta) = \delta, b(x) = x,
\Sigma_0 = I$ and
\[\psi < \psi^* \equiv \psi^*\left(J,s,C_{\pm}\right) \numberthis
\label{eq:psi_bound_2}\] 
for some constant $\psi^*$ that only depends on the parameters inside the parentheses,
defined in \eqref{eq:explicit_psi}.
Then for any $k\in \W^{2s, \infty}$, there exists a proper function $u$ such that
$\E [u (\delta)\mid X] = k(X)$ almost surely.
\end{theorem}

\begin{proof}

By \cref{lem:T0_inv,lem:Delta_integral}, it remains to prove 
\eqref{eq:Deltahat_integral}. 

\noindent \smallskip \textbf{Step 1: Bound $L^1$ norm of $\Delta_\omega(x)$ and its derivatives by
$(1+\norm{\omega}^2)^{-s/2}$.}

By the Taylor expansion,
\[\Delta_{\omega}(x) = \langle \nabla \hat{q}((\lambda \Sigma(x) + (1 - \lambda)I)\omega), (\Sigma(x) - I)\omega\rangle\]
for some $\lambda\in (0, 1)$ that depends on $x$. By \cref{as:location-scale-assns}(3),
$\psi < 1/2$ and
\[\|(\lambda \Sigma(x) + (1 - \lambda)I)\omega\|_2 \ge (1 - \psi) \|\omega\|_2\ge  \|\omega\|_2/2.\]
By \cref{as:location-scale-assns}(2), 
\begin{align*}
\|\nabla \hat{q}((\lambda \Sigma(x) + (1 - \lambda)I)\omega)\|_2&\le C_+(1 +  \|\omega\|_2^2/4)^{-s/2-1/2}\\
& \le C_+2^{s+1}(1 + \|\omega\|_2^2)^{-s/2-1/2}.
\end{align*}
Then 
\begin{align*}
|\Delta_\omega(x)|&\le C_+2^{s+1}(1 + \|\omega\|_2^2)^{-s/2-1/2}\cdot \|\Sigma(x) - I\|_{\op}\|\omega\|_2\\
& \le C_+2^{s+1} (1 + \|\omega\|_2^2)^{-s/2}\cdot \|\Sigma(x) - I\|_{\op}
\end{align*}
By \cref{as:location-scale-assns}(3),
$\|\Delta_{\omega}\|_{L^1} \le \psi C_+2^{s+1}(1 + \|\omega\|_2^2)^{-s/2}.$ 
By \cref{as:location-scale-assns}(2) and \cref{lem:hatq}, for any $\alpha\in \N^J$
with $0 < |\alpha|\le M$
\[\|D_x^\alpha\Delta_\omega\|_{L^1} = \|D_x^\alpha \hat{q}(\Sigma(x)\omega)\|_{L^1}\le \psi C_+ C_{J,s} (1 + \|\omega\|_2^2)^{-s/2}.\]

\noindent \smallskip \textbf{Step 2: Bound the Fourier transform $|\hat \Delta_\omega|$ with $L^1$
norm of $\Delta_\omega$ derivatives.}

By \cref{prop:fourier_bound},
\begin{align*}
|\hat{\Delta}_{\omega}(\upsilon)|&\le \psi \td{C} (1 + \|\omega\|_2^2)^{-s/2}(1 + \|\upsilon\|_2)^{-M}\\
& \le \psi \td{C} (1 + \|\omega\|_2^2)^{-s/2}(1 + \|\upsilon\|_2^2)^{-M/2},
\end{align*}
where $\td{C} = c_{J,M} C_+ C_{J,s}.$

\noindent \smallskip \textbf{Step 3: Bound the key condition \eqref{eq:Deltahat_integral}.}

By \cref{as:location-scale-assns}(2),
\begin{align*}
& \left|\int_{\R^J} |\hat{\Delta}_{\omega}(\omega - \upsilon)|\frac{|\hat{q}(\omega)|}{|
\hat{q}(\upsilon)|^2}d\omega\right|\\
& \le \psi \frac{\td{C}C_+}{C_-^2}\int_{\R^J} \frac{(1 + \|\omega - \upsilon\|_2^2)^{-M/2}(1 + \|\omega\|_2^2)^{-s}}{(1 + \|\upsilon\|_2^2)^{-s}}d\omega\\
& = \psi \frac{\td{C}C_+}{C_-^2}\int_{\R^J} \frac{(1 + \|\upsilon\|_2^2)^{s}}{(1 + \|\omega - \upsilon\|_2^2)^{M/2}(1 + \|\omega\|_2^2)^{s}}d\omega.
\end{align*}
Note that 
\[1 + \|\upsilon\|_2^2 \le 1 + 2(\|\omega\|_2^2 + \|\omega - \upsilon\|_2^2) \le 2(1 + \|\omega\|_2^2)(1 + \|\omega- \upsilon\|_2^2).\]
Thus, 
\begin{align*}
& \left|\int_{\R^J} |\hat{\Delta}_{\omega}(\omega - \upsilon)|\frac{|\hat{q}(\omega)|}{|
\hat{q}(\upsilon)|^2}d\omega\right|\\
& \le \psi \frac{2^{s}\td{C}C_+}{C_-^2}\int_{\R^J} (1 + \|\omega - \upsilon\|_2^2)^{-(\frac{M}{2}-s)}d\omega\\
& = \psi \frac{2^{s}\td{C}C_+}{C_-^2}\int_{\R^J} (1 + \|\upsilon\|_2^2)^{-(\frac{M}
{2}-s)}d\upsilon \tag{Change of variables}\\
& = \psi \frac{2^{s}\td{C}C_+}{C_-^2}\int_{\R^J} (1 + \|\upsilon\|_2^2)^{-(J+1)/2}d\upsilon\\
& = \psi \frac{2^{s}\td{C}C_+}{C_-^2}\frac{\pi^{(J+1)/2}}{\Gamma((J + 1)/2)}.
\end{align*}
This can be made less than $1$ because we can choose  
\[\psi < \td{\psi}^*(J,s,C_{\pm}) =\frac{C_-^2}{2^{s}c_{J,M} C_+ C_{J,s}C_+}\frac{\Gamma(
(J + 1)/2)}{\pi^{(J+1)/2}}. \numberthis \label{eq:explicit_psi}\]
The bound is uniform in $\omega$. Thus, \eqref{eq:Deltahat_integral} is proved. The proof is then completed.   
\end{proof}

\begin{lemma}\label{lem:T0_inv}
For $g: \R^J\mapsto \R^J$, let $\cT$ be the operator defined by
\[(\cT_0 g)(x) \equiv  \int_{\R^J} q(x - \delta)g(\delta)d\delta.\]
Under \cref{as:location-scale-assns}(2) and \cref{as:location-scale-assns}(3), $\W^
{2s,\infty}$ has a unique preimage in $\W^ {s,\infty}$ under $\cT_0$.
\end{lemma}

\begin{proof}
  By \cref{lem:symbol}, 
\begin{equation}\label{eq:T0_pdo}
(\cT_0 g)(x) = (2\pi)^{-J}\int_{\R^J} e^{i\langle x, \omega\rangle}\hat{q}(\omega)\hat{g}(\omega)d\omega.
\end{equation}
Note that $\cT_0$ is a multiplication operator, i.e., 
\[\widehat{\cT_0 g}(\omega) = \hat{q}(\omega) \hat{g}(\omega).\]
For any $k\in \W^{2s,\infty}$, let $u$ be the function with 
\begin{equation}\label{eq:T0_inv_fourier}
\hat{u}(\omega) = \frac{\hat{k}(\omega)}{\hat{q}( \omega)}.
\end{equation}
\Cref{as:location-scale-assns}(2) implies 
\begin{align*}
|\hat{u}(\omega)|& = \underbrace{\frac{|\hat{k}(\omega)|}{|\hat{q}^2(\omega)|}}_{\in
L^\infty(\R^J)}\cdot \underbrace{| \hat{q} (\omega)|}_{\in \W^{s,\infty}}.
\end{align*}
Thus, $u\in \W^{s,\infty}$. 
This also implies that $u$ is a
tempered distribution (which can be identified with a function). By Theorem 7.1.10 of \cite{hormander2003analysis}, $
\widehat{\cT_0 u} = \hat{k}\Longleftrightarrow \cT_0 u = k.
$
Thus, $
\cT_0^{-1}k \in \W^{s,\infty}, \quad \forall k\in \W^{2s,\infty}.$
\end{proof}

\begin{lemma}\label{lem:Delta_integral}
Let
$\cT$ be the operator \[
  (\mathcal T g)(x) = \int_{\R^J} q\pr{\Sigma(x)^{-1} (x-t)} g(t) \,dt = \E_{\delta = X +
  \Sigma(X) \epsilon, \epsilon \sim q(\cdot)} [g(\delta) \mid X=x]. \numberthis \label{eq:t_def}
\] 
For any $\omega\in \R^J$, let 
$\Delta_\omega(x) = \hat{q}(\Sigma(x)\omega) - \hat{q}(\omega).$
Assume that 
\begin{equation}\label{eq:Deltahat_integral}
\max_{\upsilon \in \R^J}\int_{\R^J} |\hat{\Delta}_{\omega}(\omega - \upsilon)|\frac{|
\hat{q}(\omega)|}{|\hat{q}(\upsilon)|^2}d\omega \le 1 - \eta
\end{equation}
for some constant $\eta > 0$. Under \cref{as:location-scale-assns}(1)--(3), $\W^{2s,
\infty}$ has a unique preimage in $\W^
{s, \infty}$ under $\cT$.
\end{lemma}

\begin{proof}
  
Define $\cE = \cT - \cT_0$. By \cref{lem:symbol} and \eqref{eq:T0_pdo}, for any $g\in \W^{s,\infty}$,  
\[(\cE g)(x) = (2\pi)^{-J}\int_{\R^J} e^{i\langle x, \omega\rangle}\Delta_\omega(x)\hat{g}(\omega)d\omega.\]
Then,
\begin{align}
\widehat{(\cE g)}(\upsilon) &= \int_{\R^J} e^{-i\langle x, \upsilon\rangle} \lb (2\pi)^{-J}\int_{\R^J} e^{i\langle x, \omega\rangle}\Delta_\omega(x)\hat{g}(\omega)d\omega \rb dx\nonumber\\
& = \int_{\R^J}\int_{\R^J} (2\pi)^{-J}e^{i\langle x, \omega - \upsilon\rangle}\Delta_\omega(x)\hat{g}(\omega)  dx d\omega\nonumber\\
& = \int_{\R^J} \hat{\Delta}_{\omega}(\omega - \upsilon)\hat{g}(\omega)d\omega.\label{eq:cEg}
\end{align}
By \cref{as:location-scale-assns}(2) and the standard definition of $\W^{s,\infty}$, for
any $k\in \W^{s,\infty} (\R^J)$,
\[\|k\|_{\W^{s,\infty}} < \infty \Longleftrightarrow \left\|\frac{\hat{k}}{
\hat{q}}\right\|_{L^\infty} \equiv \|k\|_{\hat{q}} < \infty.\]
Then we have
\[\W^{s,\infty}(\R^J) = \left\{k: \left\|k\right\|_{\hat{q}} < \infty\right\}.\]
By \eqref{eq:T0_inv_fourier}, we have that $\widehat{\mathcal T_0^{-1} \rho} = \hat \rho /
\hat q$ for any $\rho \in \W^{2s,\infty}$. 
We first verify that $\cE g\in \W^{2s,\infty}$. 
By \eqref{eq:cEg}, 
\begin{align}
\sup_{\upsilon\in \R^J}\frac{|\widehat{(\cE g)}(\upsilon)|}{|\hat{q}(\upsilon)|^2} & = \sup_{\upsilon\in \R^J} \int_{\R^J} \frac{|\hat{\Delta}_{\omega}(\omega - \upsilon)\hat{g}(\omega)|}{|\hat{q}^2(\upsilon)|}d\omega\nonumber\\
& = \sup_{\upsilon\in \R^J} \int_{\R^J} |\hat{\Delta}_{\omega}(\omega - \upsilon)|\frac{|\hat{q}(\omega)|}{|\hat{q}(\upsilon)|^2} \frac{|\hat{g}(\omega)|}{|\hat{q}(\omega)|}d\omega\nonumber\\
& = \sup_{\upsilon\in \R^J} \int_{\R^J} |\hat{\Delta}_{\omega}(\omega - \upsilon)|\frac{|\hat{q}(\omega)|}{|\hat{q}(\upsilon)|^2} d\omega \cdot \left\|\frac{\hat{g}}{\hat{q}}\right\|_{L^\infty}\nonumber\\
& \le (1 - \eta) \|g\|_{\hat{q}}\label{eq:1-eta}\\
& < \infty,\nonumber
\end{align}
where \eqref{eq:1-eta} uses \eqref{eq:Deltahat_integral} and the last line follows from $g\in \W^{s,\infty}$. 
Applying \eqref{eq:T0_inv_fourier},
\begin{align*}
 \|\cT_0^{-1}\cE g\|_{\hat{q}} &= \left\|\frac{\widehat{\pr{\cT_0^{-1}\cE g}}}{
 \hat{q}}\right\|_{L^\infty}\\
& = \sup_{\upsilon\in \R^J} \int_{\R^J} \frac{|\hat{\Delta}_{\omega}(\omega - \upsilon)\hat{g}(\omega)|}{|\hat{q}^2(\upsilon)|}d\omega\\
& \le (1 - \eta) \|g\|_{\hat{q}},
\end{align*}
where the last line follows from \eqref{eq:1-eta}. This implies that 
\[\|\cT_0^{-1}\cE\|_{\|\cdot\|_{\hat{q}} \leftrightarrow \|\cdot\|_{\hat{q}}} = \|\cT_0^{-1}\cE\|_{\|\cdot\|_{\W^{s,\infty}} \leftrightarrow \|\cdot\|_{\W^{s,\infty}}} \le 1 - \eta.\]
As a result, the following Neumann series converges on $\W^{s,\infty}$:
\[\sum_{\ell\ge 0}(- \cT_0^{-1}\cE)^{\ell}.\]
By \cref{lem:T0_inv}, for any $k\in \W^{2s,\infty}$, $\cT_0^{-1}k\in \W^{s,\infty}$. Let 
\[g = \sum_{\ell\ge 0}(- \cT_0^{-1}\cE)^{\ell}\cT_0^{-1}k.\]
Then $\|g\|_{\W^{s,\infty}} \le \eta^{-1}\|\cT_0^{-1}k\|_{\W^{s,\infty}} <
\infty$ and
\[(\mathrm{Id} + \cT_0^{-1}\cE)g = \cT_0^{-1}k.\]
We thus have that 
\[\cT g= (\cT_0 + \cE)g = \cT_0(\mathrm{Id} + \cT_0^{-1}\cE)g = \cT_0 \cT_0^{-1} k = k.\]

\end{proof}

\begin{lemma}\label{lem:hatq}
Under \cref{as:location-scale-assns}(2)--(3), 
for any $\alpha\in \N^J$
with $|\alpha|\le M$ and
$\omega\in \R^J$,
\[\sup_{x\in \R^J}|D_x^{\alpha}\hat{q}(\Sigma(x)\omega)| + \int_{\R^J}|\partial_x^{\alpha}\hat{q}(\Sigma(x)\omega)|dx\le \psi C_+ C_{J,s} (1 + \|\omega\|_2^2)^{-s/2},\]
where 
\[C_{J,s} = \max_{\alpha: |\alpha|\le M}2^{s+1+|\alpha|}
\sum_{\kappa\in \N^J, 1\le |\kappa|\le |\alpha| }
\sum_{\substack{(k_{\gamma,j}) \in \mathbb{N} \\
\sum_{\gamma} k_{\gamma,j} = \kappa_j, \,\forall j \\
\sum_{\gamma,j} k_{\gamma,j}\,\gamma = \alpha}}
\frac{\alpha!}{\displaystyle\prod_{\gamma} k_{\gamma,j}!\,(\gamma!)^{k_{\gamma,j}}}.\]
\end{lemma}

\begin{proof}
  
Fix $\omega\in \R^J$. Let $g(x) = \Sigma(x)\omega$. For any $\gamma$ with $1\le |\gamma|\le |\alpha|$,
\[D_x^{\gamma}g(x)= \lb D_x^{\gamma}\Sigma(x)\rb\omega.\]
Thus, 
\[\|D_x^{\gamma}g(x)\|_{\op}\le \|D_x^{\gamma}\Sigma(x)\|_{\op} \|\omega\| \]
Furthermore, by \cref{as:location-scale-assns}(3),
\[|D_\upsilon^\kappa \hat{q}(\Sigma(x)\omega)| \le C_+(1 + \|\Sigma(x)\omega\|_2^2)^{-s/2 - |\kappa|/2}.\]
By \cref{as:location-scale-assns}(2), 
\[\|\Sigma(x)\|_{\op}\ge 1 - \psi \ge \frac{1}{2}.\]
Then 
\begin{align*}
|D_\upsilon^\kappa \hat{q}(\Sigma(x)\omega)| &\le C_+(1 + \|\omega\|_2^2/4)^{-s/2 - |\kappa|/2}\\
& \le C_+2^{s+|\kappa|}(1 + \|\omega\|_2^2)^{-s/2 - |\kappa|/2}\\
& \le C_+2^{s+|\alpha|}(1 + \|\omega\|_2^2)^{-s/2 - |\kappa|/2}\\
& \triangleq C_1(1 + \|\omega\|_2^2)^{-s/2 - |\kappa|/2}
\end{align*}
By the multivariate Fa\`{a} di Bruno formula (\cref{prop:faa_di_bruno}), 
\begin{align*}
& |D_x^{\alpha}\hat{q}(\Sigma(x)\omega)| \\
& \le \sum_{\kappa\in \N^J, 1\le |\kappa|\le |\alpha|}|D_\upsilon^\kappa \hat{q}(\Sigma(x)\omega)|
\sum_{\substack{(k_{\gamma, j}) \in \mathbb{N} \\
\sum_{\gamma} k_{\gamma,j} = \kappa_j, \,\, \forall j\\
\sum_{\gamma, j} k_{\gamma, j}\,\gamma = \alpha}}
\frac{\alpha!}{\displaystyle\prod_{\gamma, j} k_{\gamma, j}!\,(\gamma!)^{k_{\gamma, j}}}
\prod_{\gamma, j} |D_x^{\gamma}g_j(x)|^{\,k_{\gamma, j}}\\
& \le C_+2^{s+|\alpha|}\sum_{\kappa\in \N^J, 1\le |\kappa|\le |\alpha|}(1 + \|\omega\|_2^2)^{-s/2-|\kappa|/2}\\
& \qquad 
\sum_{\substack{(k_{\gamma, j}) \in \mathbb{N} \\
\sum_{\gamma} k_{\gamma, j} = \kappa_j, \, \forall j \\
\sum_{\gamma, j} k_{\gamma, j}\,\gamma = \alpha}}
\frac{\alpha!}{\displaystyle\prod_{\gamma, j} k_{\gamma, j}!\,(\gamma!)^{k_{\gamma, j}}}
\prod_{\gamma, j} (\|D_x^{\gamma}\Sigma(x)\|_{\op}\|\omega\|_2)^{\,k_{\gamma,j}}\\
& \le C_+2^{s+|\alpha|} (1 + \|\omega\|_2^2)^{-s/2}\sum_{\kappa\in \N^J, 1\le |\kappa|\le |\alpha|} (1 + \|\omega\|_2^2)^{-|\kappa|/2}\|\omega\|_2^{|\kappa|}\\
& \qquad \cdot \sum_{\substack{(k_{\gamma,j}) \in \mathbb{N} \\
\sum_{\gamma} k_{\gamma,j} = \kappa_j, \,\forall j \\
\sum_{\gamma, j} k_{\gamma,j}\,\gamma = \alpha}}
\frac{\alpha!}{\displaystyle\prod_{\gamma} k_{\gamma,j}!\,(\gamma!)^{k_{\gamma,j}}}
\prod_{\gamma,j} \|D_x^{\gamma,j}\Sigma(x)\|_{\op}^{\,k_{\gamma,j}}\\
& \le C_+2^{s+|\alpha|} (1 + \|\omega\|_2^2)^{-s/2}\sum_{\kappa\in \N^J, 1\le |\kappa|\le |\alpha| }
\sum_{\substack{(k_{\gamma,j}) \in \mathbb{N} \\
\sum_{\gamma} k_{\gamma,j} = \kappa_j, \,\forall j \\
\sum_{\gamma,j} k_{\gamma,j}\,\gamma = \alpha}}
\frac{\alpha!}{\displaystyle\prod_{\gamma,j} k_{\gamma,j}!\,(\gamma!)^{k_{\gamma,j}}}
\prod_{\gamma,j} \|D_x^{\gamma}\Sigma(x)\|_{\op}^{\,k_{\gamma,j}}.
\end{align*}
Let
\[C_1 = \sum_{\kappa\in \N^J, 1\le |\kappa|\le |\alpha| }
\sum_{\substack{(k_{\gamma,j}) \in \mathbb{N} \\
\sum_{\gamma} k_{\gamma,j} = \kappa_j, \,\forall j \\
\sum_{\gamma,j} k_{\gamma,j}\,\gamma = \alpha}}
\frac{\alpha!}{\displaystyle\prod_{\gamma} k_{\gamma,j}!\,(\gamma!)^{k_{\gamma,j}}}.\]
By \cref{as:location-scale-assns}(3), 
\[\prod_{\gamma,j} \|D_x^{\gamma}\Sigma(x)\|_{\op}^{\,k_{\gamma,j}}\le \psi^{|\kappa|}\le \psi.\]
Thus, 
\[|D_x^{\alpha}\hat{q}(\Sigma(x)\omega)| \le \psi C_+ 2^{s+|\alpha|} C_1 (1 + \|\omega\|_2^2)^{-s/2}.\]

 Similarly, to bound the $L^1$ norm of $D_x^{\alpha}\hat{q}(\Sigma(x)\omega)$, we just need to show that 
\[\left\|\prod_{\gamma,j} \|D_x^{\gamma}\Sigma(x)\|_{\op}^{\,k_{\gamma,j}}\right\|_{L^1}\le \psi.\]
This is because at least one $k_\gamma \ge 1$ and $\|ab\|_{L^1}\le \|a\|_{L^1}\|b\|_{L^\infty}$. The proof is then completed by setting $C_{J,s} = 2^{s+|\alpha|+1}C_+ C_1$.
\end{proof}

\subsection{Verifying completeness from faithfulness}

\begin{prop}
  Let $\mathcal F$ be all integrable functions of $(\delta, P)$ and $\mathcal K$ be all
  integrable functions of $X$. Suppose that for all $\E[|f(\delta)|] < \infty$,  $\E[f
  (\delta) \mid X] = 0$ implies $f = 0$. Then \cref{as:faithful_x} implies 
  \cref{as:complete_x}. 
\end{prop}

\begin{proof}
  Let $h$ be an integrable and (without loss) scalar function of $\delta, P$. Thus 
  \cref{as:faithful_x} implies that \[\E[h \mid X, Z] = 0 \in \mathcal K \implies h
  (\delta, P) = h
  (\delta).\]
We further have that $\E[h(\delta) \mid X] = 0$ implies $h(\delta) = 0$. Thus $h(\delta,
P) = 0$ and hence $(\delta, P) \mid (X, Z)$ is complete. Since any integrable function of
$S,P$ can be rewritten as an integrable function of $(\delta, P)$, we have that $(S, P)
\mid (X, Z)$ is complete as well.
\end{proof}

\end{appendix}
\newpage

\begin{appendix}

\begin{center}
\textbf{Online Appendix for ``Nonparametric Identification of Demand \\ without Exogenous
Product
Characteristics''}

Borusyak, Chen, Hull, and Lei
\medskip

February 2026
\end{center}
\vspace{-1cm}

\setcounter{section}{1}

\DoToC

\medskip

\section{Additional results and discussions}

\bigskip

\subsection{Faithfulness in causal graphs}
\label{sub:faithfulness_in_causal_graphs}

Given a candidate $H(\delta, P)=H$, consider the induced directed
acyclic graph (DAG) of the variables $(X,P,Z,\delta,H)$. Depending on whether or not $H(\delta,P)=H(\delta)$, there are two possible DAGs shown in \cref{fig:dag}.\footnote{Note that $H$ is here a deterministic function of $(\delta,P)$ while each node in a DAG is typically assumed to have its own independent random variation.} 
In the causal discovery literature \citep{spirtes2000causation}, the joint distribution of $(X, P, Z, \delta, H) \sim Q_H$ is said to be \emph{faithful} to a  graph if, whenever two variables $V_1$ and $V_2$  are conditionally independent given some set of variables $C$, the correponding  vertices in the DAG are \emph{d-separated} by a set of vertices corresponding to $C$.
 
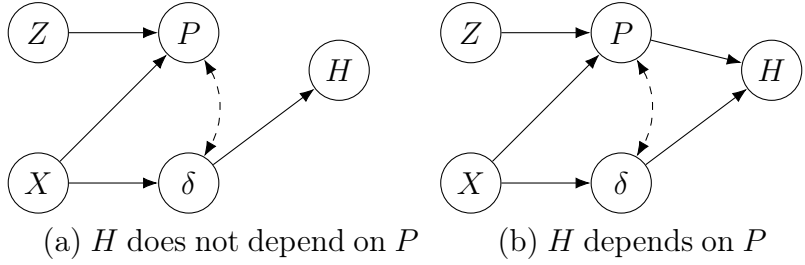
\begin{figure}[!htb]
  {
\begin{center}
  \begin{tikzpicture}[
  >={Latex[length=2mm]},
  every node/.style={circle, draw, minimum size=8mm, inner sep=0pt}
]

\node (Z) at (0,-0.5)   {$Z$};
\node (X) at (0,-2.5)   {$X$};
\node (P) at (2,-0.5)   {$P$};
\node (d) at (2,-2.5)   {$\delta$};
\node (H) at (4,-1.0)   {$H$};

\draw[->] (X) -- (P);     %
\draw[->] (Z) -- (P);     %
\draw[->] (X) -- (d);     %
\draw[->] (d) -- (H);     %

\draw[<->, dashed, bend left=30] (P) to (d);
\end{tikzpicture}\hspace{2em}
\begin{tikzpicture}[
  >={Latex[length=2mm]},
  every node/.style={circle, draw, minimum size=8mm, inner sep=0pt}
]

\node (Z) at (0,-0.5)   {$Z$};
\node (X) at (0,-2.5)   {$X$};
\node (P) at (2,-0.5)   {$P$};
\node (d) at (2,-2.5)   {$\delta$};
\node (H) at (4,-1.0)   {$H$};

\draw[->] (X) -- (P);     %
\draw[->] (Z) -- (P);     %
\draw[->] (X) -- (d);     %
\draw[->] (P) -- (H);     %
\draw[->] (d) -- (H);     %

\draw[<->, dashed, bend left=30] (P) to (d);
\end{tikzpicture}

(a) $H$ does not depend on $P$ \hspace{2em} (b) $H$ depends on $P$
\end{center}}
\caption{DAG for $(X,P,Z,\delta, H(\delta, P))$.}
\label{fig:dag}

\end{figure}

Suppose $Q_H$ is known to be faithful to one of the two graphs, (a) or (b). Then, upon observing $Z
 \indep H
\mid X$, we can infer that $Q_H$ is faithful to (a) and not (b), meaning $H$ is constant in $P$. This is similar to the type of inference that \cref{as:faithfulness} allows.

In this sense, \cref{as:faithfulness} essentially requires that every $Q_H$ is faithful to
one of (a) and (b), depending on whether $H$ varies with $P$---although 
\cref{as:faithfulness} only speaks to whether the edge $P \to H$ is faithfully reflected by
 $Q_H$.\footnote{\cref{as:faithfulness} is stronger in that it precludes the dependence of $H$ on $P$ from its \emph{mean}-independence with $Z$ given $X$ only, though our general identification
 argument accommodates the analogue of \cref{as:faithfulness} with full independence
 instead.}

\subsection{Faithfulness verification without proxies}
\label{sec:no_x}

\subsubsection{Faithfulness verification from functional form}
Here we consider identification under the added restriction that $P$ combines linearly with the utility index, and show that this suffices for faithfulness (given a sufficiently strong instrument $Z$) even absent variation in $X$.  We interpret these results as faithfulness verification without a proxy for $\delta$. The results follow identically to Proposition 1 in \cite{borusyak2025estimating}.\footnote{The difference relative to \cite{borusyak2025estimating} is that we maintain the index restriction on $X$.}

Consider a version of \Cref{eq:demand} where $\delta$ and $P$ enter additively\[
  S = \sigma(\delta(\bar{X},\xi) + P, \tilde X).
\]
We consider an extreme case where $X$ has no variation and drop it (in addition to $\tilde{X}$, as before) from the model: $S = \sigma(\delta + P)$. This is equivalent to an analysis which conditions fully on $\bar X$. Under this model, we have that $\delta =
\sigma^{-1}(S) - P $. We can then consider all functions $h(S,P) = h(S) -P$ that are mean-independent of $Z$. For any such $h$, we have that $h(S,P) = H(\delta, P) = f
(P+\delta) - P$ for $f = h(\sigma(\cdot))$.

Faithfulness here says that for all $H(\delta,P) = f(P+\delta) - P$, \[
  \E[H(\delta,P) \mid Z] = 0 \implies H(\delta, P) = f(P + \delta) - P = H(\delta). 
\]
Suppose that $Z$ is a strong instrument in the sense that $(P+\delta) \mid Z$ is complete.\footnote{Since we condition on
$\bar X$, this should be read as conditional completeness given $\bar X$.} Then \[
  \E[f(P+\delta) \mid Z] = \E[P \mid Z]
\]
has a unique solution in $f$. Since $f(t) = \sigma^{-1}(t) - \E[\delta]$ is such a
solution, we conclude that it is the unique solution. The corresponding $H$ satisfies \[
  H(\delta, P) = \sigma^{-1}(P+\delta) - P - \E[\delta] = \delta - \E[\delta],
\]
and is thus solely a function of $\delta$. Thus, in this example, faithfulness follows without any variation in $X$. This is because the functional form restriction on how $\delta$ and $P$ interact in $\sigma$ suffices to rule out the kind of complex interactions between $\delta$ and $P$ that a strong proxy $X$ is otherwise helpful for ruling out.

\subsubsection{Identification in triangular models as faithfulness verification}
\label{sub:triangular_model}

Here we connect faithfulness to the literature on nonparametric identification of triangular models: i.e., \citet{imbensnewey},
\citet{torgovitsky2015identification}, and \citet{d2015identification}. We show that arguments in these papers can be viewed as verifying
faithfulness without a strong proxy $X$ by imposing instead that
outcomes are scalar and monotone in $\delta$. 

Consider a special case of our model where $J = 1$ and let \[
  S = \sigma(\delta, P), \quad P = f(Z, \omega).
\]
Suppose that $\sigma(\cdot, P)$ is strictly increasing and $f(Z, \cdot)$ is strictly
increasing (and their respective arguments are one-dimensional). Normalize the marginal
distribution of $\omega$ to $\Unif[0,1]$. Since $\omega$ is the conditional quantile of $P
\mid Z$, we can assume that $\omega$ is observed. Lastly, assume that $\delta$ is
continuously distributed and $\delta \mid \omega$ is also continuously distributed.

Under these restrictions, the identified set for $\sigma$ is (Theorem 1,
\citet{torgovitsky2015identification}): \[
  \Theta_I = \br{h^{-1}(\delta, p) \text{ strictly increasing in $\delta$} : (h(S, P),
  \omega)
  \indep Z}.
\]
Analogous to \cref{lemma:main}, for identifying price counterfactuals, it suffices to show
that, for $h^ {-1}$ strictly increasing in $\delta$ and $H(\delta, P) \equiv h(\sigma
(\delta, P), P)$, the following analogue of faithfulness holds: \[ 
(H(\delta, P), \omega) \indep Z \implies
H(\delta, P) 
\text{ does not depend on $P$. } \numberthis \label{eq:triangular_faithful}
\]

Let $H = H(\delta, P)$. \citet{imbensnewey} observe that \[
  F_{H \mid P, \omega}(t \mid p, \omega) = \int \one(H(\delta, p) \le t) \,dF_{\delta
  \mid \omega} ( \delta
  \mid \omega). 
\]
If the support of $\omega \mid P$ is $[0,1]$ for all $P$ (Assumption 2, 
\citet{imbensnewey}), we can
then integrate the above
with respect to the marginal distribution of $\omega$: \[
  \int_0^1 F_{H \mid P, \omega}(t \mid p, \omega) d\omega = \int \one(H(\delta, p) \le t)
  \, dF_{\delta}(\delta) = \P_{\delta \sim F_\delta}(H(\delta, p) \le t) = F_\delta\pr{H^
  {-1}(t, p)}.
  \numberthis \label{eq:imbens_newey_eq}
\]
Note that since $(H, \omega) \indep Z$ and $P = f(Z, \omega)$, we have that \[H \indep P
\mid \omega.\] Thus $F_{H \mid P, \omega}(t \mid p, \omega) = F_{H \mid \omega}(t \mid
\omega)$ does not depend on $p$. Thus, the left-hand side of \cref{eq:imbens_newey_eq}
does not depend on $p$. The only way for this to occur is if $H^{-1}(t,p)$ does not depend
on $p$ either, since $F_\delta$ is strictly increasing. This argument is in equation
(5)--(7) of \citet{imbensnewey} preceding their Theorem 3. We can thus view this argument
as verifying the version of faithfulness \eqref{eq:triangular_faithful}.

\citet{torgovitsky2015identification} avoids the need in \eqref{eq:imbens_newey_eq} to
integrate over $\omega$,\footnote{A similar argument appears in \citet{d2015identification}.} thus allowing for
instruments with few support
points.\footnote{$\omega$ is the quantile of $P$ in the $P\mid Z$ distribution. If there
are only finitely many values for $Z$, then each price value has only finitely many
quantiles that can be candidates for $\omega$, meaning that $\omega \mid P$ cannot have full
support.} Suppose that $Z$ instead takes two values $ \br{z_0, z_1}$. Let $\pi (p) =
f(z_1, f^{-1}(p,z_0))$ such that for some identical value of $\omega$, \[
  \pi(p) = f( z_1, \omega), \quad p = f(z_0, \omega).
\]
\citet{torgovitsky2015identification} observes that, since $Z \indep (\delta, \omega)$, $H\indep P\mid \omega$ and hence
\[
  \P(\delta \le H^{-1}(t, \pi(p)) \mid \omega) = F_{H\mid P, \omega}(t \mid \pi(p), \omega)
  = F_{H \mid P, \omega} (t \mid p, \omega) = \P(\delta \le H^{-1}(t, p) \mid \omega).
\]
This implies that for any $p$, for all $t$,  \[
  H^{-1}(t, \pi(p)) = H^{-1}(t, p). 
\]
\citet{torgovitsky2015identification} imposes further assumptions that ensure that there
is some value $p_0$ such that for any value $p$, the sequence $\pi(p), \pi(\pi(p)),
\ldots$ approaches
$p_0$.\footnote{One example is to assume that $p_0$ is the left support end point of $P$
and that $f(z_1, \omega) < f(z_0, \omega)$ for all $\omega > 0$ with $p_0 = f(z_0, 0) =
f(z_1, 0)$. Then we have $\pi (p) <
p$ approaching $p_0$.} Thus, \[
   H^{-1}(t, \pi(p)) =  H^{-1}(t, \pi(\pi(p))) = H^{-1}(t, \pi(\pi(\pi(p)))) = \cdots
\]
equals $H^{-1}(t, p_0)$ given continuity. This ensures $H^
{-1}(t, p) = H^{-1}(t, p_0)$, proving \eqref{eq:triangular_faithful}. 

\subsection{Example of faithfulness without completeness}
\label{sub:fwithoutC}
Combining restrictions on $P$  and $\delta$  analogous to those in  \cref{subsec:model-p-lambda,subsec:model-delta}, we can construct a broad class of DGPs under which faithfulness holds yet $\delta\mid X$ can be incomplete. In particular, this implies the failure of joint completeness of $(P, \delta)$ given $(X, Z)$ (i.e., \cref{as:completeness}). 

We use \cref{as:exogenous,as:p-index_x}, and the following additional assumptions. To simplify notation, we assume that $\tilde{X}$ is conditioned on and treated as fixed. We also write $\Lambda$ for $\lambda = \lambda(X,Z)$ defined in \cref{as:p-index_x} to distinguish between the random variable $\lambda$ and the value that $\lambda$ takes.

\begin{as}\label{ass:location_model}
There exists a measurable $a:\R^J\to\R^J$ and a noise vector $\epsilon\in\R^J$ such that
\[
\delta=a(X)+\epsilon,\qquad \epsilon\perp (X,\Lambda).
\]
Assume $\check{X}=a(X)$ has a positive density everywhere and $(-\epsilon)$ has a density $p_\epsilon$ with the characteristic function $
\varphi_\epsilon(\omega):=\E[e^{i\langle \omega, \epsilon\rangle}], \, \omega\in\R^J$. 
\end{as}

\begin{as}\label{ass:P}
Let $q(p\mid \lambda,d)$ denote the conditional density of $P\mid(\Lambda=\lambda,\delta=d)$.
Fixing $\lambda_0\in\R^k$,

\begin{enumerate}[label=(\arabic*)]
\item There exists a reference density $\mu(p)>0$ and a constant $c_->0$ such that
\[
q(p\mid \lambda,d)\ \ge\ c_-\,\mu(p)\qquad \forall (p,\lambda',d)\in \R^{3J}.
\]
\item There exists a measurable function $\zeta:\R^J\to[0,\infty)$ such that
\[
\bigl|q(p\mid \lambda,d)-q(p\mid \lambda_0,d)\bigr|\ \le\ \zeta(d)\,\mu(p)\qquad \forall(p,\lambda,d).
\]
\item Let $p_\delta$ be the marginal density of $\delta$.
There exists $a_0>0$ such that
\[\sup_{d\in\R^J}\frac{e^{a_0\norm{d}}\zeta(d)}{p_\delta(d)}\ <\ \infty.
\]
\end{enumerate}
\end{as}

\begin{as}[Completeness of $P\mid \Lambda$ given $\delta$]
  \label{ass:complete}
  For any function $g: \R^{2J}\mapsto \R$,
  \[\E[g(\delta,P)\mid \delta,\Lambda]= 0\quad\text{a.s.}\Longrightarrow g(\delta,P)=0\quad\text{a.s.}\]
\end{as}

\begin{rmk}\label{rmk:necessity}
Under \cref{as:p-index_x}, \cref{ass:complete} is necessary for faithfulness, such that we retain the need for $Z$ to be a strong instrument for price. We prove this result in \Cref{subsec:proof_necessity}. 
\end{rmk}

\begin{theorem}\label{thm:unrestricted-faithfulness}
Under \Cref{as:exogenous,as:p-index_x,ass:location_model,ass:P,ass:complete}, $(\delta,P)\mid (X,\Lambda)$ is faithful. Furthermore, if these exists $\omega_0\in \R^J$ such that $\varphi_\epsilon(\omega_0) = 0$, then $\delta\mid X$ is not complete and hence  \Cref{as:complete_x} fails.
\end{theorem}

\noindent The proof is deferred to \Cref{subsec:proof_unrestricted-faithfulness}. 

We next provide a concrete class of examples that satisfies \Cref{ass:location_model,ass:P,ass:complete}, with the proof presented in \Cref{subsec:proof_example}.

\begin{prop}\label{prop:example}
Assume that 
\begin{enumerate}[label=(\alph*)]
\item The conditional density of $P$ is given by
\begin{equation}\label{eq:qdef}
q(p\mid \lambda,d)
=\exp\big\{\langle \theta(\lambda,d), S(p)\rangle-A(\theta(\lambda,d))\big\}\,\mu(p), \,\, \theta(\lambda,d)=\zeta(d) t(\lambda)
\end{equation}
where $S(p)$ is bounded, non-constant, and continuous, $\mu(p)$ is a carrier measure with positive density everywhere, 
\[
\zeta(d)=e^{-c\|d\|^2}, \quad c > 1/4,
\]
and $A(\theta)$ is the log-normalizer
\[
A(\theta)=\log\int_{\mathbb{R}^J}\exp\big(\langle \theta,  S(u)\rangle\big)\,\mu(u)\,du.
\]
\item $X\sim N(0, I_J), \nu\sim N(0,I_J)$ and $\eta=(\eta_1,\dots,\eta_J)$ with $\eta_j\stackrel{\mathrm{indep.}}{\sim}\mathrm{Unif}([-b_j,b_j])$ for some $b_j > 0$. Assume $(\nu,\eta)$ are independent of $(X,\Lambda)$ and 
\[
\epsilon=\nu+\eta,\qquad \delta=X+\epsilon.
\]
\end{enumerate}
Then \Cref{ass:location_model,ass:P,ass:complete} hold. Moreover, $\delta\mid X$ is incomplete.
\end{prop}

\begin{rmk}\label{rem:fragility}
Note that $b_j$ can be arbitrarily small in (b). \Cref{prop:example} then shows that, under the $\lambda$-index and additive $\delta$ model, completeness is sensitive to infinitesimal perturbations in $\delta$ while faithfulness is not. 
\end{rmk}

\subsection{Example of completeness without faithfulness}
\label{sub:failure_of_faithfulness}

We assume $J=1$ for this example, though analogous constructions exist for $J > 1$.
Suppose $Z, X > 0$ and \[ P \mid \delta, Z, X \sim \Norm(r(\delta, X) \tau
(\delta, Z, X), \tau^2(\delta, Z, X)).
\]
This construction is robust to taking monotone transformations of $(\delta, P)$, and
thus we may view $P$ as log price instead---for instance---if we would like price to be
strictly positive. Let $\delta$ be supported on $[1,2]$ with PDF \[
f(\delta \mid x) = \exp(r(\delta, x)^2/2) a(x) \one(\delta \in [1,2])
\]
for \[
  a(x)^{-1} = \int_1^2 e^{r(\delta, x)^2/2} \,d \delta. 
\]

\begin{prop}\label{prop:comp_not_faith}
  Let $r(\delta, x) = x \delta$ and $\tau(\delta, z, x) = z$. Then the above construction
  specifies $(\delta, P) \mid (Z, X)$  for which \cref{as:completeness,as:exogenous}
  hold but \cref{as:faithfulness} fails. 
\end{prop}

We present the proof in  \Cref{subsec:proof_comp_not_faith}.

\subsection{Bertrand--Nash pricing}\label{sub:BertrandNash}
Here we show how Bertrand--Nash pricing implies \Cref{eq:bertrand-nash}; see also Appendix A in \cite{berryhaile}. Consider the case of single-product firms.\footnote{The argument extends to multi-product firms if the ``ownership matrix'' is conditioned upon in the same way characteristics $\tilde X$ are.} As is well-known, the first-order condition for profit maximization of firm $j$ with a constant marginal cost $C_j$ is
$$\sigma_j(\delta,P)+(P_j-C_j)\frac{\partial \sigma_j(\delta,P)}{\partial P_j}=0.$$
Provided the equilibrium is unique, the solution for the vector of $P$ can therefore be written as a function of $\delta$ and $C$: i.e.,  \Cref{eq:bertrand-nash} holds.

\section{Auxiliary proofs}
\subsection{Auxiliary lemmas for \cref{thm:main}}

\begin{lemma}\label{lem:B2}
If $\hat{q}(\upsilon) = (1 + \|\upsilon\|_2^2)^{-s/2}$ for some $s > J$, then 
\cref{as:location-scale-assns}(2) holds.

\end{lemma}

\begin{rmk}
\label{rmk:radial_gamma}
Since $\hat{q}(0) = 1$, $\hat{q}$ is the characteristic function of a density. In fact, for any $s > J$, $q$ has the following expression
\[q(x) = \Theta_{J,s}\|x\|_2^{(s - J)/2}K_{(s-J)/2}(\|x\|_2)\]
where $K_{v}$ is the modified Bessel function of the second kind and $\Theta_{J,s}$ is the normalizing constant. For large $x$, 
\[q(x) \sim \|x\|_2^{(s-  J -1)/2}\exp\{-\|x\|_2\}.\]
Thus, $q$ behaves like a radial Gamma distribution.
\end{rmk}

\begin{proof}
We will prove that
\[D_\upsilon^\alpha \hat{q}(\upsilon) = \sum_{\ell=1}^{|\alpha|}(1 + \|\upsilon\|_2^2)^{-s/2-\ell}F_{\alpha, \ell}(\upsilon),\]
where $F_{\alpha, \ell}$ is a homogeneous polynomial of $\upsilon$ of order $\ell$. We prove this claim by induction on $|\alpha|$. When $\ell = 1$, 
\[D_\upsilon^{e_j} \hat{q}(\upsilon) = (1 + \|\upsilon\|_2^2)^{-s/2-1}\cdot (-s\upsilon_j)\]
when $e_j$ is the $j$-th canonical basis. Suppose the claim holds for $|\alpha|-1$. For any given $\alpha$, assume WLOG that $\alpha_1 \ge 1$. Let $\td{\alpha} = \alpha - e_1$. Then 
\begin{align*}
D_\upsilon^{\alpha}\hat{q}(\upsilon) &= D_\upsilon^{e_1}D_\upsilon^{\td{\alpha}}\hat{q}(\upsilon) = D_\upsilon^{e_1}\sum_{\ell=1}^{|\td{\alpha}|}(1 + \|\upsilon\|_2^2)^{-s/2-\ell}F_{\td{\alpha}, \ell}(\upsilon)\\
& = \sum_{\ell=1}^{|\td{\alpha}|}(1 + \|\upsilon\|_2^2)^{-s/2-\ell-1}\cdot -(s+2\ell)\upsilon_1 F_{\td{\alpha}, \ell}(\upsilon)\\
& \qquad + \sum_{\ell=1}^{|\td{\alpha}|}(1 + \|\upsilon\|_2^2)^{-s/2-\ell}\cdot D_\upsilon^{e_1} F_{\td{\alpha}, \ell}(\upsilon)\\
& = \sum_{\ell=1}^{|\td{\alpha}|}(1 + \|\upsilon\|_2^2)^{-s/2-\ell}\cdot \lb D_\upsilon^{e_1} F_{\td{\alpha}, \ell}(\upsilon)-(s+2\ell-1)\upsilon_1 F_{\td{\alpha}, \ell-1}(\upsilon)\rb,
\end{align*}
where $F_{\td{\alpha},0}(\upsilon) = 0$ for notational convenience. 
Clearly, $D_\upsilon^{e_1} F_{\td{\alpha}, \ell}(\upsilon)-(s+2\ell-1)\upsilon_1 F_{\td{\alpha}, \ell-1}(\upsilon)$ is a homogeneous polynomial of order $\ell$. Thus, the claim holds for $|\alpha|$. By induction the proof is completed.
\end{proof}

\begin{lemma}[Multivariate Fa\`{a} di Bruno formula; see e.g., Theorem 6.8 of \cite{schumann2019multivariate}]\label{prop:faa_di_bruno}
Let $g(x) = (g_1(x), \ldots, g_n(x)): \R^m \mapsto \R^n$ and $f(y): \R^n \mapsto \R$ be two functions that have $M$ derivatives over $\R^n$. Then 
\[D_x^{\alpha}(f\circ g)(x) = \sum_{\kappa\in \N^n, 1\le |\kappa|\le |\alpha|}D_y^\kappa f(g(x)) B_{\alpha, \kappa}((D_x^\gamma g_j(x))_{\gamma\in \N^n \setminus \{\mathbf{0}\}, 1\le j \le n}),\]
where $B_{\alpha, \kappa}$ is the multivariate Bell polynomial defined as 
\[
B_{\alpha,\kappa}\bigl((z_{\gamma,j})_{\gamma\in \N^n \setminus \{\mathbf{0}\},\;1\le j\le n}\bigr)
=
\sum_{\substack{(k_{\gamma,j}) \in \mathbb{N} \\
\sum_{\gamma} k_{\gamma,j} = \kappa_j\ \forall j \\
\sum_{\gamma,j} k_{\gamma,j}\,\gamma = \alpha}}
\frac{\alpha!}{\displaystyle\prod_{\gamma,j} k_{\gamma,j}!\,(\gamma!)^{k_{\gamma,j}}}
\prod_{\gamma,j} z_{\gamma,j}^{\,k_{\gamma,j}}.
\]
\end{lemma}

\begin{lemma}[Corollary 3.3.10 of \cite{grafakos2008classical}]\label{prop:fourier_bound} 
Let $f: \R^n\mapsto \R$ with $D^\alpha f\in L^1$ for any $\alpha\in \N^n$ with $|\alpha|\le M$. Then, for any $\upsilon\in \R^J$, 
\[|\hat{f}(\upsilon)|\le c_{n,M} \max\left\{\|f\|_{L^1}, \max_{|\alpha| = M}\|D^\alpha f\|_{L^1}\right\}(1 + \|\upsilon\|)^{-M},\]
where $c_{n,M}$ is a constant that only depends on $n$ and $M$.
\end{lemma}

\begin{lemma}\label{lem:symbol}
Let $q$ be a density that satisfies \cref{as:location-scale-assns}(2). For some
$\Sigma (x)$ taking values in $\S_+^J$, let $\cT$ be defined as in \eqref{eq:t_def}
and let $q_x(\zeta) = \det(\Sigma(x))^{-1} q\pr{\Sigma(x)^{-1}\zeta}$.
Then, $\cT$ is a pseudo-differential
operator with symbol $\hat{q}_x(\omega)$ in Kohn--Nirenberg quantization, i.e.,
\[(\cT g)(x) = \frac{1}{(2\pi)^J}\int_{\R^J}e^{i\langle x, \omega\rangle}\hat{q}_x(\omega)\hat{g}(\omega)d\omega.\]
\end{lemma}
\begin{proof}
By the assumptions on $q(\cdot)$, 
\begin{align*}
(\mathcal{T}g)(x) & = \int_{\R^J}q_x(x-\delta)g(\delta)d\delta\\
 & = \frac{1}{(2\pi)^J}\int_{\R^J}\int_{\R^J}e^{i\langle x - \delta, \omega\rangle}
 \hat{q}_x(\omega)g(\delta)d\omega d\delta\,\,\tag{Inverse Fourier transform}\\
 & = \frac{1}{(2\pi)^J}\int_{\R^J}\int_{\R^J}e^{i\langle x, \omega\rangle}\hat{q}_x(\omega)e^{-i\langle \delta, \omega\rangle}g(\delta)d\delta d\omega\\
  & = \frac{1}{(2\pi)^J}\int_{\R^J}\int_{\R^J}e^{i\langle x, \omega\rangle}\hat{q}_x(\omega)\hat{g}(\omega) d\omega.
\end{align*}
\end{proof}

\begin{lemma}\label{lem:k_b_inv}
Assume $M\ge 2s \ge 1$. If $k\in \K^{2s}$, then  and $k \circ b^{-1}\in \K^{2s}$. 
\end{lemma}

\begin{proof}
  By \eqref{eq:b_inv_norm}, $b^{-1}\in \K^M$. Since $M \ge 2s$, \cref{lem:composition_Ks} implies $k\circ b^{-1}\in \K^{2s}$.
\end{proof}

\begin{lemma}\label{lem:tilde_Sigma}
Under \cref{as:location-scale-assns}, 
\begin{align*}
&\|\tilde{\Sigma}(\cdot) - I\|_{\W^{M,\infty}} + \|\tilde{\Sigma}(\cdot) - I\|_{\W^{M,1}}\\
& \le B_M\Big(\|\Sigma(\cdot) - \Sigma_0\|_{\W^{M,\infty}} + \|\Sigma(\cdot) - \Sigma_0\|_{\W^{M,1}}, \|Db\|_{\W^{M-1,\infty}}, \|Db^{-1}\|_{L^\infty}, \lambda_{\min}^{-1}(\Sigma_0)\Big)
\end{align*}
for some function $B_M$ that only depends on $M$. In particular, 
\[\lim_{y\rightarrow 0}B_{M}\left(y, \|Db\|_{\W^{M-1,\infty}}, \|Db^{-1}\|_{L^\infty}, \lambda_{\min}^{-1}(\Sigma_0)\right) = 0.\]
\end{lemma}
\begin{proof}
First, we note that 
\begin{align*}
&\|\tilde{\Sigma}(\cdot) - I\|_{\W^{M,\infty}} + \|\tilde{\Sigma}(\cdot) - I\|_{\W^{M,1}}\\
& \le \lb\|\Sigma\circ b^{-1} - \Sigma_0\|_{\W^{M,\infty}} + \|\Sigma\circ b^{-1} - \Sigma_0\|_{\W^{M,1}}\rb \lambda_{\min}^{-1}(\Sigma_0)\\
& = \lb\|(\Sigma(\cdot) - \Sigma_0)\circ b^{-1}\|_{\W^{M,\infty}} + \|(\Sigma(\cdot) - \Sigma_0)\circ b^{-1}\|_{\W^{M,1}}\rb \lambda_{\min}^{-1}(\Sigma_0)
\end{align*}
By \cref{as:location-scale-assns}(3), $\Sigma(\cdot) - \Sigma_0\in \W^{M,\infty}$. By 
\cref{lem:f_inv_K}
and \cref{as:location-scale-assns}(4), 
\begin{equation}\label{eq:b_inv_norm}
b^{-1}\in \K^M\Longrightarrow \|Db^{-1}\|_{\W^{M-1,\infty}} < \infty.
\end{equation}
Furthermore, 
\begin{equation}\label{eq:b_infty}
b\in \K^M\Longrightarrow Db\in L^\infty.
\end{equation}
By \cref{lem:ug_M_infty}, \eqref{eq:b_inv_norm}, and \eqref{eq:b_infty},
\[\left\|(\Sigma(\cdot) - \Sigma_0)\circ b^{-1}\right\|_{\W^{M,\infty}}\le B_{M,\infty}(\|\Sigma(\cdot) - \Sigma_0\|_{\W^{M,\infty}}, \|D b^{-1}\|_{\W^{M-1,\infty}}, \|Db\|_{L^\infty}),\]
and 
\[\left\|(\Sigma(\cdot) - \Sigma_0)\circ b^{-1}\right\|_{\W^{M,1}}\le B_{M,1}(\|\Sigma(\cdot) - \Sigma_0\|_{\W^{M,1}}, \|D b^{-1}\|_{\W^{M-1,\infty}}, \|Db\|_{L^\infty}).\]
The proof is completed by defining 
\begin{align*}
& B_M\Big(\|\Sigma(\cdot) - \Sigma_0\|_{\W^{M,\infty}} + \|\Sigma(\cdot) - \Sigma_0\|_{\W^{M,1}}, \|Db\|_{\W^{M-1,\infty}}, \|Db^{-1}\|_{L^\infty}, \|\Sigma_0^{-1}\|_{\op}\Big)\\
& =  \bigg\{B_{M,\infty}(\|\Sigma(\cdot) - \Sigma_0\|_{\W^{M,\infty}}, \|D b^{-1}\|_{\W^{M-1,\infty}}, \|Db\|_{L^\infty})\\
& \qquad + B_{M,1}(\|\Sigma(\cdot) - \Sigma_0\|_{\W^{M,1}}, \|D b^{-1}\|_{\W^{M-1,\infty}}, \|Db\|_{L^\infty})\bigg\}\lambda_{\min}^{-1}(\Sigma_0).
\end{align*}
\end{proof}

\begin{lemma}\label{lem:f_inv_K}
Let $s\ge 1$ and let
$
f(x)=Ax+u(x)\in \K^s.
$
Assume 
\begin{equation}\label{eq:Df_inv}
\bigl\|(D f)^{-1}\bigr\|_{L^\infty}<\infty.
\end{equation}
Then $f^{-1}\in \K^s$. 
\end{lemma}

\begin{proof}

  We split the proof into a few steps. 

\smallskip\textbf{Step 1: the linear part $A$ is invertible}

\noindent To contradiction, if $A$ were not full rank, then there exists $0\neq e\in 
\mathrm{Ran} (A)^\perp$ (equivalently $A^T e=0$). For all $x\in \R^J$,
\[
e\cdot f(x)=e\cdot(Ax+u(x))=e\cdot u(x).
\]
Since $u\in \W^{s,\infty}\subset L^\infty$, the right-hand side is bounded; hence $e\cdot f(\R^J)$ is bounded. If $f$ is surjective, $f(\R^J)=\R^J$, but $e\cdot y$ is unbounded on $\R^J$, a contradiction. Thus $A$ is full rank.

\smallskip \textbf{Step 2: isolate the linear part of the inverse}

\noindent Let $g=f^{-1}$. Using $y=f(g(y))=Ag(y)+u(g(y))$ we obtain
\[
g(y)=A^{-1}y-A^{-1}u(g(y))=:A^{-1}y+v(y),
\qquad v(y):=-A^{-1}u(g(y)).
\]
Since $u\in L^\infty$, we have $v\in L^\infty$ and
\[
\norm{v}_{L^\infty}\le \norm{A^{-1}}\,\norm{u}_{L^\infty}.
\]
Hence $g\in \Lin\oplus L^\infty$. It remains to show $v\in \W^{s,\infty}$.

\smallskip\textbf{Step 3: prove the result for integer $s$}

\noindent Differentiating $f(g(y))=y$ gives
\begin{equation}\label{eq:DfDg}
Df(g(y))\,Dg(y)=I,
\end{equation}
Then
\[Dg(y)=\bigl(Df(g(y))\bigr)^{-1}.\]
The assumption $\|(Df)^{-1}\|_{L^\infty}<\infty$ implies $Dg\in L^\infty$. Since $g(y) = A^{-1}y + v(y)$, we have $Dv = Dg - A^{-1}\in L^{\infty}$. Thus, $v\in \W^{1,\infty}$. 

We now prove $v\in \W^{m,\infty}$ for all $1\le m\le s$ by induction. Suppose for some $2\le m\le s$ we prove that $v\in \W^{m-1,\infty}$. Then
\begin{equation}\label{eq:induction_hypothesis}
  D^j v\in L^\infty, \quad j = 1,\ldots, m-1.
  \end{equation}
Differentiating \eqref{eq:DfDg} $m$ times we obtain, by the Leibniz rule and the multivariate Fa\`{a} di Bruno formula, 
\[\{(Df)\circ g\}\cdot  D^m g = R_m((D^1 f)\circ g, \ldots, (D^m f)\circ g, D^1 g, \ldots, D^{m-1}g),\]
where $D^k g, (D^k f)\circ g\in (\R^J)^{\otimes k}$ is the $k$-th differentials and $R_m$ is a tensor in $(\R^J)^{\otimes m}$ for which each coordinate is a polynomial of the entries of $(D^1 f)\circ g, \ldots, (D^m f)\circ g, D^1 g, \ldots, D^{m-1}g$. Since $f\in \K^s$, $D^k\in L^\infty$. Thus, by the induction hypothesis \eqref{eq:induction_hypothesis}, $R_m \in L^{\infty}$. By \eqref{eq:Df_inv},
\[D^m g = (Df (g))^{-1}R_m \in L^{\infty}.\]
Since $Dv = Dg - A^{-1}$, $D^m v = D^m g\in L^{\infty}$. Thus, $v\in \W^{m,\infty}$. The induction argument then implies $v\in \W^{s,\infty}$.

\smallskip\textbf{Step 4: prove the result for non-integer $s$}

\noindent Let $s=m+\sigma$ with $m=\lfloor s\rfloor\in\mathbb N$ and $\sigma\in(0,1)$. For any $1\le k\le m$, $f\in \K^s$ implies
\[D^k f\in \W^{\sigma,\infty}.\]
In Step 3, we have proved that $g\in \W^{m,\infty}\in \W^{1,\infty}$ and thus $g$ is
Lipschitz. By \cref{lem:composition_C_sigma},
\[(D^k f)\circ g\in \W^{\sigma,\infty}, \quad k = 1, \ldots, m.\]
By \eqref{eq:Df_inv} and \cref{lem:inversion_C_sigma},
\[((D f)\circ g)^{-1}\in \W^{\sigma,\infty}.\]
Recall that each coordinate of $R_m$ is a polynomial of the entries of $(D^1 f)\circ g,
\ldots, (D^m f)\circ g, D^1 g, \ldots, D^{m-1}g$. Since $g\in \W^{m,\infty}$, $D^1 g, \ldots, D^{m-1}g$ are all Lipschitz and hence in $\W^{\sigma,\infty}$. Since each monomial is a product, \cref{lem:product_C_sigma} implies 
\[R_m \in \W^{\sigma,\infty}.\]
By \cref{lem:product_C_sigma} again, we conclude that
\[D^m v = D^m g - A^{-1}I(m = 1)\in \W^{\sigma,\infty}.\]
Therefore, $g\in \K^s$. 
\end{proof}

\begin{lemma}\label{lem:composition_Ks}
  If $f,g\in \K^s$ for some $s\ge 1$, then $f\circ g\in \K^s$.
\end{lemma}

\begin{proof}

  We split the proof into a few steps.
  
\smallskip \textbf{Step 1: reduction to the nonlinear part}

\noindent Let
\[
f(x)=Ax+u(x),\qquad g(x)=Bx+v(x),
\]
with $u,v\in \W^{s,\infty}$. Then
\[
(f\circ g)(x)=A(Bx+v(x))+u(Bx+v(x))=(AB)x+Av(x)+u\circ g(x).
\]
Since $Av\in \W^{s,\infty}$, it suffices to show $u\circ g\in \W^{s,\infty}$.

\smallskip \textbf{Step 2: proof for integer $s$}

\noindent Since $g\in \K^s$, $D^k g\in \W^{s-k, \infty}$. By the multivariate Fa\`{a} di Bruno formula, for each integer $k\le s$,
\begin{equation}\label{eq:Dk}
D^k (u\circ g) = S_k((D^1 u)\circ g, \ldots, (D^k u)\circ g, D^1 g, \ldots, D^k g),
\end{equation}
where $S_k$ is a tensor in $(\R^J)^{\otimes k}$ for which each coordinate is a polynomial of the entries of $(D^1 f)\circ g, \ldots, (D^k f)\circ g, D^1 g, \ldots, D^k g$. In particular, each monomial involves at least one coordinate of $D^j f\circ g$ for some $j$. 

Since $k\le s$,
\[(D^j u)\circ g\in L^\infty,\,\, D^j g\in L^\infty, \quad j=1,\ldots, k.\]
Thus,
\[D^k (u\circ g)\in L^\infty.\]
Since this holds for all $k\le s$, we conclude that $u\circ g\in \W^{s,\infty}$.

\smallskip \textbf{Step 3: proof for non-integer $s$}

\noindent Let $s=m+\sigma$ with $m=\lfloor s\rfloor\in\mathbb N$ and $\sigma\in(0,1)$. In Step 2, we have already proved that
\[D^m (u\circ g) = S_m((D^1 u)\circ g, \ldots, (D^m u)\circ g, D^1 g, \ldots, D^m g).\]
Since $u, g\in \W^{m+\sigma,\infty}$,
\[(D^1 u)\circ g, \ldots, (D^m u)\circ g, D^1 g, \ldots, D^m g\in \W^{\sigma,\infty}.\]
By \cref{lem:product_C_sigma}, $D^m (u\circ g)\in \W^{\sigma,\infty}$. This implies $u\circ g\in \W^{s,\infty}$.
\end{proof}

\begin{lemma}\label{lem:ug_M_infty}
If $u\in \W^{M,\infty}$ for some integer $M$ and $g\in \K^M$ and $g$ is invertible. Then 
\[\|u\circ g\|_{M, \infty}\le B_{M}(\|u\|_{M,\infty}, \|Dg\|_{M-1, \infty}, \|Dg^{-1}\|_
{L^\infty}),\]
for some function $B_{M}$ that depends on $M$.
\end{lemma}
\begin{proof}
Since $g\in \K^M$, $Dg\in \W^{M-1,\infty}$. By definition, 
\[\|u\circ g\|_{\W^{M,\infty}} = \sum_{\alpha:|\alpha|\le M}\|D^\alpha (u\circ g)\|_{L^\infty}.\]
Since $g\in \K^M$ and $g$ is invertible, \cref{lem:f_inv_K} implies $g^{-1}\in \K^M\in
L^\infty$. 
In addition,
\[\|(D^j u)\circ g\|_{L^\infty}\le \|D^j u\|_{L^\infty}.\]
For any $k\le M$, each monomial in $S_k$ defined by \eqref{eq:Dk} involves at least one coordinate of $(D^j u)\circ g$. Note that $a\in L^\infty, b\in \W^{1,\infty}$ imply $\|ab\|_{L^\infty}\le \|a\|_{L^\infty}\|b\|_{L^\infty}$. Letting $a = D^j u, b = g$, each monomial is in $L^\infty$. Therefore, we can find $\bar{S}_k$ such that 
\[\left|S_k((D^1 u)\circ g, \ldots, (D^k u)\circ g, D^1 g, \ldots, D^k g)\right|\le \bar{S}_k(\|u\|_{M,\infty}, \|Dg\|_{\W^{M-1,\infty}}, \|Dg^{-1}\|_{L^\infty}).\]
In particular, $\bar{S}_k\rightarrow 0$ when $\|u\|_{M,\infty}\rightarrow 0$. Thus, 
\begin{align*}
\|u\circ g\|_{\W^{M,\infty}}&\le M^J \max_{0\le k\le M}\bar{S}_k(\|u\|_{\W^{M,\infty}}, \|Dg\|_{\W^{M-1,\infty}})\\
& := B_M(\|u\|_{\W^{M,\infty}}, \|Dg\|_{\W^{M-1,\infty}}, \|Dg^{-1}\|_{L^\infty}).
\end{align*}
\end{proof}

\begin{lemma}[Composition with a Lipschitz map preserves $\W^{\sigma,\infty}$ regularity]
\label{lem:composition_C_sigma}
If $h\in \W^{\sigma,\infty}$ for some $\sigma\in [0, 1)$ and $\phi:\R^J\to\R^J$ is Lipschitz with Lipschitz constant $\Lip(\phi)$, then $h\circ\phi\in \W^{\sigma,\infty}$ and
\[
[h\circ\phi]_{\W^{\sigma,\infty}}
\le
[h]_{\W^{\sigma,\infty}}\,\Lip(\phi)^\sigma,
\]
\end{lemma}

\begin{proof}
For $x\neq y$,
\[
\|h(\phi(x))-h(\phi(y))\|
\le [h]_{\W^{\sigma,\infty}}\,\|\phi(x)-\phi(y)\|^\sigma
\le [h]_{\W^{\sigma,\infty}}\,\Lip(\phi)^\sigma\,\|x-y\|^\sigma.
\]
Taking the supremum over $x\neq y$ gives the claim.
\end{proof}

\begin{lemma}[Inversion preserves $\W^{\sigma,\infty}$ regularity]
\label{lem:inversion_C_sigma}
Let $M:\R^J\to \S_+^J$ be such that $M\in \W^{\sigma,\infty}$ for some $\sigma\in [0,1)$ and $\|M^{-1}\|_{L^\infty}\le C$. Then $M^{-1}\in \W^{\sigma,\infty}$ and
\[
[M^{-1}]_{\W^{\sigma,\infty}}\le C^2\,[M]_{\W^{\sigma,\infty}}.
\]
\end{lemma}

\begin{proof}
Use the identity
\[
M^{-1}(x)-M^{-1}(y)= M^{-1}(x)\bigl(M(y)-M(x)\bigr)M^{-1}(y).
\]
Taking norms, dividing by $\|x-y\|^\sigma$, and taking suprema yields the estimate.
\end{proof}

\begin{lemma}[$\W^{\sigma,\infty}$ is closed in products]\label{lem:product_C_sigma}
If $a,b\in \W^{\sigma,\infty}\cap L^\infty(\R^J)$ for some $\sigma\in [0, 1)$, then $ab\in \W^{\sigma,\infty}$ and
\[
[ab]_{\W^{\sigma,\infty}}
\le
\|a\|_{L^\infty}\,[b]_{\W^{\sigma,\infty}}
+\|b\|_{L^\infty}\,[a]_{\W^{\sigma,\infty}}.
\]
\end{lemma}

\begin{proof}
Write
\[
a(x)b(x)-a(y)b(y)=a(x)(b(x)-b(y)) + b(y)(a(x)-a(y)).
\]
Divide by $\|x-y\|^\sigma$ and take suprema over $x\neq y$.
\end{proof}

\subsection{Proof of the claim in \Cref{rmk:necessity}}\label{subsec:proof_necessity}
Consider any function $g$ with $\E[g(\delta, P)\mid \delta, \Lambda] = 0$ a.s. \cref{as:p-index_x} implies $(X, Z) = (X, \lambda^{-1}(X, \Lambda))$. Thus, $(X, Z)$ is measurable with respect to $(X, \Lambda)$. By the tower property, 
\[\E[g(\delta, P)\mid X,Z] = \E[\E[g(\delta, P)\mid X,\Lambda]\mid X,Z] = \E[\E[g(\delta, P)\mid \delta,X,\Lambda]\mid X,Z].\]
By \cref{as:p-index_x} again, 
\[\E[g(\delta, P)\mid \delta,X,\Lambda] = \E[g(\delta, P)\mid \delta,\Lambda]=0, \,\, \mathrm{a.s.}\]
This implies $\E[g(\delta, P)\mid X,Z] = 0$. Since the RHS does not depend on $Z$, faithfulness implies $g(\delta, P) = g(\delta)$. Then 
\[0 = \E[g(\delta, P)\mid \delta, \Lambda] = g(\delta) = g(\delta, P), \,\, \mathrm{a.s.}\]

\subsection{Proof of \Cref{thm:unrestricted-faithfulness}}\label{subsec:proof_unrestricted-faithfulness}

Through this section, we consider any function $H$ with
\begin{equation}
\label{eq:premise}
\E[H(\delta,P)\mid X,Z]=k(X)\quad\text{a.s.}, \quad \E[|H(\delta, P)|] < \infty,
\end{equation}
for some measurable $k:\R^J\to\R$. Define
\begin{equation}
\label{eq:rf}
r_H(\lambda,d):=\E[H(\delta,P)\mid \Lambda=\lambda,\delta=d].
\end{equation}

We start by proving a useful high-level result.

Fix any $\lambda_0\in S_\lambda$ for which \eqref{eq:r_H_lambda} in the proof of \Cref{prop:p-index_x} in  \Cref{subsec:model-p-lambda} holds. 
Define
\begin{equation}\label{eq:h_lambda}
  H_\lambda(d) = r_H(d,\lambda) - r_H(d,\lambda_0),
\end{equation}
and
\[\cH_\lambda = \{H_\lambda(d): \E[|H(\delta,P)|] < \infty\}.\]
Note that $\cH_\lambda$ is a set of functions of $\delta$.

\begin{lemma}\label{lem:high-level}
  Assume that, for Lebesgue-almost every $\lambda\in \R^{k}$,
  \begin{equation}\label{eq:key_assumption}
    \cH_\lambda\cap \Ker(\cT) = \{\mathbf{0}\},
  \end{equation}
  where $\mathbf{0}$ denotes the class of functions of $\delta$ that are Lebesgue-almost everywhere $0$. Under \Cref{as:exogenous,as:p-index_x,,ass:location_model,ass:P,ass:complete}, $(\delta,P)\mid (X,Z)$ is faithful.
\end{lemma}

\begin{rmk}
  The condition \eqref{eq:key_assumption} reveals the minimal condition for faithfulness under the $\lambda$-index assumption. In particular, \eqref{eq:key_assumption} becomes trivial under completeness of $(\delta,P)\mid X,Z$ (\cref{as:complete_x}), which implies $\Ker(\cT) = \{\mathbf{0}\}$.
\end{rmk}

\begin{proof}[\textbf{Proof of \Cref{lem:high-level}}]
  By \eqref{eq:r_H_lambda}, there exists a full-measure subset $\cS \subset S_\lambda$ such that \eqref{eq:r_H_lambda} holds on $\cS$. Fix any $\lambda\in \cS$. Then we have
  \begin{equation}\label{eq:H_lambda_cond_exp_zero}
    \E[H_\lambda(\delta)\mid X] = \E[r_H(\delta,\lambda) - r_H(\delta,\lambda_0)\mid X] = 0\quad\text{a.s.}
    \end{equation}
  This implies
  \[r_H(\delta,\lambda) - r_H(\delta,\lambda_0)\in \Ker(\cT).\]
  By \eqref{eq:key_assumption},
  \[r_H(\delta,\lambda) - r_H(\delta,\lambda_0) = 0, \quad\text{a.s.}\]
  This can be written as 
  \[\PP(r_H(\delta,\Lambda) \neq r_H(\delta,\lambda_0)\mid \Lambda = \lambda) = 0, \quad\text{a.s.}.\]
  As a result,
  \[\PP(r_H(\delta,\Lambda) \neq r_H(\delta,\lambda_0)) = 0\Longleftrightarrow r_H(\delta,\Lambda) = r_H(\delta,\lambda_0)\quad \text{a.s.}\]
  Write $h_0(\delta)$ for $r_H(\delta,\lambda_0)$ to highlight it is just a function of $\delta$. By definition of $r_H$ in \eqref{eq:rf},
  \[\E[H(\delta,P)\mid \Lambda,\delta] = r_H(\delta,\Lambda) = h_0(\delta) \quad\text{a.s.}\]
  This implies
  \[\E[H(\delta,P)-h_0(\delta)\mid \Lambda,\delta]= 0 \quad\text{a.s.}\]
  By \cref{ass:complete}, we conclude that
  \[H(\delta,P) = h_0(\delta)\quad \text{a.s.}.\]
  The proof is then completed.
\end{proof}

~\\
\noindent \textbf{Proof of \Cref{thm:unrestricted-faithfulness}}. Fix $H$ with $\E[|H(\delta, P)|] < \infty$ and assume
\begin{equation}\label{eq:premise}
\E[H(\delta,P)\mid X,\Lambda]=k(X)\quad\text{a.s.}
\end{equation}
Let $r_H$ and $H_\lambda$ be defined in \eqref{eq:rf} and \eqref{eq:h_lambda}, respectively. Define the class of exponentially integrable functions as 
\begin{equation}\label{eq:S_a0}
  \cS_{a_0} = \left\{\int_{\R^J} e^{a_0\norm{d}}\bigl|H_\lambda(d)\bigr|\mathrm{d}d<\infty\right\}.
\end{equation}

\noindent \textbf{Step 1}. By definition and \cref{ass:P} (2),
\begin{align*}
|H_\lambda(d)|
& =\left|\int H(d,p)\bigl(q(p\mid\lambda,d)-q(p\mid\lambda_0,d)\bigr)\,dp\right|\\
& \le \int |H(d,p)|\,|q(p\mid\lambda,d)-q(p\mid\lambda_0,d)|\,dp\\
& \le \zeta(d)\int |H(d,p)|\mu(p)\,dp.
\end{align*}
By \cref{ass:P} (1), for any $\lambda'\in\R^k$,
\[
\int |H(d,p)|\mu(p)\,dp
\le \frac{1}{c_-}\int |H(d,p)|\,q(p\mid\lambda',d)\,dp
= \frac{1}{c_-}\E\bigl[|H(\delta,P)|\mid \Lambda=\lambda',\delta=d\bigr].
\]
Averaging over $\lambda'\sim \Lambda\mid \delta=d$ yields
\[
\int |H(d,p)|\mu(p)\,dp \le \frac{1}{c_-}\E\bigl[|H(\delta,P)|\mid \delta=d\bigr].
\]
Therefore,
\[
|H_\lambda(d)|\le \frac{1}{c_-}\zeta(d)\E\bigl[|H(\delta,P)|\mid \delta=d\bigr].
\]
Multiplying by $e^{a_0\norm{d}}$, integrating over $d$, inserting $p_\delta(d)$, and using \cref{ass:P} (3) gives
\begin{align*}
& \int e^{a_0\norm{d}}|H_\lambda(d)|\mathrm{d}d\\
& \le \frac{1}{c_-}\int \frac{e^{a_0\norm{d}}\zeta(d)}{p_\delta(d)}\,\E\bigl[|H(\delta,P)|\mid \delta=d\bigr]\ p_\delta(d)\mathrm{d}d\\
& \le \frac{1}{c_-}\sup_{d\in \R^J}\frac{e^{a_0\norm{d}}\zeta(d)}{p_\delta(d)}\E|H(\delta,P)|<\infty.
\end{align*}
As a result,
\begin{equation}\label{eq:H_lambda_S_a0}
  H_\lambda \in \cS_{a_0}.
  \end{equation}

\noindent \textbf{Step 2}. By \cref{ass:location_model} and \eqref{eq:H_lambda_cond_exp_zero},
\[
0 = \E[H_\lambda(\delta)\mid X=x]=\E[H_\lambda(a(x)+\epsilon)]=(H_\lambda*p_\epsilon)(a(x)),
\]
where $*$ denotes convolution. Since $\check{X}=a(X)$ has a density positive everywhere, it follows that
\begin{equation}\label{eq:convolution-zero}
(H_\lambda*p_\epsilon)(\check{x})=0\quad\text{for Lebesgue-almost every }\check{x}\in\R^J.
\end{equation}
Because $H_\lambda\in L^1(\R^J)$, implied by \eqref{eq:H_lambda_S_a0}, and $p_\epsilon\in L^1(\R^J)$, we may take Fourier transforms:
\[
\widehat{H_\lambda*p_\epsilon}(\omega)=\widehat H_\lambda(\omega)\,\widehat{p_\epsilon}(\omega)
=\widehat H_\lambda(\omega)\,\varphi_\epsilon(\omega).
\]
From \eqref{eq:convolution-zero}, $\widehat{H_\lambda*p_\epsilon}(\omega)\equiv 0$, hence
\begin{equation}\label{eq:fourier-product}
\widehat H_\lambda(\omega)\,\varphi_\epsilon(\omega)=0\qquad\forall\omega\in\R^J.
\end{equation}

\noindent \textbf{Step 3}. We now prove a lemma that $H_\lambda$ is real-analytic. The proof is deferred to the end of the proof.

\begin{lemma}\label{lem:real-analytic}
Assume that for some $a>0$,
\begin{equation}\label{eq:exp-moment}
\int_{\R^J} e^{a\norm{d}}|h(d)|\mathrm{d}d<\infty.
\end{equation}
Define $\widehat h(\omega)=\int_{\R^J}e^{-i\langle \omega, d\rangle}h(d)\mathrm{d}d$ for $\omega\in\R^J$.
Then $\widehat h$ is real-analytic on $\R^J$. 
\end{lemma}

\noindent \textbf{Step 4.} From \eqref{eq:fourier-product} and continuity of $\varphi_\epsilon$ with $\varphi_\epsilon(0)=1$, there exists $\varepsilon>0$ such that
\[\varphi_\epsilon(\omega)\neq 0,\quad \forall \norm{\omega}<\varepsilon.\]
Therefore $\widehat H_\lambda(\omega)=0$ for all $\norm{\omega}<\varepsilon$.
Since $\widehat H_\lambda$ is real-analytic, by the identity theorem, vanishing on the nonempty open ball $\{\norm{\omega}<\varepsilon\}$ implies
\[\widehat H_\lambda(\omega)\equiv 0 \text{ for any } \omega\in\R^J. \]
Since $H_\lambda\in \cS_{a_0}\subset L^1(\R^J)$, injectivity of the Fourier transform on $L^1$ yields $H_\lambda(\omega)=0$ for Lebesgue-almost every $\lambda$. Thus, the condition \eqref{eq:key_assumption} is proved. By \Cref{lem:high-level}, the proof of faithfulness is completed.

~\\
\noindent \textbf{Proof of incompleteness.} If there exists $\omega_0\in \R^J$ such that $\varphi_\epsilon(\omega_0) = 0$,
\[\E[\exp\{i\langle \omega_*, \delta\rangle\}\mid X] = \exp\{i\langle \omega_*, a(X)\rangle\}\E[\exp\{i\langle \omega_*, \epsilon\rangle\}] = 0.\]
Thus,
$\E[\cos(\langle \omega_*, \delta\rangle )\mid X] = 0.$
As a result, $\delta\mid X$ is incomplete.

~\\
\noindent \textbf{Proof of \Cref{lem:real-analytic}.} First, we prove that, for any $b<a$ and any integer $m\ge 0$,
\begin{equation}\label{eq:moments}
\int_{\mathbb{R}^J} \|d\|^m e^{b\|d\|}\,|h(d)|\,\mathrm{d}d
\;\le\; \frac{m!}{(a-b)^m}\,M_a.
\end{equation}
In fact, using the elementary inequality $t^m\le m!e^t$ with $t=b-a$ (since $e^t=\sum_{k\ge 0} t^k/k!\ge t^m/m!$), we obtain that 
\[
\|d\|^m \le \frac{m!}{(a-b)^m}\,e^{(a-b)\|d\|}.
\]
Multiplying by $e^{b\|d\|}|h(d)|$ and integrating yields \eqref{eq:moments}.

Let $z=\xi+i\chi$ with $\|\chi\|<a$ and
\begin{equation}\label{eq:GLaplace}
\tilde{h}(z) := \int_{\mathbb{R}^J} h(d)\,e^{-i \langle d, z\rangle}\,\mathrm{d}d
\;=\;\int_{\mathbb{R}^J} h(d)\,e^{-i \langle d, \xi\rangle}\,e^{\langle \chi, d\rangle}\,\mathrm{d}d,
\qquad z=\xi+i\chi.
\end{equation}
Since $\langle \chi, d\rangle \le \|\chi\|\,\|d\|$,
\[
|h(d)e^{-i \langle d, z\rangle}|=|h(d)| e^{\langle\chi, d\rangle}\le |h(d)|e^{\|\chi\|\|d\|} \le |h(d)| e^{a\|d\|}.
\]
The RHS is integrable by \eqref{eq:exp-moment}, so \eqref{eq:GLaplace} is absolutely convergent. 

Fix $z_0=\xi_0+i\chi_0\in T_a$ and set $\delta:=a-\|\chi_0\|>0$.
For $w\in\mathbb{C}^J$,
\[
\tilde{h}(z_0+w)=\int_{\mathbb{R}^J} h(d) e^{-i \langle d, z_0\rangle}\,e^{-i \langle d, w\rangle}\,\mathrm{d}d.
\]
Use the scalar power series $e^{-i \langle d, w\rangle}=\sum_{m=0}^\infty \frac{(-i)^m}{m!}(\langle d, w\rangle)^m$.
Fix $r\in(0,\delta)$ and assume $\|w\|\le r$. Then $|\langle d, w\rangle|^m\le (\|d\|\cdot\|w\|)^m\le (\|d\|r)^m$, hence
\begin{align*}
&\sum_{m=0}^\infty \frac{1}{m!}\int_{\mathbb{R}^J} |h(d)|\,|e^{-i \langle d, z_0\rangle}|\,|\langle d, w\rangle|^m\,\mathrm{d}d\\
&\le
\sum_{m=0}^\infty \frac{r^m}{m!}\int_{\mathbb{R}^J} |h(d)|\,e^{\langle \chi_0, d\rangle}\,\|d\|^m\,\mathrm{d}d \\
&\le
\sum_{m=0}^\infty \frac{r^m}{m!}\int_{\mathbb{R}^J} |h(d)|\,e^{\|\chi_0\|\,\|d\|}\,\|d\|^m\,\mathrm{d}d.
\end{align*}
Applying \eqref{eq:moments} with $b=\|\chi_0\|$ to bound the inner integral by $\frac{m!}{\delta^m}M_a$, we have
\[
\sum_{m=0}^\infty \frac{r^m}{m!}\int_{\mathbb{R}^J} |h(d)|\,e^{\|\chi_0\|\,\|d\|}\,\|d\|^m\,\mathrm{d}d
\le
M_a\sum_{m=0}^\infty (r/\delta)^m
<\infty,
\]
so the integrated series is absolutely summable and the convergence is uniform over $\|w\|\le r$.
By dominated convergence, we have
\begin{equation}\label{eq:taylor}
\tilde{h}(z_0+w)
=\sum_{m=0}^\infty \frac{(-i)^m}{m!}\int_{\mathbb{R}^J} h(d) e^{-i \langle d, z_0\rangle}\,(\langle d, w\rangle)^m\mathrm{d}d
\end{equation}
For each $m$, the map $w\mapsto \int h(d)e^{-i \langle d, z_0\rangle}(\langle d, w\rangle)^m\,\mathrm{d}d$ is a polynomial in $w$. Since the series \eqref{eq:taylor} converges uniformly on compact subsets $\{w:\|w\|\le r\}$ for every $r<\delta$, $\tilde{h}$ is holomorphic in a neighborhood of $z_0$. Because $z_0\in T_a$ was arbitrary, $\tilde{h}$ is holomorphic on $T_a$. By definition, for $z\in \R^J$, 
\[\tilde{h}(z) = \hat{h}(z).\]
We can then conclude that $\hat{h}$ is real-analytic.

\subsection{Proof of \cref{prop:example}} \label{subsec:proof_example}
To verify the incompleteness of $\delta\mid X$, we simply prove $\varphi_\epsilon$ has zeros. Since $\nu, \eta_1, \ldots, \eta_J$ are independent, 
\[\varphi_\epsilon(\omega) = \varphi_\nu(\omega)\prod_{j=1}^J\varphi_{\eta_j}(\omega_j) = \varphi_\nu(\omega)\prod_{j=1}^J \frac{\sin(b_j\omega)}{b_j\omega}.\]
Thus,
  \[\left\{\omega: \varphi_\epsilon(\omega) = 0\right\} = \left\{\omega: \omega_j = \frac{k\pi}{b_j}\text{ for some }j\in [J]\text{ and }k\in \mathbb{Z}\right\}.\]

Next, we verify \cref{ass:P}. Since $S(p)$ and $t(\lambda)$ are bounded,
\begin{equation}\label{eq:bounded_q}
  C_1^{-1}\mu(p)\ \le\ q(p\mid \lambda,d)\ \le\ C_1\mu(p),
  \end{equation}
  for some constant $C_1>0$. Thus, \cref{ass:P} (1) is proved.

  Write
  \[q_\theta(p)=\exp(\langle \theta, S(p)\rangle-A(\theta))\mu(p).\]
For each $i$,
\[
\partial_{\theta_i}q_\theta(p)
=\big(S_i(p)-\partial_{\theta_i}A(\theta)\big)\,q_\theta(p)
=\big(S_i(p)-\E_\theta[S_i(P)]\big)\,q_\theta(p),
\]
where the last line uses the well-known property of exponential families $\E[S_i(P)] = \partial_{\theta_i}A(\theta)$. Thus, 
\[|\partial_{\theta_i}q_\theta(p)|\le C_2 q_\theta(p),\]
for some constant $C_2 > 0$. By \eqref{eq:bounded_q},
\[|\partial_{\theta_i}q_\theta(p)|\le C_1C_2 \mu(p).\]
Then
\[|q_\theta(p)-q_{\theta_0}(p)|\le C_1C_2\|\theta-\theta_0\|_1\,\mu(p).\]
By definition,
\[\theta(\lambda,d)-\theta(\lambda_0,d)=\zeta(d)\big(t(\lambda)-t(\lambda_0)\big).\]
Thus, 
\begin{equation}\label{eq:q_lambda}
|q(p\mid \lambda,d)-q(p\mid \lambda_0,d)| \le C_3 \zeta(d) \mu(p),
\end{equation}
where $C_3 = 2C_1 C_2 \sup_{\lambda}|t(\lambda)|$. As a result,
\begin{align*}
|H_\lambda(d)|
&=\left|\int H(d,p)\big(q(p\mid \lambda,d)-q(p\mid \lambda_0,d)\big)\,dp\right|\\
&\le \int |H(d,p)| \cdot |q(p\mid \lambda,d)-q(p\mid \lambda_0,d)|\,dp\\
&\le C_3 \zeta(d) \int |H(d,p)|\,\mu(p)\,dp\\
& \le C_1 C_3 \zeta(d) \int |H(d,p)| q(p\mid d, \lambda')dp,
\end{align*}
for any $\lambda' \in \R^J$ that may differ from $\lambda$, where the last line follows from \eqref{eq:bounded_q}. As a result,
\[|H_\lambda(d)|\le C_1 C_3 \zeta(d) \E[|H(\delta, P)|\mid \Lambda=\lambda', \delta = d].\]
Taking integral of $\lambda'$ over the distribution of $\Lambda\mid \delta = d$, we obtain that
\[|H_\lambda(d)|\le C_1 C_3 \zeta(d) \E[|H(\delta, P)|\mid \delta = d].\]
Let $p_{\delta}(d)$ be the marginal density of $\delta$ and $a_0$ be any constant less than $c - 1/4$. Then
\begin{align}
  &\int_{\R^J} e^{a_0\|d\|}|H_\lambda(d)|\mathrm{d}d \nonumber \\
  & \le C_1 C_3 \int_{\R^J} \frac{e^{a_0\|d\|}\zeta(d)}{p_{\delta}(d)} \cdot \E[|H(\delta, P)|\mid \delta = d] p_{\delta}(d) \mathrm{d}d\nonumber \\
  & \le C_1 C_3 \E[|H(\delta, P)|] \cdot \sup_{d\in \R^J}\frac{e^{a_0\|d\|}\zeta(d)}{p_{\delta}(d)}. \label{eq:h_lambda_L1}
\end{align}
Under condition (b), $\delta=(X+\nu)+\eta$ where $X+\nu\sim N(0,2I_J)$ and $\eta\in[-1,1]^J$.
Thus,
\[p_\delta(d)=\E_{\eta}[\phi_J ((d-\eta)/2)],\]
where $\phi_J(x)$ denote the $N(0,I_J)$ density. Note that $p_\delta(d)$ is strictly positive and continuous everywhere.
Moreover $p_\delta(d)$ has Gaussian tails $\asymp \exp(-\|d\|^2/4)$, so the ratio
\[e^{a_0\|d\|}\zeta(d)/p_\delta(d) \rightarrow 0, \quad \|d\|\to\infty\]
and is continuous on $\mathbb{R}^J$. Therefore it has a finite supremum. By \eqref{eq:h_lambda_L1}, the proof os \cref{ass:P} is then completed.

Lastly, condition (a)  implies that the law of $P\mid \Lambda, \delta=d$ is an exponential family with sufficient statistic $S(p)$ given any $d$. Since $S$ is continuous and non-constant, the range of $S$ contains an open set. By the Lehmann-Scheff\`{e} completeness theorem, \cref{ass:complete} holds.

\subsection{Proof of \Cref{prop:comp_not_faith}}\label{subsec:proof_comp_not_faith}

This choice is such that \begin{align*}
\E\bk{
    \one(P > 0) \mid \delta,Z, X
  } &= \P_{Q\sim \Norm(0,1)}\pr{
    r(\delta, X) \tau(\delta, Z, X) + \tau(\delta, Z, X) Q > 0
  } 
  \\&=\Phi\pr{
    r(\delta, X)
  }
\end{align*}
Thus \[
  \E[\one(P > 0) \mid Z, X] = \E[\Phi\pr{
    r(\delta, X)
  } \mid X] \equiv k(X).
\]
Hence faithfulness does not hold. \Cref{as:exogenous} holds by construction. Thus the
remainder of the proof verifies that \cref{as:completeness} holds.

The density of $P$ is
\begin{align*}
  f(p \mid x,\delta,z) &= \frac{1}{\sqrt{2\pi}} \exp\pr{\log\frac{1}{\tau(\delta, Z, X)} - 
\frac{(p - r \tau)^2}{2\tau^2}
  } \\
  &= \frac{1}{\sqrt{2\pi}} \exp\pr{
    \log(1/\tau) - \frac{p^2}{2\tau^2} - \frac{r^2}{2} + \frac{yr}{\tau}
  }
\end{align*}
Under our choices for $r, \tau$, \[
  f(p \mid x,\delta,z) = \frac{1}{\sqrt{2\pi} z} \exp\pr{
    -\frac{p^2}{2z^2} - \frac{r^2}{2} + p \delta\frac{x}{z}
  }
\]
Thus \[
  f(\delta, p \mid x, z) = \frac{1}{\sqrt{2\pi} z} a(x) \one(x \in [1,2]) \exp\pr{
    -p^2\frac{1}{2z^2} + p \delta \frac{x}{z}
  }. 
\]

It is an exponential family supported on $[1,2] \times \R$ with natural parameters \[
  \eta(x,z) = \pr{-\frac{1}{2z^2}, \frac{x}{z}}
\]
It is possible to choose $X,Z$ such that  the support of $\eta(X,Z)$ contains an open set of $(-\infty, 0) \times
(0,
\infty)$. The
sufficient statistics are \[
  (T_1, T_2) = \pr{
    P^2, \delta P
  }
\]
Note that \[
  \delta = |T_2/\sqrt{T_1}|, P = \sgn(T_2) \sqrt{T_1}
\]
so that $(\delta, P)$ and $(T_1, T_2)$ are bijective. Thus, take any function $h(\delta, P)$, we can
write it in terms of the sufficient statistics $h(T_1, T_2)$. Since the support of $\eta
(x, z)$ contains an open set in $(-\infty, 0) \times (0,
\infty)$, $T_1, T_2$ are complete sufficient
statistics. Thus \[
  \E[H(\delta,P) \mid Z,X] = 0 \implies \E[\tilde H(T_1, T_2) \mid \eta] = 0 \implies
  \tilde H = H = 0.
\]

\end{appendix}
\end{document}